\documentclass[12pt]{amsart}
\usepackage{amssymb,amscd}

\usepackage[colorlinks=true,allcolors = blue]{hyperref}

\let\frak\mathfrak
\let\Bbb\mathbb

\def\>{\relax\ifmmode\mskip.666667\thinmuskip\relax\else\kern.111111em\fi}
\def\<{\relax\ifmmode\mskip-.333333\thinmuskip\relax\else\kern-.0555556em\fi}
\def\vsk#1>{\vskip#1\baselineskip}
\def\vv#1>{\vadjust{\vsk#1>}\ignorespaces}
\def\vvn#1>{\vadjust{\nobreak\vsk#1>\nobreak}\ignorespaces}

  \let\ssize\scriptstyle
\let\sssize\scriptscriptstyle

\let\Medskip\medskip
\def\medskip{\par\Medskip}
\let\Bigskip\bigskip
\def\bigskip{\par\Bigskip}

\let\Maketitle\maketitle
\def\maketitle{\Maketitle\thispagestyle{empty}\let\maketitle\empty}

\newtheorem{thm}{Theorem}[section]
\newtheorem{cor}[thm]{Corollary}
\newtheorem{lem}[thm]{Lemma}

\numberwithin{equation}{section}

\theoremstyle{definition}
\newtheorem*{rem}{Remark}

\let\mc\mathcal
\let\nc\newcommand

\let\eps\varepsilon

\let\ka\kappa
\let\la\lambda

\let\phi\varphi

\let\Si\Sigma

\let\Om\Omega

\let\der\partial
\let\Hat\widehat

\let\geq\geqslant

\let\leq\leqslant

\let\on\operatorname
\let\bi\bibitem
\let\bs\boldsymbol

\def\C{{\mathbb C}}
\def\Z{{\mathbb Z}}

\def\F{{\mc F}}

\def\+#1{^{\{#1\}}}

\def\diag{\on{diag}}

\def\GR{{\on{Gr}_0(H)}}
\def\Gr{{\on{Gr}_{0}(H)}}

\def\tr{\on{tr}}
\def\Wr{\on{Wr}}

\def\sln{\mathfrak{sl}_N}

\def\beq{\begin{equation}}
\def\eeq{\end{equation}}
\def\be{\begin{equation*}}
\def\ee{\end{equation*}}

\nc{\bea}{\begin{eqnarray*}}
\nc{\eea}{\end{eqnarray*}}
\nc{\bean}{\begin{eqnarray}}
\nc{\eean}{\end{eqnarray}}
\nc{\Ref}[1]{{\rm(\ref{#1})}}

\let\ga\gamma
\let\Ga\Gamma

\nc{\Il}{{\mc I_{\bs\la}}}
\nc{\bla}{{\bs\la}}
\nc{\Fla}{\F_\bla}
\nc{\tfl}{{T^*\Fla}}
\nc{\GL}{{GL_n(\C)}}
\nc{\GLC}{{GL_n(\C)\times\C^*}}

\let\sd s 

\def\Wh{\Hat W}

\def\ddk_#1{\kk_{#1}\<\>\frac\der{\der\<\>\kk_{#1}}}

\def\bul{\mathbin{\raise.2ex\hbox{$\sssize\bullet$}}}
\def\intt{\mathchoice
{\mathop{\raise.2ex\rlap{$\,\,\ssize\backslash$}{\intop}}\nolimits}
{\mathop{\raise.3ex\rlap{$\,\sssize\backslash$}{\intop}}\nolimits}
{\mathop{\raise.1ex\rlap{$\sssize\>\backslash$}{\intop}}\nolimits}
{\mathop{\rlap{$\sssize\<\>\backslash$}{\intop}}\nolimits}}

\def\LL{\mathcal L}

\let\kk q 
\let\cc c

\let\Ko K

\def\GZ/{Gelfand-Zetlin}
\def\KZ/{{\slshape KZ\/}}
\def\qKZ/{{\slshape qKZ\/}}
\def\XXX/{{\slshape XXX\/}}

\nc{\slnl}{{\sln (\lambda)}}
\nc{\PCN}{{   (\on{P}(\C[x]))^N   }}
\nc{\di}{\text{Diag}}
\nc{\dio}{\text{Diag}_0}
\nc{\Mm}{{\mc M}}
\nc{\Nn}{{\mc N}}

\nc{\A}{{\mathbb A}}

\nc{\PCr}{{  P  (\C[x])^n   }}

\nc{\Pk}{{(\bs{P}^1)^k}}

\nc{\N}{{\Bbb N}}

\def\D{{\mc D}}

\nc{\Ll}{{\mc L}}

\nc{\ord}{{\text{ord}\,}}
\nc{\GM}{{\on{Gr}_{mKdV}}}

\newcommand\res{{\rm res}}
\newcommand\p{\partial}

\newenvironment{aside}{\begin{quote}\sffamily}{\end{quote}}

\nc\W{{\tilde{\on{W}}}}

\begin{document}

\hrule width0pt
\vsk->

\title[XXX $\widehat{\frak{sl}_N}$ Bethe alsatz equations
 and integrable hierarchies]
{Incarnations of XXX $\widehat{\frak{sl}_N}$ Bethe ansatz equations
\\ and integrable hierarchies}

\author
[I.\,Krichever,\  A.\,Varchenko]
{ Igor Krichever\>$^\circ\<$ and  Alexander Varchenko$\>^{\star}$}

\maketitle

\begin{center}
{\it
$^\circ\<$\> Columbia University,  2990 Broadway, New York, NY 10027, USA}
\vsk.3>

{\it
$^\circ\<$\>  Skolkovo Institute for
Science and Technology,
3 Nobelya,  Moscow, 121205, Russia}

\vsk.3>
{\it
$^\circ\<$\> Landau Institute for Theoretical Physics,
2 Kosygina,  Moskva,  119334, Russia}

\vsk.3>
{\it
$^\circ\<$\>  National Research University Higher
School of Economics,
7 Vavilova\\
 Moskva,  117312, Russia}

\vsk.3>
{\it
$^\circ\<$\>  Institute for Information Transmission Problems,
19 Bolshoi Karetnyi
\\ Moscow, 127051, Russia  }

\vsk.3>

{\it  $^{\star}$\> Department of Mathematics, University of North Carolina
at Chapel Hill\\ Chapel Hill, NC 27599-3250, USA\/}
\vsk.3>
{\it $^{\star}$\> Faculty of Mathematics and Mechanics, Lomonosov Moscow State
University\\ Leninskiye Gory 1, 119991 Moscow GSP-1, Russia\/}

\end{center}

\bigskip

{\let\thefootnote\relax
\footnotetext{\vsk-.8>\noindent
$^\circ\<${\sl E\>-mail}:\enspace krichev@math.columbia.edu
\\
$^\star\<${\sl E\>-mail}:\enspace anv@email.unc.edu\>,
supported in part by NSF grant DMS-1665239}}

\bigskip
\hfill
In memory of Boris Dubrovin (1950-2019)

\bigskip

\medskip
\begin{abstract}
We consider the space of solutions of the Bethe ansatz equations of the $\widehat{\frak{sl}_N}$ XXX quantum integrable model, associated with the trivial representation of $\widehat{\frak{sl}_N}$. We construct a family of commuting flows on this space and identify the flows with the flows of coherent rational  Ruijesenaars-Schneider systems. For that we develop in full generality the spectral transform for the rational  Ruijesenaars-Schneider system.

\end{abstract}

\medskip
\noindent
Keywords:  XXX Bethe ansatz equations, tau-functions, Baker-Akhieser functions

\setcounter{footnote}{0}
\renewcommand{\thefootnote}{\arabic{footnote}}

{\small \tableofcontents   }

\section{Introduction}

In the Gaudin model associated with a Lie algebra
one considers a commutative family of linear operators (Hamiltonians)
 acting on a tensor product of representations of the Lie algebra. To find common eigenvectors of  Hamiltonians
 one considers a suitable system of Bethe ansatz equations, and then
assigns an eigenvector to each solution of the system. That construction is called the Bethe ansatz method.

It turns out that the set of solutions of the Bethe ansatz equations is an interesting object. For example,
for the  affine Lie algebra  $\widehat{\frak{sl}_N}$ and its trivial representation the associated
system of the Bethe ansatz equations has the form
\bean
\label{Be 1}
\sum_{i' \neq i}\frac {2}{ u_i^{(n)} - u_{i'}^{(n)}}
-\sum_{i'=1}^{k_{n+1}} \frac{1}{ u_i^{(n)} -u_{i'}^{(n+1)}} -\sum_{i'=1}^{k_{n-1}}
\frac{1}{ u_i^{(n)} -u_{i'}^{(n-1)}}
= 0,
\eean
where $n = 1, \dots , N$ and $i = 1, \dots , k_n$. The system itself depends on the choice of nonnegative integers
$k_1,\dots,k_N$, which must satisfy the equation
\bean
\label{eqnk}
 \sum_{j=1}^N \frac {(k_j-k_{j+1})^2}2 -\sum_{j=1}^Nk_j\,=0\,.
\eean
Here we adopt the notations $k_{N+n} = k_{n}$ and  $u_i^{(N+n)} = u_i^{(n)}$ for all $i, n$.
The set of solutions of such a system  forms one cell or an empty set.  In \cite{VWr}
a family of commuting flows, acting on such a cell, was constructed. The family of flows was
identified with the flows of the  $N$ mKdV integrable hierarchy.

\vsk.3>

The initial goal of this paper was to extend these results to the $\widehat{\frak{sl}_N}$  XXX quantum integrable
model, associated with the trivial representation of $\widehat{\frak{sl}_N}$. In this case
the Bethe ansatz equations take the form
\bean
\label{BAE}
&&
\prod_{\ell=1}^{k_{n-1}} (u^{(n)}_i-u^{(n-1)}_\ell+1)
\prod_{\ell=1}^{k_n} (u^{(n)}_i-u^{(n)}_\ell-1)
\prod_{\ell=1}^{k_{n+1}} (u^{(n)}_i-u^{(n+1)}_\ell)
\\
\notag
&&
\phantom{aa}
+\ \
\prod_{\ell=1}^{k_{n-1}}
(u^{(n)}_i-u^{(n-1)}_\ell)
\prod_{\ell=1}^{k_n}
(u^{(n)}_i-u^{(n)}_\ell+1)
\prod_{\ell=1}^{k_{n+1}}
(u^{(n)}_i-u^{(n+1)}_\ell-1) =0\,,
\eean
where $n=1,\dots,N$, \ $i=1,\dots,k_n$, and the parameters
$k_1,\dots,k_N$ still satisfy   equation \Ref{eqnk}.

\vsk.3>

It turns out that we can do much more than just simple identification with a proper discrete analog
of the $N$ mKdV hierarchy. Roughly speaking we explicitly solve equations \Ref{BAE} using interplay with the theory of finite-dimensional integrable systems of particles, which are known to be equivalent to the theory of rational solutions of basic hierarchies considered in the framework of the theory of integrable partial differential, differential-difference and difference-difference equations.
One way to write any solution of the Bethe ansatz equations \Ref{BAE}
is to start with a suitable matrix $A$ and write the polynomials $(y_n(x) = \prod_{i=1}^{k_i}(x-u^{(n)}_i))_{n=1}^N$ as discrete Wronskians of
some auxiliary polynomials in $x$ associated with $A$, see Theorem \ref{5.5}.
Another way to write any solution is to start with a suitable flag in some infinite-dimensional vector space and write
these polynomials $(y_n(x))_{n=1}^N$  as discrete Wronskians of
some auxiliary polynomials in $x$ associated with the flag, see Corollary \ref{cor BA tau}.

\vsk.2>
In the remarkable paper \cite{amkm} it was observed that the dynamics of poles of the
elliptic (rational or trigonometric) solutions of the Korteweg\,--\,de Vries equation (KdV)
can be described in terms of commuting flows of the elliptic (rational or trigonometric) Calogero-Moser (CM) system restricted to the space of stationary points of the CM system. In \cite{K1} and \cite{K3}
this constrained correspondence between the theory of the elliptic CM system
and the theory of the elliptic solutions of the KdV equation was extended
to a similar construction of solutions of the KP equation in terms of the flows of the Calogero-Moser system.
Moreover
it was discovered  for the first time that this correspondence of solutions
can be established at the level of  {\it auxiliary linear} problems.

In the rational case, which we consider in this paper, the corresponding result is as follows:
the linear equation
\beq\label{intCM}
(\p_t-\p_x^2+u(x,t))\psi(x,t)=0
\eeq
with a {\it rational in $x$ potential} $u(x,t)$ vanishing as infinity, $u(x,t)\to 0$ as $x\to \infty$,
 has a {rational in $x$ solution}
if and only if the potential $u(x,t)$ is of the form
\begin{equation}\label{pot}
u(x,t)=2\sum_{i=1}^k(x-u_i(t))^{-2}=-2\p^2_x \ln y(x,t),
\end{equation}
and its poles $u_i(t)$ (a.k.a. the
zeros of the polynomial $y(x,t)$) as functions of $t$ satisfy the equations of motion of the rational CM system.

\vsk.3>
Recall, that the rational CM system with $k$ particles is a Hamiltonian system
with coordinates $u=(u_1,\dots,u_k)$, momentums $p=(p_1,\dots,p_k)$,  the canonical Poisson brackets $\{u_i,p_j\}=\delta_{ij}$, and  the Hamiltonian
\beq\label{HCM}
H=\frac 1 2 \sum_{i=1}^k p_i^2 +\sum_{i\neq j} \frac 1 {(u_i-u_j)^2}\,.
\eeq
The corresponding equations of motion,
\beq\label{CMequations}
\ddot u_i=2\sum_{j\neq i}\frac 1 {(u_i-u_j)^3}, \qquad i=1,\dots,k,
\eeq
admit the  Lax presentation $\dot L=[M,L]$ with
\beq\label{LCM}
L_{ij}\,=\,p_i\delta_{ij}\,+\,2\,\frac {1-\delta_{ij}}{u_i-u_j} ,  \qquad p_i=\dot u_i\,.
\eeq
The commuting flows, generated by the integrals $H_k=k^{-1}\tr L^k$, are called the
{\it hierarchy of the rational CM system}. Note that the Hamiltonian $H$ equals $H_2$.

\vsk.3>

It was shown in \cite{KZ} that the linear equation
\beq\label{intRS}
\p_t \psi(x,t)=\psi(x+1,t)+w(x,t)\psi(x,t)
\eeq
with
\beq\label{wRS}
w(x,t)=\p_t\ln \left(\frac{y(x+1,t)}{y(x,t)}\right),
\eeq
where $y(x,t)$ is a polynomial in $x$, has a solution rational in $x$ if and only if the zeros $u_i(t)$ of $y(x,t)$ satisfy the equations of motion of the rational Ruijesenaars-Schneider (RS) system.

\vsk.2>
The rational RS system with $k$ particles is a Hamiltonian system with coordinates
$u=(u_1,\dots,u_k)$,
momentums $p=(p_1,\dots,p_k)$,  the canonical Poisson brackets $\{u_i,p_j\}=\delta_{ij}$, and  the Hamiltonian
\beq\label{RSHam}
H(u,p)=\sum_{i=1}^k \ga_i
\eeq
where
\beq\label{gadef}
\ga_i:=e^{p_i}\prod_{j\neq i} \left(\frac{(u_i-u_j-1)(u_i-u_j+1)}{(u_i-u_j)^2}\right)^{1/2}.
\eeq
It is a completely integrable Hamiltonian system, whose equations of motion,
\bean\label{RSequations}
\dot u_i&=&\ga_i,  \\
\dot \ga_i&=&\sum_{j\neq i}\ga_i\ga_j\left(\frac 1{u_i-u_j-1}+\frac 1{u_i-u_j+1}-\frac 2{u_i-u_j}\right),
\qquad
i=1,\dots,k,
\eean
admit the Lax representation $\dot L=[M,L]$, where
\beq\label{dLR}
L_{ij}(u,\ga)=\frac {\ga_i}{u_i-u_j-1},\
\qquad i,j=1,\dots, k,
\eeq
\beq\label{MRSa}
M_{ij}=\left(\sum_{\ell\ne i}
\frac{\ga_\ell}{u_i-u_\ell}+\sum_{\ell}\frac{\ga_\ell}{u_i-u_\ell+1}\right)\delta_{ij}+(1-\delta_{ij})\frac {\ga_i}{u_i-u_j}.
\eeq
The functions $H_m=\tr L^m$ are integrals of the system.
Note that the Hamiltonian $H$ of the system equals $H_1$.
These integrals are in involution, and hence generate commuting flows called the {\it rational RS hierarchy}.

\vsk.2>

A scheme, in which an
integrable system of particles arises as a condition for a linear equation with elliptic (trigonometric, rational) coefficients to have
a double Bloch  solution (trigonometric, rational), was called in \cite{KZ}
a {\it  generating linear problems scheme}.

\vsk.2>
The next step had been done in \cite{KLWZ}. There  the system of linear equations
\beq\label{intBA}
\psi_{n+1}(x)=\psi_n(x+1)-v_n(x)\psi_n(x),\qquad n\in \mathbb Z,
\eeq
with respect to unknown functions  $(\psi_n(x))_{n\in\Z}$ was considered with
$$
v_n(x)=\frac{y_n(x)y_{n+1}(x+1)}{y_{n}(x+1)y_{n+1}(x)}\,,
$$
where  $(y_n(x))_{n\in\Z}$ is a given sequence of polynomials. It was shown that system \Ref{intBA} has a
solution $(\psi_n(x))_{n\in\Z}$ rational in $x$ with the  poles of $\psi_n(x)$ only  at the zeros of
$y_n(x)$, if and only if the zeros $(u_i^{(n)})_{i=1}^{k_n}$ of  $y_n(x)$
satisfy the Bethe ansatz equation \Ref{BAE}.

We stress that in \cite{KLWZ} the Bethe ansatz equations were considered for sequences of polynomials without the periodicity assumption that $y_n(x)=y_{n+N}(x)$ for some $N$.

\begin{rem}
In \cite{K4a} and \cite{K4} all three linear problems with $y(x,t)$ being an {\it entire} function in
 $x$ were used for the proof of the remarkable Welter's trisecant conjecture on
  the characterization of the Jacobians of smooth algebraic curves.
\end{rem}

In this paper  we apply  these ideas to  relate  solutions of the $N$-periodic Bethe ansatz equations \Ref{BAE} with
the equations of motion in the $N$-tuple of coherent
rational  Ruijesenaars-Schneider systems with respectively $k_1,\dots,k_N$ particles.

\vsk.2>
The paper is organized as follows.
In Section \ref{sec Inc} we reformulate the Bethe ansatz equations \Ref{BAE}
and prove formula \Ref{eqnk}. In Section \ref{sec GEN} we describe the procedure of generation of new solutions
of the system of Bethe ansatz equations, if one solution is given.  Theorem \ref{one gen} says that all solutions are obtained from
the single solution, namely, from the solution corresponding to the case of $k_1=\dots = k_N=0$.
\vsk.2>
In Section \ref{S:gen} we start using the generating linear problem \Ref{intBA} and its interplay with
two other generating linear problems. Having a solution of the Bethe ansatz equations we construct
a family of solutions $(\psi_n(x,z))$ of \Ref{intBA} parameterized by a complex parameter $z$,
see Theorem \ref{gener1}.
The construction reveals an unexpected connection with the theory of the RS system.
Namely, one of the steps in the proof of Theorem \ref{gener1} can be seen as a map from the space
of $N$-tuples of
 polynomials $(y_n(x))$ representing solutions of the Bethe ansatz equations
 to the product of $N$ phase spaces of the rational RS systems with respectively $k_1,\dots,k_N$ particles,
   i.e. as the map
\beq\label{intmap}
(y_n)\longmapsto (u^{(n)},\ga^{(n)})\,, \qquad n=1,\dots,N,
\eeq
where $\ga_i^{(n)}$ are defined in \Ref{dan}. On each of  these phase spaces
we define commuting flows with some times
$t=(t_1,t_2,\dots)$. That definition induces commuting flows with times $t$
on the product of the phase spaces.
One of our main results is the statement that the image of this map is invariant under
these commuting flows on the product of the phase spaces, see Theorem \ref{thm:imbed}.

\vsk.2>
In Section \ref{S:direct} we consider  the functions
$(\psi_n(x,z))$, constructed in Theorem \ref{gener1},
and study their analytic properties with respect to the {\it spectral}
parameter $z$. In this way we identify the functions $(\psi_n(x,z))$ with
a particular case of more general notion of the so-called Baker-Akhiezer functions.
The results of Section \ref{S:direct} can be seen as a construction of  the direct spectral transform for the
rational RS system. To our surprise we were unable to find in the literature such a construction in its full generality.

The analogous result for the rational CM system was obtained in \cite{W}.
Our construction of the direct spectral transform is different from the one in \cite{W}.
It is pure algebraic and does not require the use of infinite dimensional Grassmanians, whose
definition  involves elements of real analysis, in particular,  of
the theory of Fredholm operators.

\vsk.2>
In Section \ref{SRS} we  write equations for zeros of the polynomials obtained by the construction
of the Baker-Akhiezer functions corresponding to the spectral data of the rational RS systems.

\vsk.2>
In
 Section \ref{sec 7.1} we identify the spectral data
  corresponding to solutions of the $N$-periodic Bethe ansatz equations.
The rest of Section \ref{sec 7}
  is on the inverse spectral transform.
First we construct a family of solutions of the
generating linear problem starting from a certain matrix $A$,
see Theorem \ref{thm BAk}. That is done without any $N$-periodicity assumptions.
Then in Section \ref{sec per} we describe the matrices $A$ that give
$N$-periodic answers. Theorem \ref{5.5} can be seen as one of our main results.

\vsk.2>
For completeness in Section \ref{S:flows} we briefly present the integrable hierarchy,
 whose rational solutions describe the commuting flows on the space of solutions
 of the Bethe ansatz equations. We call it the {\it discrete $N$ mKdV hierarchy}.
 Section \ref{sec MO} contains a short remark of discrete Miura opers.

\vsk.2>
In Section \ref{sec COD} we discuss combinatorial data that will be used in Section  \ref{sec TB}. In Section
\ref{sec TB} we identify solutions of the Bethe ansatz equations with points of a suitable infinite dimensional Grassmannian.
We introduce a family of commuting flows on the Grassmannian and identify the flows induced on the space of solutions of the Bethe ansatz equations with the  flows of the discrete   $N$ mKdV hierarchy,
 introduced in Section \ref{S:flows}.

\section{Incarnations of the Bethe ansatz equations}
\label{sec Inc}

\subsection{Bethe ansatz equations} Let $N>2$ be a positive integer,
$\vec k=(k_1,\dots,k_N)\in \Z^N_{\geq 0}$. Denote
 $k:=k_1+\dots+k_N$.
Consider $\C^k$ with coordinates $u$ collected into $N$ groups, the $n$-th group consists of $k_n$ variables,
\bea
u=(u^{(1)},\dots,u^{(N)}),
\qquad
u^{(n)} = (u^{(n)}_1,\dots,u^{(n)}_{k_n}).
\eea
We adopt the notations $k_{N+n} = k_{n}$ and  $u_i^{(N+n)} = u_i^{(n)}$ for all $i, n$.

\vsk.2>

The   {\it Bethe ansatz equations} is the following system of $k$ equations:
\bean
\label{bae}
&&
\prod_{\ell=1}^{k_{n-1}} (u^{(n)}_i-u^{(n-1)}_\ell+1)
\prod_{\ell=1}^{k_n} (u^{(n)}_i-u^{(n)}_\ell-1)
\prod_{\ell=1}^{k_{n+1}} (u^{(n)}_i-u^{(n+1)}_\ell)
\\
\notag
&&
\phantom{aa}
+\ \
\prod_{\ell=1}^{k_{n-1}}
(u^{(n)}_i-u^{(n-1)}_\ell)
\prod_{\ell=1}^{k_n}
(u^{(n)}_i-u^{(n)}_\ell+1)
\prod_{\ell=1}^{k_{n+1}}
(u^{(n)}_i-u^{(n+1)}_\ell-1) =0\,,
\eean
where $n=1,\dots,N$, \ $i=1,\dots,k_n$.

\vsk.2>

These are the Bethe ansatz equations associated with the XXX quantum integrable
model of the
affine Lie algebra $\widehat{\frak{sl}_N}$ and the single representation with zero highest weight.
To study the associated Hamiltonians one assigns an eigenvector of Hamiltonians to a solution of
 the Bethe ansatz equations. We will not discuss this topic in this paper. Different versions of the Bethe ansatz equations associated with Lie algebras see, for example in \cite{OW, MV2, MV3, MV4}.

\begin{rem}
Equation \Ref{bae} with $N=2$  is the quasi-classical limit  of the Bethe ansatz equations derived in \cite{AL} for the Quantum Internal Long Wave model.
\end{rem}

\subsection{Polynomials representing a solution}
\label{PRCP}

Given $u = (u_i^{(n)})\in \C^k$,
introduce an $N$-tuple of polynomials $ y=(y_1(x),$ $\dots ,$ $ y_{N}(x))$,
\bean
\label{y}
y_n(x)\ =\ c_n \prod_{i=1}^{k_n}(x-u_i^{(j)}), \qquad
c_n\neq 0 .
\eean
We adopt the notations $y_{N+n}(x) = y_{n}(x)$ for any $n\in \Z$.
 Each polynomial is considered up to multiplication by a nonzero number.
The $N$-tuple defines a point in the direct product
$\PCN$, where $\on{P}(\C[x])$ is the projective space associated with $\C[x]$.
We say that the tuple $ y$  {\it represents} the point $u$.

\vsk.2>

Denote
\beq
\label{F_n}
F_n(x) := \frac {y_{n-1}(x+1)y_{n+1}(x)}{y_{n}(x+1)y_{n}(x)},\qquad L_n(x):= \frac{y_{n}(x+1)y_{n+1}(x-1)}{y_{n}(x)y_{n+1}(x)},
\eeq
\begin{lem}
\label{lem Fn}
Each equation in (\ref{bae}) can be reformulated as one of the following equations:
\beq\label{bae y}
y_{n-1}(u_j^{(n)}+1)
y_{n}(u_j^{(n)}-1)
y_{n+1}(u_j^{(n)})
+
y_{n-1}(u_j^{(n)})
y_{n}(u_j^{(n)}+1)
y_{n+1}(u_j^{(n)}-1) = 0,
\eeq
\bean
\label{bae F}
&
\res_{x=u^{(n)}_i}\left( F_n(x) + F_n(x-1)\right) =0,
&
\\
 \label{bae L}
&
\res_{x=u^{(n)}_i}\left( L_n(x) + L_{n-1}(x)\right) =0.
&
\eean
\qed
\end{lem}

\noindent
An important corollary of \Ref{bae L} is
\begin{cor}\label{thm nbae}
A generic $N$-tuple $ y$ represents a solution of the Bethe ansatz equations
\Ref{bae} if and only if the following equation holds:
\bean
\label{nbae}
L(x):=\sum_{n=1}^N L_n(x) = N.
\eean
\end{cor}
\noindent
This equation is a discrete version of \lq\lq{}the new form\rq\rq{} of the Bethe ansatz equations in the Gaudin model of an
arbitrary Kac-Moody algebra,  see \cite{MSTV}.
\begin{proof} Equation \Ref{bae L} is equivalent to the condition that the function $L(x)$ defined in \Ref{nbae} has no poles.
Each of the function $L_n(x)$  tends to 1 as $x\to \infty$. Hence, $L(x)=N$.
\end{proof}
In its own turn Corollary \ref{thm nbae} directly implies the following important statement. Consider the quadratic form
\bea
Q(k_1,\dots,k_N)
&=& \sum_{j=1}^N k_j(k_j-1) - k_1k_2-\dots-k_{N-1}k_N-k_Nk_1\,
\\
&=& \sum_{j=1}^N \frac {(k_j-k_{j+1})^2}2 -\sum_{j=1}^Nk_j\,,
\eea
introduced in \cite{MV3}.

\begin{cor}
\label{cor Q=0} If an $N$-tuple of polynomials $(y_1,\dots,y_N)$ of degrees $(k_1,\dots,k_N)$ represents a solution of the Bethe ansatz equations \Ref{bae}, then
then
\beq
\label{Q=0}
Q(k_1,\dots,k_N)=0.
\eeq
\end{cor}
\begin{proof} Expanding  at infinity, we observe that $L(x)-N = Q(k_1,\dots,k_N)x^{-2} + \mc O(x^{-3})$.
\end{proof}

\begin{cor}
\label{cor k=0} If an $N$-tuple of polynomials $(y_1,\dots,y_N)$ of degrees $k_1=\dots=k_N$ represents a solution of the Bethe ansatz equations \Ref{bae}, then $k_1=\dots=k_N=0$.
\qed
\end{cor}

\begin{rem}
Equations \Ref{bae y}, \Ref{bae F},
\Ref{nbae} can be thought of as
incarnations of the Bethe ansatz equations \Ref{bae}.
\end{rem}

\section{Generation of solutions of Bethe ansatz equations}
\label{sec GEN}
\subsection{Discrete Wronskian}

For arbitrary functions $f_1(x), \dots, f_m(x)$ introduces
the {\it discrete Wronskian}
by the formula:
\bean
\label{wr}
\Wh(f_1, \dots, f_m) = {\det}_{i,j=1}^m\big(f_i(x+j-1)\big).
\eean
For example,
\bea
\Wh(f_1, f_2) = f_1(x)f_2(x+1) - f_1(x+1)f_2(x).
\eea

\vsk.2>

Denote
\bean
\label{Del}
\phantom{aaaa}
\Delta f(x) = f(x+1)-f(x)\,,
\quad
\Delta ^{(n+1)}f(x)
= \Delta(\Delta ^{(n)}f)(x)\,,
\quad
\Delta ^{(0)}f(x)=f(x)\,.
\eean
Then
\bean
\label{W-Del}
\Wh(f_1,\dots,f_n) = {\det}_{i,j=1}^n(\Delta^{j-1}f_i(x)).
\eean

\begin{lem}
 [\cite{MV2}]
\label{lem 9.1} We have
\bean
\label{9.1}
\Wh(1, f_1, . . . , f_n)(x) = \Wh(\Delta f_1, \dots, \Delta f_n)\,.
\eean
\end{lem}

\begin{lem}
 [{\cite[Lemma 9.4]{MV2}}]
\label{lem 6.4}
For functions $f_1(x),\dots,f_n(x),g_1(x),g_2(x)$  we have
\bean
\label{wr MV}
&&
\Wh(\Wh(f_1,\dots,f_n,g_1),\Wh(f_1,\dots,f_n,g_2))(x)
\\
\notag
&&
\phantom{aaaaaa}
 = \Wh(f_1,\dots,f_n)(x)\,\Wh(f_1,\dots,f_n,g_1,g_2)(x+1)\,.
\eean
\end{lem}

\subsection{Elementary generation}
\label{Elg}

We say that an   $N$-tuple of  polynomials $  y=(y_1(x),$ $\dots ,$ $ y_{N}(x))$ is {\it generic} if for any $n$,
the polynomial $y_{n}(x)$ has no common zeros with the polynomials
$y_{n}(x+1)$,
$ y_{n-1}(x+1)$,
$y_{n+1}(x)$.

\vsk.2>
We say that an  $N$-tuple of  polynomials $ y=(y_1(x),$ $\dots ,$ $ y_{N}(x))$ is {\it fertile}, if for any $n$
the first order difference equation
\bean
\label{fert}
\Wh(y_n,\tilde y_n) = y_{n-1}(x+1)y_{n+1}(x)
\eean
with respect to $\tilde y_n(x)$
has a polynomial solution.

If $\tilde y_n(x)$ is a polynomial solution of \Ref{fert}, then all other polynomials solutions are of the form
\bea
\tilde y_n(x,c) \,=\, \tilde y_n(x) + c y_n(x)
\eea
 for $c\in\C$.
The tuples
\bean
\label{sml}
 y^{(n)}(c) := (y_1(x), \dots , \tilde y_n(x,c),\dots, y_{N}(x))
\quad \in  \quad \PCN \
\eean
 form a one-parameter family.  This family  is called
 the {\it generation  of tuples  from $ y$ in the $n$-th direction}.
 A tuple of this family is called an  {\it immediate descendant} of $  y$ in the $n$-th direction.

\vsk.2>
For example, the $N$-tuple
\bean
\label{y empty}
 y^\emptyset = (1,\dots,1)
\eean
 of constant polynomials is fertile,
 and
$ y^{\emptyset, (n)}(c) = (1,\dots, 1, x+c, 1\dots, 1)$.

\vsk.2>
It is convenient to think that $ y^\emptyset$ represents a solution of the Bethe ansatz equations with
$k=0$, see \Ref{bae F}.

\begin{thm}
[\cite{MV2}, cf. \cite{MV1}]
\label{f cor}
${}$

\begin{enumerate}
\item[(i)]
A generic tuple $ y = (y_1, \dots , y_{N})$
represents a solution of the Bethe ansatz equations \Ref{bae}
if and only if $ y$ is fertile.

\item[(ii)] Let $  y$ represent a solution of the Bethe ansatz equations \Ref{bae},
$n \in \{1,\dots,N\}$, and $  y^{(n)}(c)$
an immediate descendant of $ y$, then $ y^{(n)}(c)$ is fertile for any $c\in\C$.

\item[(iii)]
If $ y$ is generic and fertile, then for almost all values of the parameter
 $c\in \C$ the corresponding $n$-tuple $ {y}^{(n)}(c)$ is generic.
The exceptions form a finite set in $\C$.

\end{enumerate}
\end{thm}

\subsection{Degree increasing generation}
\label{Dig}

For $n=1,\dots,N$, let $k_n=\deg y_n$. The polynomial $\tilde y_n$ in \Ref{fert}
is of degree $k_n$ or
$\tilde k_n=k_{n-1} + k_{n+1}+ 1 - k_n$. We say that the {\it generation is
  degree increasing }  if $\tilde k_n > k_n$. In that
case $\deg \tilde y_n=\tilde k_n$ for all $c$.

\vsk.2>

If the generation is degree increasing, we will normalize the family
 \Ref{sml} and construct a map
$ Y_{ y,n} : \C \to (\C[x])^N$ as follows. First we multiply the polynomials $y_1,\dots,y_N$ by numbers to make them monic.
 Then we
choose a monic polynomial $ y_{n,0}(x)$ satisfying the equation $\Wh(y_n,  y_{n,0})$
 $=\,\on{const}\,y_{n-1}(x+1)y_{n+1}(x)$
 and such that the coefficient of $x^{k_n}$ in $\tilde y_{n,0}(x)$ equals zero. We define
 \bean
 \label{ti y}
 \tilde y_n(x,c)=y_{n,0}(x) + cy_n(x)
 \eean
 and
\bean
\label{nzd}
\phantom{aaa}
 Y_{ y,n} \ :\ \C\ \to\ (\C[x])^N, \qquad c \mapsto\  y^{(n)}(c) = (y_1(x),\dots,\tilde y_n(x,c),\dots, y_N(x)).
 \eean
All polynomials of the tuple $ y^{(n)}(c)$ are monic.

\subsection{Degree-transformations and generation of vectors of integers}
\label{sec dtr}

For $j=1,\dots,N$, the degree-transformation
\bean
\label{l-tr}
\phantom{aaaaaa}
\vec k=(k_1,\dots,k_N)\ \ \mapsto \ \
\vec k^{(j)}=(k_1,\dots,k_{j-1},k_{j-1}+k_{j+1}-k_j+1,k_{j+1},\dots, k_N)
\eean
corresponds to the shifted action of the affine reflection $ w_j\in  W_{A_{N-1}}$,
where $W_{A_{N-1}}$ is the affine Weyl group of type $A_{N-1}$ and $w_1,\dots, w_N$
are its standard generators, see Lemma 3.11 in \cite{MV1} for more detail.

\vsk.2>

We take formula \Ref{l-tr} as the definition of {\it degree-transformations}:
\bean
\label{sl-tr}
 w_j\ :\ \vec k=(k_1,\dots,k_N)\ \ \mapsto\ \
\vec k^{(j)}=(k_1,\dots,k_{j-1}+k_{j+1}-k_j+1,\dots, k_N)
\eean
for $j=1,\dots,N$. The degree-transformations act
 on arbitrary vectors $\vec k=(k_1,\dots,k_N)$.

\vsk.2>
In this formula we consider the indices of
the coordinates modulo $N$, that is, we have $k_{N+j}=k_j$ for all $j$.

\vsk.2>

We start with the vector $\vec k^\emptyset=(0,\dots,0)$ and a sequence $J=(j_1,j_2,\dots,j_m)$ of integers,
$1\leq j_i\leq N$. We apply the corresponding degree transformations to
the vector  $\vec k^\emptyset$ and obtain
a sequence of vectors $\vec k^\emptyset,$\ $  \vec k^{(j_1)} :=w_{j_1}\vec k^\emptyset, $\ $
  \vec k^{(j_1,j_2)} := w_{j_2} w_{j_1}\vec k^\emptyset$,\dots ,
\bean
\label{g v}
\vec k^J  := w_{j_m}\dots w_{j_2} w_{j_1}\vec k^\emptyset .
\eean
We say that the {\it vector $\vec k^J $ is generated from $(0,\dots,0)$ in the direction of $J$}.

\vsk.2>

We call the sequence $J$ {\it degree increasing} if for every $i$ the transformation
$w_{j_i}$ applied to  $  w_{j_{i-1}}\dots w_{j_1}\vec k^\emptyset$
increases the $j_i$-th coordinate.

\subsection{Multistep generation}
\label{sec gp}

Let $J = (j_1,\dots,j_m)$ be a degree increasing sequence of integers.
Starting from $ y^\emptyset=(1,\dots,1)$ and $J$, we construct,
by induction on $m$, a map
\bea
Y^J\  :\  \C^m \ \to\  (\C[x])^N.
\eea
If $J=\emptyset$, the map $Y^\emptyset$ is the map $\C^0=(pt)\ \mapsto  y^\emptyset$.
If $m=1$ and $J=(j_1)$,  the map
$Y^{(j_1)} :  \C \to (\C[x])^N$ is given by formula \Ref{nzd}.
More precisely,
\bea
Y^{(j_1)}\ :\ \C \mapsto (\C[x])^N, \qquad
c \mapsto (1,\dots,1,x+c, 1\dots,1),
\eea
where $x+c$ stands at the $j_1$-th position. By Theorem \ref{f cor}
all tuples in the image are fertile and almost all tuples are generic
(in this example all tuples are generic).

Assume that for ${\tilde J} = (j_1,\dots,j_{m-1})$,  the map
$Y^{{\tilde J}}$ is already constructed. To obtain  $Y^J$ we apply the
generation procedure in the $j_m$-th
direction to every tuple of the image of $Y^{{\tilde J}}$. More precisely, if
\bean
\label{J'}
Y^{{\tilde J}}\ : \
{\tilde c}=(c_1,\dots,c_{m-1}) \ \mapsto \ (y_1(x,{\tilde c}),\dots, y_N(x,{\tilde c})).
\eean
Then
\bean
\label{Ja}
&&
{}
\\
\notag
&&
Y^{J} : \C^m \mapsto (\C[x])^N, \quad
({\tilde c},c_m) \mapsto
(y_1(x,{\tilde c}),\dots,  y_{j_m,0}(x,{\tilde c}) + c_m y_{j_m}(x,{\tilde c}),\dots,
y_N(x,{\tilde c})),
\eean
see formula \Ref{ti y}.
The map  $Y^J$ is called  the {\it generation  of $N$-tuples   from $ y^\emptyset$ in the $J$-th direction}.

\vsk.2>

All tuples in the image of $Y^J$ are fertile and almost all tuples are generic. For any $c\in\C^m$
the $N$-tuple $Y^J(c)$ consists of monic polynomials. The degree vector of this tuple
equals $\vec k^J$, see \Ref{g v}.

\vsk.2>

The set of all tuples $(y_1,\dots,y_N)\in (\C[x])^N$ obtain from $ y^\emptyset=(1,\dots,1)$
by generations in all degree increasing directions will be called the {\it population of $N$-tuples}
generated from $ y^\emptyset$.

\subsection{Population generated from $ y^\emptyset$}

\begin{thm} [{\cite{MV4}}]
\label{one gen}

If an $N$-tuple of polynomials $ y=(y_1,\dots,y_N)$ with degree vector $\vec k$
represents a solution of the Bethe ansatz equations \Ref{bae}, then
$ y$  is a point of the population generated from $ y^\emptyset$
by degree increasing generations, that is,
 there exist a degree increasing sequence $J=(j_1,\dots,j_m)$  and a point $c\in\C^m$
such that  $ y = Y^J(c)$.

Moreover,  for any other $N$-tuple $ y\rq{}$,
representing a solution of the Bethe ansatz equations \Ref{bae} and having the same degree vector
$\vec k$,  there is a
point $c\rq{}\in\C^m$ such that  $  y\rq{}=Y^J(c\rq{})$.

\end{thm}

By Theorem \ref{one gen} the $N$-tuples $ y$,
representing solutions of the Bethe ansatz equations \Ref{bae} with the same degree vector
$\vec k$,  form one cell  $\C^m$.

\begin{proof}
The proof of Theorem \ref{one gen}  is word by word the same as the proof of \cite[Theorem 3.8]{MV5},
although the generation procedure in \cite{MV5} is slightly different from the generation procedure in this paper.
The key point of the proof is the equality $Q(\vec k)=0$, which is
 proved in Corollary \ref{Q=0} for our generation procedure and was
 proved in the proof of \cite[Theorem 3.8]{MV5}.  See also the proof of \cite[Theorem 6.4]{VW1}.
\end{proof}

\begin{rem}
The condition of fertility of an $N$-tuple $ y$
 can be also
thought of as another incarnation of the Bethe ansatz equations \Ref{bae}, see Theorem \ref{f cor}.
\end{rem}

\section{Generating linear problem}
\label{S:gen}

\subsection{Non-periodic sequences of polynomials}
\label{sec npsp}

In this section we consider sequences of polynomials $ y=(y_n(x))_{n\in\Z}$\,, not
assuming that the sequences are $N$-periodic. Let
\bea
y_n(x)\ =\ c_n \prod_{i=1}^{k_n}(x-u_i^{(n)}), \qquad c_n\neq 0.
\eea
The  system of the {\it Bethe ansatz equations} in this case is the infinite system of equations:
\bean
\label{baei}
&&
\prod_{\ell=1}^{k_{n-1}} (u^{(n)}_i-u^{(n-1)}_\ell+1)
\prod_{\ell=1}^{k_n} (u^{(n)}_i-u^{(n)}_\ell-1)
\prod_{\ell=1}^{k_{n+1}} (u^{(n)}_i-u^{(n+1)}_\ell)
\\
\notag
&&
\phantom{aa}
+\ \
\prod_{\ell=1}^{k_{n-1}}
(u^{(n)}_i-u^{(n-1)}_\ell)
\prod_{\ell=1}^{k_n}
(u^{(n)}_i-u^{(n)}_\ell+1)
\prod_{\ell=1}^{k_{n+1}}
(u^{(n)}_i-u^{(n+1)}_\ell-1) =0,
\eean
where $n\in\Z$, \ $i=1,\dots,k_n$.

\vsk.2>
We say that the sequence $ y$  is {\it generic} if for any $n$
the polynomial $y_{n}(x)$ has no common zeros with the polynomials
$y_{n}(x+1)$,
$ y_{n-1}(x+1)$,
$y_{n+1}(x)$.

\vsk.2>

As in the periodic case the system of  the Bethe ansatz equations \Ref{baei} can be reformulated as the infinite
system of equations   \Ref{bae y}, or equations \Ref{bae F}, or equations
\Ref{bae L}.

\begin{rem}
Let the degrees $(k_n)_{n\in\Z}$ of the polynomials  $(y_n(x))_{n\in\Z}$
be all equal. Then for each $n$ system \Ref{baei} can be regarded as a system of equations for
$(u_i^{(n+1)})$ with $(u_i^{(n)})$ and $(u_i^{(n-1)})$ given.
Hence,   system \Ref{baei} can be seen as a second
 order discrete time dynamical system. In such a form these equations
  were
 introduced in \cite{NRK} as an integrable time-discretization
  of the Ruijesenaars-Schneider system,
  which in its turn was introduced as a relativistic analog of the  Calogero-Moser (CM) system.

In \cite{KLWZ} for system \Ref{baei}
the discrete time Lax representation with a "spectral parameter"  was found with the
 help of a "generating linear problem", see Theorem 6.1 in \cite{KLWZ}.
 The Hamitonian approach for this system was developed in \cite{K3}.

Notice that the case of all $(k_n)_{n\in\Z}$ being equal is not allowed in the periodic case
by Corollary \ref{cor k=0}. This fact can be interpreted as the
statement that {\it the   time-discretization
  of the Ruijesenaars-Schneider system has no periodic orbits.}

\end{rem}

Given a generic  sequence of polynomials $ y=(y_n(x))_{n\in\Z}$ the associated
 {\it generating linear problem} is the infinite system of equations
\beq\label{laxdd}
\psi_{n+1}(x)=\psi_{n}(x+1)-v_n(x) \psi_{n}(x), \qquad n\in\Z,
\eeq
with respect to the unknown sequence of functions $ \psi = (\psi_n(x))_{n\in\Z}$
with $ v=(v_n(x))$ given by the formulas
\bean
\label{vpot}
v_n(x)=\frac{y_{n}(x)\,y_{n+1}(x+1)}{y_n(x+1) \,y_{n+1}(x)}\,.
\eean

We say that a solution $ \psi = (\psi_n(x))_{n\in\Z}$ of system \Ref{laxdd}
is {\it admissible}  if for any $n$ the function $y_n(x)\psi_n(x)$ is holomorphic.

Define the nonzero numbers
\bean
\label{dan}
\ga_i^{(n)}
:=
\res_{x=u_i^{(n)}-1} v_n(x)
=
\frac{y_n(u_i^{(n)}-1)\, y_{n+1}(u_i^{(n)})}{
\prod_{j\neq i}(u_i^{(n)}-u_j^{(n)})
\,y_{n+1}(u_i^{(n)}-1)}\, ,
\eean
where $i=1,\dots,k_n$, and nonzero numbers
\bean
\label{resonemore}
\eps_i^{(n)}\,
:=
\,\res_{x=u_i^{(n+1)}}v_n(x)\,=\,
\frac {y_n(u_i^{(n+1)})\, y_{n+1}(u_i^{(n+1)}+1)}{y_n(u_i^{(n+1)} +1)
\prod_{j\neq i} (u_i^{(n+1)}-u_i^{(n+1)})}\, ,
\eean
where $i=1,\dots,k_{n+1}$.
\begin{lem}
\label{lem ep+ga}
The infinite system of equations
\bean
\label{ga+ep}
\ga^{(n+1)}_i + \eps^{(n)}_i =0, \qquad n\in \Z, \quad
i=1,\dots,k_{n+1},
\eean
is equivalent to the infinite system of equations \Ref{bae y}.
\qed
\end{lem}

In its turn the property of the infinite system of equations \Ref{bae y} to have a solution $y$
is equivalent to the
property of $ y$ to represent a solution of the Bethe ansatz equations \Ref{baei}, see
Lemma \ref{lem Fn}.

\begin{thm}
\label{gener1}

Let $ y=(y_n(x))_{n\in\Z}$ be a generic sequence of polynomials. Then
the  system of  equations \Ref{laxdd} has an admissible  solution $ \psi=(\psi_n(x))_{n\in\Z}$
if and only if  $ y$ represents a solution of   system
\Ref{baei}.
Moreover, if a generic sequence $ y$ represents a
solution of  system \Ref{baei}, then there exists a unique
one-parameter family $ \Psi(z) = (\Psi_n(x,z))$ of
admissible solutions of
system
\Ref{laxdd}, which has the form
\beq
\label{bakdd}
\Psi_n(x,z)=z^{n} (1+z)^x\left(1+\sum_{i=1}^{k_{n}} \xi^{(n)}_{i}\!(x) z^{-i}\right),
\qquad n\in\Z,
\eeq
where $\xi^{(n)}_i\!(x)$ are rational functions in $x$ such that
the functions $y_n(x)\,\xi^{(n)}\!(x)$ are holomorphic in $x$.
\end{thm}

\begin{rem}
The first statement of the theorem is an analog of Lemma 5.1 in \cite{K4}, and the second statement is a stronger version
of Lemma 5.2 in \cite{K4}.
\end{rem}

\vsk.2>
\begin{rem}

The equivalence in Theorem \ref{gener1}
of the existence of an admissible solution $ \psi$ of system
\Ref{laxdd} and the property of $ y$ to represent a solution of system
 \Ref{baei}
may be thought of as another incarnation of the Bethe ansatz equations.
\end{rem}
\vsk.2>

\begin{proof} Let $ \psi$ be an admissible solution of  the generating linear problem equation (\ref{laxdd}).
For any $n\in\Z$ and $i=1,\dots,k_n$, consider
 the Laurent expansion of $\psi_n(x)$
at  $x=u_i^{(n)}$,
\bean
\label{psidexp}
\psi_n(x)
&=&
\frac {\alpha_i^{(n)}}{x-u_i^{(n)}}
+\mc O\left(1\right)\,, \qquad
\alpha_i^{(n)}
\in\C\,.
\eean
 The comparison of the residues of the left and right-hand sides of equation \Ref{laxdd} at
$x=u_i^{(n)}-1$ and $x=u_j^{(n+1)}$ gives us the equations
\bean
\label{res-1}
\alpha_i^{(n)}
&=&
\ga_i^{(n)}\,
\psi_{n}(u_i^{(n)}-1)\,,
\\
\label{res+1}
\alpha_j^{(n+1)}
&=&
-\,\eps_j^{(n)}\,
\psi_{n}(u_j^{(n+1)})\,,
\eean
respectively.  We obtain the third set of  equations
\bean
\label{psieq}
\psi_{n+1}(u_j^{(n+1)}-1)\,=\,\psi_n(u_j^{(n+1)})\,, \qquad j=1,\dots, k_{n+1},
\eean
by substituting $x=u_j^{(n+1)}-1$ to equation \Ref{laxdd} and taking into account that $v_n(u_j^{(n+1)}-1)$  $=0$.
Shifting the index $(n,i) \to (n+1,j)$ in (\ref{res-1}) we obtain
\bean
\label{res-11}
\alpha_j^{(n+1)}
=
 \ga_j^{(n+1)}\,
\psi_{n+1}(u_j^{(n+1)}-1).
\eean
Using  (\ref{res+1}), (\ref{psieq}), \Ref{res-11}  we
obtain equations
$\ga^{(n+1)}_i + \eps^{(n)}_i =0$
for $n\in\Z$ and $j=1,\dots,k_{n+1}$,
which are equations \Ref{ga+ep}.  By Lemma \ref{lem ep+ga}  this means that
the sequence $ y$ represents a solution of the Bethe ansatz equations
\Ref{baei}. That proves the "only if" part of the first statement of the theorem.

\vsk.3>
Now the goal is to construct the family  $ \psi(z)$ of admissible solutions
of \Ref{laxdd} assuming that $ y$ is generic and
represents a solution of  \Ref{baei}. The construction has two steps.
First, we construct  a certain sequence of functions $ \psi(z)$
 by using the generic
 $ y$, but not using the fact that $ y$ satisfies \Ref{baei}. Then we  prove that
 $ \psi(z)$  has  the form \Ref{bakdd}  and
is a solution of  \Ref{laxdd}, if $ y$ represents a solution of \Ref{baei}.

\begin{lem}
\label{lem 4.3}

Let $ y$ be a generic sequence of polynomials. Then for $n\in \Z$
there exists a unique function $\psi_n(x,z)$ of the form
\beq
\label{psid}
\psi_n(x,z) =z^n(1+z)^x\left(1+\sum_{i=1}^{k_n}\frac{C_i^{(n)}\!(z)}{x-u_i^{(n)}}\right)
\eeq
such that the function
\beq
\label{phi}
\phi_n(x,z):=\psi_n(x+1,z)-v_n(x)\psi_n(x,z)
\eeq
has no residues at $x=u_i^{(n)}-1$ for all $i=1,\dots,k_n$,
\beq\label{resphi}
\res_{x=u_i^{(n)}-1}\phi_n(x,z)=0\,.
\eeq
\end{lem}

\begin{rem}

Notice that $C_i^{(n)}(z)$ are some functions in $z$. The proof shows that
$C_i^{(n)}(z)$ are  rational functions in $z$.

\end{rem}

\begin{rem}
Notice that $\phi_n(x,z)$ would  be equal to $\psi_{n+1}(x,z)$ if
the sequence $(\psi_n(x,z))$ were a solution of the system
of the generating  linear problem equations \Ref{laxdd}.
\end{rem}

\begin{proof}

By \Ref{psid} the function $\psi_n(x,z)$ is regular at $x=u_i^{(n)}-1$. We also have
\beq
\label{psiresshift}
\res_{x=u_i^{(n)}-1} \psi_n(x+1,z)=\res_{x=u_i^{(n)}}\psi_n(x,z)\,.
\eeq
Hence, equation (\ref{resphi}) is equivalent to the equation
\beq\label{resdd}
\res_{x=u_i^{(n)}}\psi_n(x,z)-\ga_i^{(n)}\,\psi_n(u_i^{(n)}-1,z)\,=\,0.
\eeq

Let $C^{(n)}\!(z)$ be the $k_n$-vector with coordinates $C_i^{(n)}\!(z)$ appearing in
\Ref{psid}.  Let  $\ga^{(n)}$ be the $k_n$-vector with coordinates $\ga_i^{(n)}$.
Let  $L^{(n)}\!(z)$ be the  $k_n\times k_n$-matrix with entries
\beq
\label{dL}
L^{(n)}_{ii}\!(z)=1+z+\ga_i^{(n)}, \qquad L^{(n)}_{ij}\!(z)=\frac {-\,\ga_i^{(n)}}{u_i^{(n)}-u_j^{(n)}-1},\qquad i\neq j\,.
\eeq
Then the substitution of (\ref{psid}) into  (\ref{resdd})
gives an inhomogenous linear equation
\beq
\label{jan29a}
L^{(n)}\!(z)\,C^{(n)}\!(z)\,=\,\ga^{(n)}
\eeq
with respect to $C^{(n)}\!(z)$. Indeed, the substitution gives us
\bea
(1+z)\,C^{(n)}_i\!(z)  \,-\, \ga_i^{(n)}\left(1+\sum_{j=1}^{k_n}\frac{C_j^{(n)}\!(z)}{u_i^{(n)}-u_j^{(n)}-1}\right)
=0,
\eea
which implies \Ref{jan29a}.  It is clear that  for generic $z$ we have
 $\det L^{(n)}\!(z) \ne 0$  and equation \Ref{jan29a} has a unique solution
$C^{(n)}\!(z)$. The lemma is proved.
\end{proof}

Below we  give a determinant formula for $\psi_n(x,z)$. By Cramer's rule we have
\beq\label{kramerd}
C_i^{(n)}\!(z)=\frac {\det L^{(n)}_{i}\!(z)}{\det L^{(n)}\!(z)},
\eeq
where $L^{(n)}_{i}\!(z)$ is the matrix obtained from $L^{(n)}\!(z)$ by replacing
the $i$-th column by the
\linebreak
vector\ $\ga^{(n)}$.

Define a  $(k_{n}+1)\times (k_{n}+1)$  matrix  $\widehat L^{(n)}\!(x,z)$, whose rows and columns are
labeled by indices $0,\dots,k_n$ and entries are given by the formulas:
\bean
\label{hatLd}
&&
\widehat L^{(n)}_{0,0}=1\,,\qquad  \widehat L^{(n)}_{0,j}=\frac 1{x-u_j^{(n)}}\,,
\qquad \widehat L^{(n)}_{i,0}=-\ga^{(n)}_i\,,
\\
\notag
&&
\phantom{aaaaaaa}
\widehat L^{(n)}_{i,j}\,=\, L^{(n)}_{i,j}\,, \qquad i,j=1,\dots, k_{n}\,.
\eean
Using the determinant expansion of $\widehat L^{(n)}(z)$ relative to  the $0$-th row  we obtain
the formula
\beq
\label{psiwd}
\psi_n(x,z)=z^{n}(1+z)^x\, \frac {\det \widehat L^{(n)}(x,z)}{\det L^{(n)}(z)}\,.
\eeq

\begin{lem}
If $ y$ represents a solution of the Bethe ansatz equations \Ref{baei},
then the sequence $ \Psi(z)$, constructed in Lemma \ref{lem 4.3},
is an admissible solution  of \Ref{laxdd}.
\end{lem}

\begin{proof}

By definition of $\psi_n(x,z)$ and $\phi_n(x,z)$,  the function
\bea
R_n(x,z)\,:=\,\phi_n(x,z)\,z^{-n}(1+z)^{-x}
\eea
is a rational function of $x$ with at most first order poles at the zeros of $y_{n+1}(x)$.
Since $v_n(x)\to 1$ as $x\to\infty$, we have
$R_n(x,z) \to 1+z-1=z$ as $x\to\infty$. Hence,  the function $\phi_n(x,z)$ has the form
\beq\label{phipsi}
\phi_n(x,z)\,=\,
z^{n+1}(1+z)^x\left(1+\sum_{i=1}^{k_{n+1}} \frac{ D_i^{(n)}(z)}{x-u_i^{(n+1)}}\right)
\eeq
with suitable functions $ D_i^{(n)}(z)$.

Since the function $\psi_n(x+1,z)$ is regular at $x=u_i^{(n+1)}$, it follows from (\ref{phi})  that
\beq
\label{resphi o}
\res_{x=u_i^{(n+1)}} \phi_n(x,z)\,=\,-\,\eps_i^{(n)}\psi_n(u_i^{(n+1)},z)\,.
\eeq
From the equation $v_n(u_i^{(n+1)}-1)=0$ it follows that
\beq
\label{psi=phi}
\phi_n(u_i^{(n+1)}-1,z)=\psi_n(u_i^{(n+1)},z).
\eeq
Hence
\bean
\label{psin+1}
\res_{x=u_i^{(n+1)}} \phi_n(x,z)\,+\,\eps_i^{(n)}\phi_n(u_i^{(n+1)}-1,z)\, =\,0\,.
\eean
Using equations \Ref{ga+ep} we rewrite this as
\bean
\label{psin+11}
\res_{x=u_i^{(n+1)}} \phi_n(x,z)\,-\,\ga_i^{(n+1)}\phi_n(u_i^{(n+1)}-1,z)\, =\,0\,.
\eean
By Lemma \ref{lem 4.3} the function $ \psi_{n+1}(x,z)$ is uniquely determined by the equations
\bean
\label{psin+11a}
\res_{x=u_i^{(n+1)}} \psi_{n+1}(x,z)\,-\,\ga_i^{(n+1)}\psi_{n+1}(u_i^{(n+1)}-1,z)\, =\,0\,.
\eean
Hence $\phi_n(x,z)=\psi_{n+1}(x,z)$ and the lemma is proved.
\end{proof}

For any $n\in\Z$, let $q_n(z)$ be the monic polynomial of minimal degree such that
$q_n(0)\ne 0$ and the function $q_n(z)\,\frac {\det \widehat L^{(n)}(x,z)}{\det L^{(n)}(z)}$ is a function in $z$ holomorphic on
$\C-\{0\}$.
Clearly  the polynomial $q_n(z)$ does exist, it divides the polynomial
$\det L^{(n)}\!(z)$, and $\deg q_n(z) \leq k_n$.

\begin{lem}\label{4.5}
The polynomial $q_n(z)$ does not depend on $n\in\Z$.
\end{lem}

\begin{proof} Equation \Ref{laxdd} implies
\bean
\label{laxr}
z\,
\frac {\det \widehat L^{(n+1)}(x,z)}{\det L^{(n+1)}(z)}
=(1+z)\frac {\det \widehat L^{(n)}(x+1,z)}{\det L^{(n)}(z)}
-v_n(x)\frac {\det \widehat L^{(n)}(x,z)}{\det L^{(n)}(z)}\,.
\eean

Given $\zeta\ne 0$, let $d_n$ be the multiplicity  of the root $z=\zeta$ of
the polynomial  $q_n(x)$. We need to show that
$d_n=d_{n+1}$.
Clearly the inequality $d_n<d_{n+1}$ contradicts to equation \Ref{laxr}. Now we assume that
$d_n>d_{n+1}$ and also  will  obtain a contradiction. Namely, consider the expansions
\bea
\frac {\det \widehat L^{(n)}(x,z)}{\det L^{(n)}(z)}
&=&
 c_{n}(x) (z-\zeta)^{-d_n} + \mc O\big((z-\zeta)^{-d_n+1}\big)
\\
&=&
\big(b x^a + \mc O(x^{a-1})\big)(z-\zeta)^{-d_n} + \mc O\big((z-\zeta)^{-d_n+1},
\eea
where the first equality  is the Laurant expansion of $\frac {\det \widehat L^{(n)}(x,z)}{\det L^{(n)}(z)} $
at $z=\zeta$, and  $c_n(x)  = b x^a + \mc O(x^{a-1})$ is the Laurent expansion of $c(x)$ at $x=\infty$. Here
$a$ is a suitable integer and $b$ a nonzero number.  We also have $v_n(x) = 1 + \mc O(x^{-1})$ as $x\to \infty$.
Considering the leading coefficients of these double expansions for each of the three summands in \Ref{laxr} we obtain the equation
$0 =\zeta + 1 - 1$, which is impossible.  The lemma is proved.
\end{proof}

The $n$-independent polynomial $q_n(z)$ will be denoted by $q(z)$.  Let $\ka$ be the degree of $q(z)$.

\vsk.2>
Introduce  new functions
\beq\label{Psied}
\Psi_n(x,z):=\frac{q(z)}{z^{\kappa}}\psi_n(x,z) = z^n (1+z)^x \frac{q(z)}{z^{\kappa}}
\left(1+\sum_{i=1}^{k_n}\frac{C_i^{(n)}\!(z)}{x-u_i^{(n)}}\right).
\eeq
Clearly the sequence $(\Psi_n(x,z))$ is an admissible solution of \Ref{laxdd} and
\bean
\label{qass}
\frac{q(z)}{z^{\kappa}}
\left(1+\sum_{i=1}^{k_n}\frac{C_i^{(n)}\!(z)}{x-u_i^{(n)}}\right)
=1+\sum_{i=1}^{k_{n}} \xi^{(n)}_{i}\!(x) z^{-i},
\eean
where $\xi^{(n)}_i\!(x)$ are rational functions of $x$ with at most first order poles  at
the zeros of $y_{n}(x)$. Thus the sequence of functions  $(\Psi_n(x,z))$ has the properties listed in
Theorem \ref{gener1}. Theorem \ref{gener1} is proved.
\end{proof}

\subsection{Example}

Consider the sequence $ y^\emptyset=(y_n(x))_{n\in\Z}$, where  $y_n(x)=1$ for all $n$, see \Ref{y empty}.
 As discussed in Section \ref{Elg},
this sequence represents a solution of the Bethe ansatz equations \Ref{baei} with $k_n=0$ for all $n$.
In this case,
the generating linear problem
equations \Ref{laxdd} take the form
\bean
\label{la em}
\psi_{n+1}(x)=\psi_n(x+1)-\psi_n(x), \qquad n\in\Z,
\eean
and the admissible solution $ \Psi^\emptyset(z) = (\Psi_n^\emptyset(x,z))_{n\in\Z}$ of Theorem
\ref{gener1} is
\bean
\label{Psi em}
\Psi_n^\emptyset(x,z) = z^n (1+z)^x, \qquad n \in\Z.
\eean

\subsection{Solutions $\Psi(z)$ and  the
operation of generation}

Let $ y=(y_n(x))_{n\in\Z}$ be a generic sequence of polynomials, which represents a solution
of the Bethe ansatz equations \Ref{baei}. Then there exists a unique one-parameter family
$ {\Psi}(z) = (\Psi_n(x,z))$  of solutions of the generating linear problem
equations \Ref{laxdd} given by Theorem \ref{gener1}.

\vsk.2>

Choose $m\in \Z$. Consider the one-parameter family
$ {y}^{(m)}(c)=(\tilde y_n(x,c))_{n\in\Z}$, obtained from $ y$  by generation in the
$m$-th direction, see \Ref{sml}. Here $\tilde y_n(x,c) = y_n(x)$ for $n\ne m$ and the polynomial
$\tilde y_m(x,c)$
satisfies the equation
\bean
\label{fert1}
\tilde y_m(x,c)y_m(x+1) - \tilde y_m(x+1,c)y_m(x)) = y_{m-1}(x+1)\,y_{m+1}(x)\,.
\eean

Choose the value $c=c_0$ so that the sequence $ {y}^{(m)}(c_0)$ is generic.
 Then
$ {y}^{(m)}(c_0) $ represents a solution of the Bethe ansatz equations \Ref{baei} by Theorem \ref{f cor}.
Define the sequence ${\tilde y}=(\tilde y_n(x))_{n\in\Z}$ by the formula  ${\tilde y}=  y^{(m)}(c_0)$.
Denote  $\tilde k_n=\deg \tilde y_n(x)$ for $n\in\Z$.

\vsk.2>
Starting from   $ {\tilde y} $ define a sequence of rational functions $ {\tilde v} = (\tilde v_n(x))$
by formula \Ref{vpot}.  We have $\tilde v_n(x) = v_n(x)$ if $n\ne m-1, m$ and
\bean
\label{ti v}
\tilde v_{m-1}(x)
&=&
\frac{y_{m-1}(x)\tilde y_{m}(x+1)}{y_{m-1}(x+1) \tilde y_{m}(x)}\,,
\\
\notag
\tilde v_m(x)
&=&
\frac{\tilde y_{m}(x)y_{m+1}(x+1)}{\tilde y_m(x+1) y_{m+1}(x)}\,.
\eean
Apply  Theorem \ref{gener1} to the sequence ${\tilde v}$ and obtain the unique one-parameter family
$ {\tilde\Psi}(z) = (\tilde \Psi_n(x,z))$ of admissible solutions of the generating  linear problem equation
\Ref{laxdd} with the chosen sequence $ {\tilde v}$,
\beq
\label{bati}
\tilde\Psi_n(x,z)=z^{n} (1+z)^x\left(1+\sum_{i=1}^{\tilde k_{n}} \tilde \xi^{(n)}_{i}(x) z^{-i}\right),
\eeq
where $\tilde \xi^{(n)}_i(x)$ are rational functions in $x$ with at most first order poles  at
the zeros of $\tilde y_{n}(x)$.

\begin{thm}
\label{thm gen BA}

We have $ \tilde\Psi_n(x,z) = \Psi_n(x,z)$ for $n\ne m$ and
\bean
\label{BA n=m}
\tilde\Psi_m(x,z)  =  \Psi_m(x,z)  + g(x)
\Psi_{m-1}(x,z)\,,
\eean
where
\bean
\label{g-form}
g(x) = \frac{y_{m-1}(x) y_{m+1}(x)}{y_{m}(x)\tilde y_{m}(x)}\,.
\eean

\end{thm}

\begin{proof}

\begin{lem}
\label{g-v}
We have
\bean
\label{1gv}
\tilde v_m(x) g(x) = v_{m-1}(x) g(x+1)\,,
\\
\label{2gv}
\tilde v_m(x) - v_{m}(x)= g(x+1)\,.
\eean

\end{lem}

\begin{rem}
Equations \Ref{1gv} and \Ref{2gv} imply the equation
\bean
\label{Ric}
v_m(x) g(x) - v_{m-1}(x) g(x+1)+ g(x) g(x+1) = 0\,.
\eean
This equation with respect to  $g(x)$
is called the  {\it discrete Riccati equation}, see \cite{MV3}.
This discrete Riccati equation has a rational solution
$g(x)$, given by \Ref{g-form}. On discrete Riccati equations with rational solutions
see \cite{MV3}.

\end{rem}

\begin{proof}
The proof of \Ref{1gv} is straightforward. We also have
\bea
&&
\tilde v_m(x) - v_{m}(x) =
\frac{\tilde y_{m}(x)y_{m+1}(x+1)}{\tilde y_m(x+1) y_{m+1}(x)}
- \frac{ y_{m}(x)y_{m+1}(x+1)}{y_m(x+1) y_{m+1}(x)}
\\
&&
\phantom{aaa}
=\,\frac{y_{m+1}(x+1)}{y_{m+1}(x)}\,
\frac{\tilde y_m(x)y_m(x+1)- \tilde y_m(x+1)y_m(x)}
{\tilde y_m(x+1)\tilde y_m(x+1)}
\\
&&
\phantom{aaa}
=\,\frac{y_{m+1}(x+1)}{y_{m+1}(x)}\,
\frac{y_{m-1}(x+1)y_{m+1}(x)}
{\tilde y_m(x+1)\tilde y_m(x+1)} = g(x+1)\,.
\eea
\end{proof}

Let us check that the functions $\Psi_{m+1}(x,z)$, $ \Psi_m(x,z)  + g(x)\Psi_{m-1}(x,z)$,
$\Psi_{m-1}(x,z)$ satisfy equations \Ref{laxdd} with $\tilde v_{m-1}(x), \tilde v_m(x)$.
Indeed, we have
\bea
\Psi_{m}(x,z) +g(x)\Psi_{m-1}(x,z) = \Psi_{m-1}(x+1) - \tilde v_{m-1}(x) \Psi_{m-1}(x,z)
\eea
by formula \Ref{2gv} and
\bea
\Psi_{m+1}(x,z) = \Psi_{m}(x+1,z) +g(x+1)\Psi_{m-1}(x+1,z)
 - \tilde v_{m}(x) \big(\Psi_{m}(x,z) +g(x)\Psi_{m-1}(x,z)\big)\,.
\eea
by formulas \Ref{1gv} and \Ref{2gv}.
\vsk.2>

\begin{lem}
\label{lem Pform}
We have
\bean
\label{titi m}
 \Psi_m(x,z)  + g(x) \Psi_{m-1}(x,z)
=z^{m} (1+z)^x\left(1+\sum_{i>0}\, \tilde r_{i}(x) z^{-i}\right),
\eean
where $\tilde r_i(x)$ are rational functions of $x$ with at most first order poles  at
the zeros of $\tilde y_{m}(x)$.

\end{lem}

\begin{proof}
It is enough to show that the left-hand side in \Ref{titi m}
is regular at the roots
of the polynomial $y_m(x)$. Indeed,
\bea
&&
 \Psi_m(x,z)  + g(x)\Psi_{m-1}(x,z)
\\
&&
\phantom{aa}
=  \Psi_{m-1}(x+1,z) -  v_{m-1}(x)\Psi_{m-1}(x,z)  + g(x)\Psi_{m-1}(x,z)
\\
&&
\phantom{aa}
=
\Psi_{m-1}(x+1,z) -  (v_{m-1}(x)-g(x))\Psi_{m-1}(x,z)
\\
&&
\phantom{aa}
=
\Psi_{m-1}(x+1,z) -  \frac {g(x)}{g(x+1)} v_m(x)\Psi_{m-1}(x,z)
\\
&&
\phantom{aa}
=
\Psi_{m-1}(x+1,z) -  \frac {y_{m-1}(x)\tilde y_m(x+1)}
{y_{m-1}(x+1)\tilde y_m(x)}\Psi_{m-1}(x,z)\,,
\eea
and the last expression is regular at the roots of $y_m(x)$.
\end{proof}
Theorem \ref{thm gen BA} is proved.
\end{proof}

\begin{rem}
Let $ y=(y_n(x))_{n\in\Z}$ be a generic $N$-periodic sequence of polynomials representing a solution of the Bethe ansatz equations \Ref{bae}. Let $ \Psi(z) = (\Psi_n(x,z))_{n\in\Z}$ be the associated one-parameter family of
admissible solutions
determined by Theorem  \ref{gener1}.
By Theorem \ref{one gen} the sequence $ y=(y_n(x))_{n\in\Z}$ can be obtained from the sequence $ y^\emptyset$ by the iterated
generation procedure of Section \ref{sec GEN}. Theorem \ref{thm gen BA} shows how to obtain
the family of  admissible solutions
$ \Psi(z)$ from the family of admissible solutions $ \Psi^\emptyset(z)$ in \Ref{Psi em}
by transformations  of Theorem \ref{thm gen BA}.

\end{rem}

\section{Spectral transforms for the rational RS system}
\label{S:direct}

\subsection{Lax matrices}
In Section \ref{S:gen} for any sequence of polynomials $(y_n(x))_{n\in\Z}$, whose roots satisfy
the Bethe ansatz equations \Ref{baei}, we constructed solutions $(\psi_n(x,z))_{n\in\Z}$ of the generating linear problem
equation \Ref{laxdd} depending on the spectral parameter $z$. Formulas \Ref{dL}, \Ref{jan29a} of that construction
reveal a'priory unexpected connections of the construction with the theory of the rational RS system.
In this section we develop  the direct and  inverse {\it spectral transforms} for the rational RS system.

\vsk.2>
We identify the phase space of the $k$-particle  rational RS system with the subspace
 ${\mathcal P}_k\subset \C^k\times (\C^\times)^k$ of pairs of vectors $u=(u_1,\ldots,u_k)$
and $\ga=(\ga_1,\ldots, \ga_k)$, such that
\beq\label{RSphase}
\,u_i\neq u_j,\quad u_i\neq u_j+1 \ \on{for}\ i\ne j\,.
\eeq
A point $(u,\ga)\in\mc P_k$  defines the  $k\times k$ {\it Lax matrix}  $L(u,\ga)$,
\beq
\label{dLRS}
L_{ij}(u,\ga) =\frac {\ga_i}{u_i-u_j-1}\,, \qquad i,j=1,\dots,k\, .
\eeq

Notice that the Lax matrix has already appeared in \Ref{dL},
where
\bea
L^{(n)}(z)=1+z- L(u^{(n)},\ga^{(n)}).
\eea

The matrix $L(u,\ga)$ is a particular case of the {\it Cauchy matrix}. Its determinant equals
\beq\label{Ldet}
\det L(u,\ga)\,=\,\prod_{i=1}^k \ga_i\prod_{i<j} \frac {(u_i-u_j)^2}{(u_i-u_j)^2-1}\,.
\eeq
It satisfies, the so-called {\it displacement equation}
\beq\label{UL}
[U,\,L(u,\ga)]=L(u,\ga)+\Ga F\,,
\eeq
where $U=\diag(u_1,\dots,u_k)$,  $\Gamma=\diag(\ga_1,\dots,\ga_k)$,
$F=(f_{ij})$ with $f_{ij}=1$ for all $i,j$. Equation \Ref{UL} can be easily checked directly.

\vsk.2>
Let $E$ be the $k\times k$ unit matrix. Denote
\bean
\label{Lnz}
L(z\,| \,u,\ga): =(1+z) E - L(u,\ga)\,.
\eean
Let $\hat L(x,z\,|\,u,\ga)$ be the $(k+1)\times(k+1)$-matrix, whose rows and columns are
labeled by indices $0,\dots,k$ and entries are given by the formulas:
\bean
\label{hatLRS}
&&
\widehat L_{0,0}=1\,,\qquad  \widehat L_{0,j}=\frac 1{x-u_j}\,,
\qquad \widehat L_{i,0}=-\ga_i\,,
\\
\notag
&&
\phantom{aaaaaaa}
\widehat L_{i,j}\,=\, L_{i,j}(z\,|\,u,\ga)\,, \qquad i,j=1,\dots, k\,.
\eean
cf. formulas \Ref{hatLd}.
Define  the  function  $\psi(x,z\,|\,u,\ga)$ by the formula
\beq\label{psi0}
\psi(x,z\,|\,u,\ga)=(1+z)^x\,\frac{\det \hat L(x,z\,|\,u,\ga)}{\det L(z\,|\,u,\ga)}\,.
\eeq

\subsection{Direct transform in generic case}
We define the direct spectral transform first for points $(u,\ga)$
of the following open subset ${\mathcal P}_k'\subset \mathcal P_k$.

\vsk.2>
Let $\mu=(\mu_1, \ldots, \mu_k)$ be the set of eigenvalues of the matrix $L(u,\ga)$.
We have $\mu_j\neq 0$ for all $j$ by formula \Ref{Ldet}. Hence $\mu\in (\C^\times)^k$.

\vsk.2>
Define
\bean
\label{def P}
\mathcal P_k' = \{ (u,\ga)\in\mc P_k\ |\ \mu_1,\dots,\mu_k\ \on{are\, distinct}\}.
\eean
Clearly $\mc P_k\rq{}$ is nonempty,
since for big distinct $u_1,\dots,u_k$ the matrix $L(u,\ga)$ is close to the
diagonal matrix $-\on{diag}(\ga_1,\dots,\ga_k)$.

\vsk.2>

The function $\psi(x,z\,|\,u,\ga)$ has at most simple pole at $z=\mu_j-1$.
Consider the Laurent expansion of $\psi(x,z\,|\,u,\ga)$ at $z=\mu_j-1$,
\beq
\label{laurentexpansion}
\psi(x,z\,|\,u,\ga)=\frac{\phi_{j}^{(0)}(x\,|\,u,\ga)}{z-\mu_j+1}+\phi_{j}^{(1)}(x\,|\,u,\ga)+\mc O(z-\mu_j+1)\,.
\eeq

\begin{thm}
\label{lm:may5}
For $(u,\ga)\in \mathcal P_k'$, there exists a unique $a=(a_1,\dots,a_k)\in\C^k$ such that
\beq\label{alpha}
\phi_{j}^{(1)}(x\,|\,u,\ga)+a_j\phi_{j}^{(0)}(x\,|\,u,\ga)=0\,,
\qquad j=1,\dots,k\,.
\eeq
\end{thm}

\begin{proof}  The function $\psi(x,z\,|\,u,\ga)$ has the form
\beq
\label{psidd}
\psi(x,z\,|\,u,\ga) =(1+z)^x\left(1+\sum_{i=1}^{k}\frac{C_i (z)}{x-u_i}\right),
\eeq
cf.  \Ref{psid}. The vector $C(z)$ with  coordinates $C_i(z)$
is given by \Ref{kramerd}.  The vector $C(z)$
solves  equation \Ref{jan29a}. Consider the Laurent expansion of $C(z)$ at $z=\mu_j-1$,
\beq
\label{Cexpansion}
C(z)=\frac{c_{j}}{z-\mu_j+1}+d_j+\mc O(z-\mu_j+1)\,,
\eeq
where $c_j, d_j$ are $k$-vectors with coordinates denoted by $c_{ij}, d_{ij}$, respectively.
The substitution of \Ref{Cexpansion} into \Ref{jan29a} gives  the relations:
\bean
\label{mar27a}
(\mu_j-L) \,c_j\, &=&\,0\,,
\\
\label{mar27b}
(\mu_j-L) \,d_j+\,c_j\,&=&\,\ga\,,
\eean
where $L=L(u,\ga)$.
\vsk.2>

Let $\tilde c_j$ be a nonzero eigenvector of $L$ with  eigenvalue $\mu_j$.
It is unique up to multiplication by a nonzero constant. Using  \Ref{UL} we get
\beq\label{babj}
(\mu_j-L) U \tilde c_j\,=\,[U,L]\tilde c_j\,=\,(L+\Ga F)\tilde c_j\,=\,\mu_j \tilde c_j+\nu_j \ga \,,
\eeq
where $\nu_j :=\sum_{i=1}^{k}\tilde c_{ij}$. We have $\nu_j\neq 0$. Indeed, if $\nu_j=0$, then \Ref{babj} shows
that  $L$ has a nontrivial Jordan block with eigenvalue $\mu_j$. That contradicts to the assumption that $(\mu_1,\dots,\mu_k)$ are distinct. Since $\nu_j\neq 0$. We can uniquely define the vector $\tilde c_j$ by the
normalization $\nu_j=-\mu_j$.

\begin{lem}
\label{c-non-zero}
The vector $c_j$ defined in \Ref{Cexpansion} is nonzero.
\end{lem}
\begin{proof} If $c_j=0$, then  \Ref{mar27b} gives
\bea
(\mu_j-L) \,d_j=\ga.
\eea
Formula \Ref{babj} with $\nu_j=-\mu_j$ gives
\bea
\mu_j^{-1}(\mu_j-L) U \tilde c_j\,=\,
\tilde c_j - \ga \,.
\eea
Adding the two formula gives
\beq\label{may8}
(\mu_j-L)(\mu_j^{-1} U\tilde c_j+d_j)=\tilde c_j\,,
\eeq
which means that $L$ has a nontrivial Jordan block with eigenvalue $\mu_j$. Contradiction.
\end{proof}

\begin{lem}
\label{lem u&e}
Let $\tilde c_j=(\tilde c_{ij})$,\, $ \tilde d_j$ $\in\C^k$ be a solution of the system of equations
\bean
\label{ma}
(\mu_j-L) \,\tilde c_j\, &=&\,0\,,
\\
\label{mb}
(\mu_j-L) \,\tilde d_j+\,\tilde c_j\,&=&\,\ga\,,
\eean
such that $\tilde c_j\ne 0$. Then
\bea
\sum_{i=1}^k\tilde c_{ij} = -\mu_j, \qquad \tilde d_j=-\mu_j^{-1} U\tilde c_j-a_j \tilde c_j\,,
\eea
for some $a_j\in\C$.
\end{lem}

\begin{cor}
The vectors $c_j, d_j$ in \Ref{Cexpansion} satisfy the equations
\bean
\label{7.14}
\sum_{i=1}^k c_{ij} = -\mu_j, \qquad d_j=-\mu_j^{-1} U c_j-a_j c_j\,,
\eean
for some $a_j\in\C$.

\end{cor}

\smallskip
\noindent
{\it Proof of Lemma \ref{lem u&e}.} The vector $\tilde c_j$ is an eigenvector of $L$ with eigenvalue $\mu_j$.
Fix $\tilde c_j$  by the condition $\sum_{i=1}^{k}\tilde c_{ij}=-\mu_j$. Then \Ref{babj} and \Ref{mar27b}
show that $\tilde c_j$ and
 $\tilde d_j=-\mu_j^{-1}U\tilde c_j$  give a  solution to the system of equations \Ref{ma} and
\Ref{mb}.  For that $\tilde c_j$ the general solution of \Ref{mb} has the form
\beq
\label{D}
\tilde d_j=-\mu_j^{-1} U\tilde c_j-a_j \tilde c_j\,,
\eeq
where $a_j$ is an arbitrary constant.

Let $(\tilde c_j, \tilde d_j)$ and $(\hat c_j,\hat d_j)$ be two solutions of system \Ref{ma}, \Ref{mb}. Then
\bea
(\mu_j-L) \,(\tilde c_j-\hat c_j)\, &=&\,0\,,
\\
(\mu_j-L) \,(\tilde d_j-\hat d_j)+\,(\tilde c_j-\hat c_j)\,&=&\,0\,.
\eea
If $\tilde c_j-\hat c_j\ne 0$, then $L$ has a nontrivial Jordan block with eigenvalue $\mu_j$. This leads to contradiction.
Hence $\tilde c_j =\hat c_j$.
The lemma is proved.
\qed

\vsk.3>
By formula \Ref{psid} the first two coefficients of the Laurent expansion of $\psi(x,z\,|\,u,\ga)$ at $z=\mu_j-1$ are
\bean
\label{psij-1}
\phi_{j}^{(0)}(x\,|\,u,\ga)
&=&
\mu_j^x\left(\sum_{i=1}^{k} \frac {c_{ij}}{x-u_i}\right)
\\
\label{psij0}
\phi_{j}^{(1)}(x\,|\,u,\ga)
&=&
\mu_j^x\left(1+\sum_{i=1}^{k} \frac {x\mu_j^{-1}c_{ij}+d_{ij}}{x-u_i}\right)\,.
\eean
Using \Ref{7.14} we get
\bean
\label{psij01}
\phi_{j}^{(1)}(x\,|\,u,\ga)
&=&
\mu_j^x\left(1+\sum_{i=1}^{k} \frac{(\mu_j^{-1}(x-u_i)-a_j) c_{ij}}{x-u_i}\right)
\\
\notag
&=&\mu_j^x\left(1+\sum_{i=1}^{k} \Big(\mu_j^{-1}c_{ij}-a_j\frac{c_{ij}}{x-u_i}\Big)\right)=\,-\,a_j\phi_{j}^{(0)}(x\,|\,u,\ga)\,.\
\eean
The theorem is proved.
\end{proof}

Theorem \ref{lm:may5} gives us the correspondence
\beq
\label{corr27}
S\ :\ \
 (u,\ga)\ \ \mapsto\  \ (\mu,a)\,,
\eeq
where $(u,\ga)\in {\mathcal P}_k' \subset \C^k\times (\C^\times)^k$ and
$(\mu, a)\in  (\C^\times)^k\times \C^k$.

\vsk.3>
Below we will need the following  stronger version of Lemma \ref{c-non-zero}.

\begin{lem}
\label{nonzero-stronger}
Let $\mu_j$ be  an eigenvalue of $L(u,\ga)$ (of any multiplicity).  Then the function $\psi(x,z\,|\,u,\ga)$ is not holomorphic at $z=\mu_j-1$.
\end{lem}

\begin{proof}
The function $\psi(x,z\,|\,u,\ga)$ has the form \Ref{psidd} with the vector $C(z)$ that solves
equation \Ref{jan29a}. If $\psi(x,z\,|\,u,\ga)$ is holomorphic at $z=\mu_j-1$,  then
the vector $C(z)$ is holomorphic at $z=\mu_j-1$ as well. Let $d_j:=C(\mu_j-1)$,
 then \Ref{jan29a} gives  the relation:
\beq
\label{may9}
(\mu_j-L) \,d_{j}=\ga\,,
\eeq
where $L=L(u,\ga)$.

Let $c_{j,s},\, s=1,\dots, \ell$, be a Jordan chain of
the operator $L=L(u,\ga)$ of maximal length $\ell$ with eigenvalue $\mu_j$.
Let $c_{j,s,i}$ be   coordinates of the vector   $c_{j,s}$.
Then
\beq
\label{Jchain}
(\mu_j-L)\,c_{j,1}=0\,, \qquad   (\mu_j-L)c_{j,s}=c_{j,s-1}\,.
\eeq
Using the displacement equation \Ref{UL} we get
\beq\label{may10}
(\mu_j-L) U c_{j,\ell}\,=\mu_j c_{j,l}-c_{j,\ell-1}+\nu_{j,l}\ga\, ,
\eeq
where $\nu_{j,\ell}:=\sum_i c_{j,\ell,i}$
Using  equations \ref{may9}, \Ref{Jchain} with $s=\ell$ and \Ref{may10}  we get
\beq\label{may8a}
(\mu_j-L)(U c_{j,\ell}+c_{j,\ell}-\nu_{j,\ell}\,d_j)=\mu_j c_{j,\ell}\,,
\eeq
which contradicts to the assumption that the Jordan chain is of maximal length.
\end{proof}

\subsection{Inverse correspondence}
\label{sec 7.3}

We recall the construction of the correspondence inverse to \Ref{corr27},
cf. the construction in \cite{K3}.
We define it simultaneously with the construction of generic solutions to the
rational RS system.

\vsk.2>
Let $\Om(x,t,z)$ be the function in  $x, z$, and $t=(t_1,t_2,\dots)$  defined by the formula
\bean
\label{Om0}
\Om(x,t,z)\,=\, (1+z)^xe^{\sum_{j=1}^\infty t_jz^j},
\eean
in which we always  assume that only a finite number of the variables $t_j$ are nonzero.
The function $\Omega(x,t,z)$  in more details is considered in Section \ref{subsection:om}.

\vsk.2>
Let $\mu \in (\C^\times)^k$ with $\mu_i\ne\mu_j$ for $i\ne j$. Let
$\psi(x,t,z)$ be a function of the form
\beq
\label{psiRS}
\psi(x,t,z)=\Omega(x,t,z)\left(1+\sum_{j=1}^k\frac{r_j(x,t)}{z+1-\mu_j}\right).
\eeq
Consider the  Laurent expansions
\beq
\label{lex}
\psi(x,t,z)=\frac{\phi_{j}^{(0)}(x,t)}{z-\mu_j+1}+\phi_{j}^{(1)}(x,t)+\mc O(z-\mu_j+1)\,,
\qquad
j=1,\dots,k\,.
\eeq
\begin{lem}\label{7.5}
If $(\mu,a) \in (\C^\times)^k\times \C^k$ with $\mu_i\ne\mu_j$ for $i\ne j$,
then there is a unique function
 $\psi(x,t,z)$ as in \Ref{psiRS}, such that  coefficients
$\phi_{j}^{(0)}(x,t)$, $\phi_{j}^{(1)}(x,t)$ satisfy the equations
\beq\label{alph}
\phi_{j}^{(1)}(x,t)+a_j\phi_{j}^{(0)}(x,t)=0\,,
\qquad j=1,\dots,k\,.
\eeq

\end{lem}

Notice that the form of the second factor in the right-hand side of \Ref{psiRS} is just the simple fraction
decomposition of a rational function in $z$ with at most simple poles at
the points  $z=\mu_i-1$ that equals $1$ at $z=\infty$.

\begin{proof} The lemma is proved by explicit computation of $\psi(x,t,z)$. Let $r(x,t)$
be the $k$-vector with coordinates $r_i(x,t)$.
Taking the first coefficient of the Laurent expansion of $\psi(x,t,z)$
 at $z=\mu_j-1$ shows that equations \Ref{alph} are equivalent to the inhomogeneous equation
\beq\label{T}
T(x,t)\, r(x,t)\,= \,- e_0\,,
\eeq
where $e_0$ is the $k$-vector with all coordinates equal to $1$ and
$T=T(x,t)$ is the $k\times k$-matrix with entries
\beq\label{eqforr}
T_{ii}= a_i+\left(x\mu_i^{-1}+\sum_s st_s(\mu_i-1)^{s-1}\right),\qquad T_{ij}=\frac{1}{\mu_i-\mu_j}, \qquad
 i\neq j\,.
\eeq
Let $\hat T(x,t,z)$ be the $(k+1)\times(k+1)$-matrix with the entries
\beq\label{hatT}
\hat T_{00}=1,\quad \hat T_{0,j}=\frac 1{z+1-\mu_j}, \quad \hat T_{0,i}=1, \quad \hat T_{ij}=T_{ij}, \quad i,j=1,\ldots k\,.
\eeq
Then  $\psi(x,t,z)$  equals
\beq\label{psidet}
\psi(x,t,z)=\Omega(x,t,z)\,\frac{\det\hat T(x,t,z)}{y(x,t)}\,,
\eeq
where
\beq\label{ydet}
y(x,t)=\det  T(x,t)\,.
\eeq
\end{proof}

The function $y(x,t)$ will also be denoted by $y(x,t\,|\,\mu,a)$.
It is a polynomial in $x$ of degree $k$. Let $u_i(t\,|\,\mu,a)$, $i=1,\dots,k$, be its roots.
Define $\ga(t\,|\,\mu,a)=(\ga_1(t\,|\,\mu,a),\dots,\ga_k(t\,|\,\mu,a))$ by the formula
\bean
\label{gaga}
\ga_i(t\,|\,\mu,a)\,=\,\p_{t_1}u_i(t\,|\,\mu,a)\,.
\eean

\vsk.2>

Let ${\mathcal S}\subset (\C^\times)^k\times \C^k$
be the subset of points $(\mu,a)$, such that
\begin{enumerate}
\item[(a)]  $\mu=(\mu_1,\dots,\mu_k)$ has distinct coordinates;

\item[(b)]   $u(0\,|\,\mu,a)=(u_1(0\,|\,\mu,a),\dots,u_k(0\,|\,\mu,a))$ has distinct coordinates.

\end{enumerate}

\begin{thm} \label{tm:5.7} For $(\mu,a)\in {\mathcal S}$, the map
\beq
\label{gainverse}
\tilde S\ :\ (\mu,a)\ \mapsto \  (u(0\,|\,\mu,a), \ga(0\,|\,\mu,a))
\eeq
is inverse to the map in \Ref{corr27}.
\end{thm}
\begin{proof} The standard arguments based on the uniqueness of the Baker-Akhiezer function prove the following statement.
\begin{lem} \label{inversesimple} The function $\psi(x,t,z)$ given by \Ref{psidet} satisfies equation \Ref{intRS}
with $y(x,t)$ defined in \Ref{ydet}.
\end{lem}
\begin{proof} Define the function $w(x,t)$ by the formula
\beq\label{wxi}
w(x,t)=\xi_1(x,t)-\xi_1(x+1,t)-1\,,
\eeq
where $\xi_1(x,t)$ is the coefficient of the expansion of the second factor in \Ref{psidet} at $z=\infty$, i.e.
\beq\label{psiinf}
\psi(x,t,z)=\Omega(x,t,z)\left(1+\sum_{s=1}^\infty \xi_s(x,t)z^{-s}\right).
\eeq
Then  the corresponding expansion for the function
\bea
\tilde \psi(x,t,z):=\p_{t_1}\psi(x,t,z)-\psi(x+1,t,z)-w(x,t)\psi(x,t,z)
\eea
 has the form $\tilde \psi(x,t,z)=\Omega(x,t,z) \,\mc O(z^{-1})$, i.e. the simple fraction expansion for $\tilde \psi$ has the form
\beq
\label{psiRS-1}
\tilde \psi(x,t,z)=\Omega(x,t,z)\left(\sum_{j=1}^k\frac{\tilde r_j(x,t)}{z+1-\mu_j}\right).
\eeq
Since $a_j$ in \Ref{alpha} is a constant, the first two coefficients of the Laurent expansion of $\tilde \psi$ at $\mu_j-1$ satisfy  equation \Ref{alpha}, i.e. for the vector $\tilde r$ with coordinates $\tilde r_j$
the homogeneous linear equation $T\tilde r=0$ holds. Hence, $\tilde r=0$ and the equation $\tilde \psi=0$ is proved.

It remains to show that the function $w(x,t)$, defined by  \Ref{wxi}, has the form \Ref{wRS} with $y(x,t)$ given by \Ref{ydet}.
The equation
\beq\label{xir1}
\xi_1(x,t)=-\p_{t_1} \ln y(x,t)
\eeq
can be derived from the Cramer's formulas for the coordinates $r_j$ of the vector $r$ and the equation
\beq\label{xir}
\xi_1(x,t):=\sum_{j=1}^{k}r_j(x,t)\,.
\eeq
It is more instructive to prove it directly using equation \Ref{intRS}.
 Indeed, by definition,  $\xi_1(x,t)$ is a rational function in $x$
 with poles at the zeros of $y(x,t)$. The comparison of the coefficients
 at $(x-u_i)^{-2}$ of the Laurent expansion of the right and left-hand sides of \Ref{intRS} at $u_i$ gives the equation
\beq\label{gaV}
\ga_i(t):=\res_{x=u_i}w(x,t) =\res_{x=u_i}\xi_1(x,t)=\p_{t_1} u_i(t)\,.
\eeq
The latter implies \Ref{xir1}.
\end{proof}

The left-hand side of  \Ref{intRS} has poles only at the
zeros of $y(x,t)$. Hence the right-hand side of
\Ref{intRS} has no residue at $x=u_i-1$. From \Ref{wxi} it follows that
the residue of $w(x,t)$ at $x=u_i-1$ equals $-\ga_i(t)$ and we recover
the defining condition for $\psi(x,t,z)$ in Lemma \ref{lem 4.3}. Put $t=0$.  The theorem is proved.
\end{proof}

\subsection{Extension of the direct spectral transform}
Our  goal is to extend the direct spectral transform \Ref{corr27}
to the whole phase space of the rational RS system.

\vsk.2>
For $(u,\ga)\in\mc P_k$ consider the function $\psi(x,z\,|\,u,\ga)$ defined by \Ref{psi0}.
The function
\beq\label{psibig}
\Psi(x,z\,|\,u,\ga)= \det L(z\,|\,u,\ga)\,\psi(x,z\,|\,u,\ga)=(z+1)^x \det \hat L(x,z\,|\,u,\ga)
\eeq
has the form
\beq\label{Psi14}
\Psi(x,z\,|\,u,\ga)=(1+z)^x\left(z^k+\sum_{\ell=1}^{k}\xi_\ell(x\,|\,u,\ga) z^{k-\ell}\right).
\eeq
It is well-defined on the whole phase space.
The coefficients $\xi_\ell(x\,|\,u,\ga)$ are rational function   in $u, \ga$
  holomorphic on $\mathcal P_k$.

\vsk.2>
Let $(\mu_i=\mu_i(u,\ga))_{i=1}^q$
be the set of all  distinct eigenvalues of $ L(u,\ga)$
with respective multiplicities  $(m_i)_{i=1}^q$. We have
$\sum_{i=1}^qm_i=k$ and
\bea
\det L(z\,|\,u,\ga)=\prod_{i=1}^{q}(z-\mu_i+1)^{m_i}\, , \qquad \mu_i\neq \mu_j\,.
\eea

For a positive integer $\ell$ denote by $\C[z]_\ell$ the vector subspace of $\C[z]$ of polynomials of degree less than $\ell$.
We have $\dim \C[z]_\ell = \ell$.

Let $\mu\in\C$. We will often identify $\C[z]_\ell$  with $\C[z]/\langle (z-\mu+1)^\ell\rangle$
under the isomorphism $g(z) \mapsto g(z) + \langle (z-\mu+1)^\ell\rangle$.

\vsk.2>

\begin{thm}\label{thm:14}
Let $(u,\ga)\in\mc P_k$. Then for  $j=1,\dots,q$, there is a unique $m_j$-dimensional vector subspace $W_j(u,\ga)\subset \C[z]_{2m_j}$
 such that
\beq
\label{orthogonality}
\res_{z=\mu_j-1}\,\frac {g(z)\Psi(x,z\,|\,u,\ga)}{(z-\mu_j+1)^{2m_j}}\,=\,0\,, \qquad \forall \,g(x)\in W_j(u,\ga)\,.
\eeq
\end{thm}

\begin{rem}
Let $(u,\ga)\in {\mathcal P_k'}$. Then
the one-dimensional subspace $W_j(u,\ga)\subset \C[z]_2$ is spanned by the polynomial
$a_j (z-\mu_j+1)+1$. Then  equations \Ref{orthogonality}
  take the form
\beq\label{orthogonality0}
\res_{z=\mu_j-1}\frac {(a_j (z-\mu_j+1)+1)\Psi(x,z\,|\,u,\ga)}{(z-\mu_j+1)^{2}}=0,
\qquad j=1,\dots,k,
\eeq
which is the same as equations \Ref{alpha}.

\end{rem}

\begin{proof}

The coefficients of $\det L(z\,|\,u,\ga)$ are holomorphic functions on $\mathcal P_k$. Hence for any $(u',\ga')\in {\mathcal P_k'}$ in sufficiently  small neighborhood of $(u,\ga)$ the multiple eigenvalue $\mu_j$ of $L(u,\ga)$ splits into a set of simple eigenvalues of the matrix $L(u',\ga')$, i.e.
\bea
\det L(z\,|\,u',\ga')=\prod_{i=1}^{q}\prod_{s=1}^{m_j}(z-\mu_{i,s}+1)\,
\eea
where $|\mu_{j,s}-\mu_j|<\varepsilon$ for some small   $\varepsilon$.
We may assume that the $\varepsilon$-neighborhoods of $\mu_j$, $j=1,\dots,q$, do not intersect.

The set of $m_j$ equations \Ref{orthogonality0},
corresponding to a subset of the eigenvalues $\mu_{j,s}$, can be represented in the form
\beq\label{orthogonality1}
\oint_{c_j}\frac{g_{j,s}(z)\,\Psi(x,z\,|\,u\rq{},\ga\rq{})}{\prod_s (z-\mu_{j,s}+1)^2}\,dz\,,\qquad s=1,\ldots m_j\,,
\eeq
where $c_j$ is the circle $|z-\mu_j+1|=\varepsilon$ and
\beq\label{g-js}
g_{j,s}(z):=(a_j(z-\mu_{j,s}+1)+1)\prod_{\ell\neq s}(z-\mu_{j,\ell}+1)^2\,.
\eeq
It is easy to see that the polynomials $g_{j,s}(z)$ are linear independent and hence span
an $m_j$-dimensional subspace $W_j(u\rq{},\ga\rq{})$ of   $\C[z]_{2m_j}$,
i.e. $W_j(u',\ga')$ can be seen as a point of the Grassmanian ${\rm Gr}(m_j,2m_j)$.

The Grassmanian is compact.
Therefore, for any sequence $((u^{m},\ga^{m}))_{m=1}^\infty\subset {\mathcal P_k'}$ converging to $(u,\ga)$
there is a subsequence of points $W_j(u^{m},\ga^m)$ of the Grassmanian converging
 to some point
$W_j\in {\rm Gr}(m_j,2m_j)$. Since the integral in \Ref{orthogonality1} is taken over a constant circle the equations \Ref{orthogonality1} converge to \Ref{orthogonality}.

It remains to show that $W_j$  does not depend on the choice of a convergent sequence
$((u^{m},\ga^{m}))_{m=1}^\infty$.
Notice that if $\Psi(x,z\,|\,u,\ga)$ satisfies \Ref{orthogonality},
 then
\beq\label{orthogonality2}
\res_{z=\mu_j-1}\frac {g(z)\,\Psi(\ell,z\,|\,u,\ga)}{(z-\mu_j+1)^{2m_j}}=0\,,
\qquad  \ell=0,\ldots, k-1,\qquad \forall \, g(z)\in W_j\,.
\eeq
The function $\Psi(\ell, z\,|\,u,\ga)$ is a monic polynomial of degree $k+\ell$. Hence, the tuple of functions $\Psi(\ell,z)$ defines a point $W^\bot\in {\rm Gr}(k,2k)$. The $k$-dimensional vector space $W^\bot$ defines all spaces
$W_j$, $j=1,\dots, q$, uniquely.
\end{proof}

\begin{cor}
\label{mapMay}
By Theorem \ref{thm:14}, every point $(u,\ga)\in \mc P_k$
produces the two collections
\linebreak
$(\mu_i(u,\ga))_{i=1}^q$ and $(W_i(u,\ga) \in {\rm Gr}(m_i,2m_j))_{i=1}^q$.
That is an extension of the map \Ref{corr27}.
\end{cor}

\medskip
Equations \Ref{orthogonality} imply the following lemma.

\begin{lem}\label{zerooder} The function $\psi(x,z\,|\,\,u,\ga)$ has
a pole of order $m\leq m_j$ at $z=\mu_j(u,\ga)-1$ if
and only if the corresponding subspace $W_j(u,\ga)$ contains $(m_j-m)$-dimensional subspace spanned by the polynomials $(z-\mu_j+1)^{2m_j-\ell-1},\, \ell=0,\ldots, m-1$.
\qed
\end{lem}

The following statement is used below in the proof of Theorem \ref{5.5}.
Let $f(z)$ be a function holomorphic at $z=\mu_j-1$. Multiplication by $f(z)$
defines a linear operator
\bean
\label{fhat}
\phantom{aaaaaa}
f_*\ :\  \C[z]/\langle(z-\mu_j+1)^{2m_j}\rangle \ \to\  \C[z]/\langle(z-\mu_j+1)^{2m_j}\rangle,
\qquad g(z) \mapsto f(z)g(z)\,.
\eean

\begin{lem} \label{finv}
If \, $f_z(\mu_j-1)\neq 0$, then the only $m_j$-dimensional subspace $W$ of
\linebreak
$\C[z]/\langle(z-\mu_j+1)^{2m_j}\rangle$, invariant under the action of $f_*$,
  is the subspace spanned by
$(z-\mu_j+1)^\ell$,  $\ell=m_j,\ldots, 2m_j-1$.
\end{lem}

\begin{proof} The Jordan normal form of $f_*$ is the single Jordan block of size $2m_j$.
Such an operator has a single invariant $m_j$-dimensional subspace. That subspace is described in the lemma.
\end{proof}

\subsection{Extension of the inverse transform}
\label{sec 5.5}

 The construction of the inverse correspondence is straightforward.
 The spectral data  is a triple $(\mu, m, W)$,
 where $\mu=(\mu_1,\ldots, \mu_{q})$ is a set of distinct
nonzero complex numbers;
$m=(m_1,\dots,m_q)$  a set of positive integers with  $\sum_{i=1}^{q} m_j=k$;
$W=(W_1,\ldots,W_{q})$
a set of  spaces, where each $W_j$ is an $m_j$-dimensional subspace of the space
 of polynomials of degree $2m_j-1$.

\begin{lem} \label{psiGras}
Given $(\mu, m, W)$ there is a unique function $\Psi(x,t,z)$,
\beq\label{psigeneral}
\Psi(x,t,z)=\Omega(x,t,z)\left(z^k+\sum_{s=1}^k\xi_\ell(x,t)z^{k-s}\right),
\eeq
such that equations \Ref{orthogonality} hold.
\end{lem}

\begin{proof}
The proof is by explicit construction, as in its particular case of
Lemma \ref{5.2}. Choose a basis $g_{j,k}(z)$ in $W_j$.
Then equations \Ref{orthogonality} can be represented
 in the form of the inhomogeneous linear system of equations
\beq\label{www}
M(x,t\,|\,\mu,m, W)\,\xi(x,t)\,=\,-\,e_0
\eeq
with some matrix $M$, whose entries are explicit expressions that
are polynomial in $x$ and $t$ and  linear in the
coefficients of the polynomials $g_{j,k}(z)$. As before the function
$\Psi$ can be written in the same determinant form as in \Ref{hatT}:
\beq\label{pimuw}
\Psi(x,t,z\,|\,\mu,m,W)=\frac{\det \widehat M(x,t,z\,|\,\mu,m,W)}{y(x,t\,|\,\mu,m,W)}\, ,
\eeq
where
\beq\label{pimuw1}
y(x,t\,|\,\mu,m,W)=\det M(x,t\,|\,\mu,m,W)\,.
\eeq
\end{proof}

\begin{rem}
We emphasize that unlike  in the generic case  considered in Section \ref{sec 7.3}, the degree $k$ of the polynomial
$y(x,t\,|\,\mu,m,W)$ in $x$ depends not only on the number of distinct eigenvalues $\mu_j$ and
their multiplicities $m_j$ but also on the combinatorial
 types of cells of Grassmannians
 ${\rm Gr} (m_j,2m_j)$, which contain the given subspaces $W_j$.
\end{rem}

Denote the roots of the polynomial
 $y(x,t\,|\,\mu,m, W)$ by $u_i(t\,|\,\mu,m,W)$, $i=1,\dots,k$.
 Define $\ga(t\,|\,\mu,m,W)=(\ga_1(t\,|\,\mu,m,W),\dots,\ga_k(t\,|\,\mu,m,W))$ by formula \Ref{gaga}.

\vsk.2>

Let $\hat {\mathcal S}\subset (\C^\times)^{q}\times \prod_{j=1}^{q}{\rm Gr}(m_j,2m_j)$
be the subset of points $(\mu,W)$, such that
\begin{enumerate}
\item[(a)]  $\mu=(\mu_1,\dots,\mu_{q})$ has distinct coordinates;

\item[(b)]   $u(0\,|\,\mu,W)=(u_1(0\,|\,\mu,W),\dots,u_k(0\,|\,\mu,W))$ has distinct coordinates.

\end{enumerate}

\begin{thm}
\label{thm 514}
For $(\mu,W)\in \hat {\mathcal S}$, the map
\beq
\label{gainverse1}
\tilde S\ :\ (\mu,W)\ \mapsto \  (u(0\,|\,\mu,W), \ga(0\,|\,\mu,W))
\eeq
is inverse to the map in Corollary \ref{mapMay}.
\end{thm}

\begin{proof} The proof of Theorem \ref{thm 514} is similar to
the proof of Theorem \ref{tm:5.7}. The key point of the proof is  the following lemma.

\begin{lem} \label{inversegeneral}
Functions
$\Psi(x,t,z)$ and $y(x,t)$  given by \Ref{pimuw} and \Ref{pimuw1}, respectively, satisfy
equation \Ref{intRS}.
\end{lem}

The proof of the lemma is based on the uniqueness of the Baker-Akhiezer function corresponding to the data $(\mu,W)$ and almost word by word follows the proof of Lemma \ref{inversesimple}.
\end{proof}

\section{Solution of the rational RS hierarchy}\label{SRS}

The goal of this section is to write explicitly equations describing time dependence of
the roots $(u_i(t))$ of the polynomial $y(x,t)$ corresponding to the spectral data
$(\mu,W)\in \hat {\mathcal S}$.

\vsk.2>
It was proved in \cite{KZ} that the dependence of
$(u_i(t))$ in the variable $t_1$ coincides with the equation of motion of the  RS system.
Note that in \cite{KZ} this result was proved for the elliptic RS system.
The dependence of $(u_i(t))$ in  the variables $\bar t=(\bar t_0,\bar t_1,\bar t_2,\ldots)$, defined by formula
\beq\label{tt'}
x(z+1)^{-1}+\sum_{m=1}^\infty m\,t_m\,z^{m-1} \,=
\,\bar t_0\,(z+1)^{-1}\,+\,\sum_{m=1}^\infty m \,\bar t_m\,(z+1)^{m-1},
\eeq
was identified in \cite{KZ} with the pole dynamics of the elliptic (rational) solutions of
the $2D$ Toda hierarchy. In \cite{iliev} and \cite{zab} it was proved that the latter coincides with the flows defined by the higher Hamiltonians $H_k=tr L^k$ of the RS system, where $L$ is the corresponding Lax matrix.

\begin{rem} Note that the change variables \Ref{tt'} is well-defined only under the assumption that there are
only  finitely many of nonzero time variables. Nevertheless, the corresponding triangular change of the vector fields is well-defined always:
\bea
\phantom{aaa}
\p_{\bar t_0}\,=\,\p_x, \qquad
\p_{\bar t_1}\,=\,\p_{t_1}, \qquad \p_{\bar t_2}\,=\,\p_{t_2}+2\p_{t_1}, \qquad
 \p_{\bar t_3}\,=\,\p_{t_3}+3\p_{t_2}+3\p_{t_1},\qquad \ldots\ .
\eea
\end{rem}

\subsection{Hierarchies of linear equations}

In this section we  show that for any spectral data $(\mu,W)$ the
corresponding Baker-Akhiezer function $\Psi(x,t,z)$ given
by formula \Ref{pimuw} satisfies a hierarchy of linear equations.
\vsk.2>

Let $T_x=e^{\p_x}$ be the shift operator acting on functions of $x$,  $T_x: f(x) \mapsto f(x+1)$.

\begin{lem}\label{oper+}
Let $\Psi(x,t,z)$ be a formal series of the form
\beq\label{psiformal}
\Psi=z^k\Omega(x,t,z)\left(1+\sum_{s=1}^\infty \xi_s(x,t)z^{-s}\right),
\eeq
where  $k\in\Z$ and $\xi_s (x,t)$ are some functions of $x,t$.
Then for each $m\geq 1$ there is a unique difference operator $D_m$
in the variable $x$,
\beq\label{d-operator}
D_m\,=\,T_x^m\,+\,\sum_{i=1}^{m} \,w_{i,m}(x,t)\,T_x^{m-i}\,,
\eeq
such that
\beq\label{d-congruence}
D_m\Psi(x,t,z)=z^m \Psi(x,t,z)+\mc O(z^{k-1})\,\Omega(x,t,z)\,.
\eeq
The coefficients $w_{i,m}(x,t)$ of these operators
$D_m$
are (explicit) difference polynomials in $\xi_s(x,t)$,  $s=1, \ldots, m-1.$
\end{lem}
\begin{proof} Divide both sides of \Ref{d-congruence} by $\Omega(x,t,z)$ and compare the leading coefficients of Laurent series. That gives a triangular system of $m-1$ linear
equations for $m-1$ unknown functions $w_{i,m}(x,t)$. The system is  solved recurrently.
\end{proof}

The following theorem follows from  the uniqueness of the Baker-Akhiezer function.

\begin{thm}\label{hierarchies13}
Let $D_m$ be the operator defined in Lemma \ref{oper+} by the Baker-Akhiezer function $\Psi(x,t,z\,|\,\mu,W)$  given by  \Ref{pimuw}. Then
\beq\label{h2D}
\left(\p_{t_m}-D_m\right)\Psi(x,t,z\,|\,\mu,W)=0\,,
\qquad \ell \geq 1\,.
\eeq
\end{thm}

\begin{proof}
The definition of $D_m$ in Lemma \ref{oper+} implies that the left-hand side
of \Ref{h2D} has the form $\tilde R\,\Omega$, where $\tilde R$ is a
 polynomial in $z$ of degree $k-1$. The function $\tilde R\,\Omega$ satisfies the
  system of equations \Ref{orthogonality2} defining $\Psi$. Therefore the coefficients
  of $\tilde R$ satisfy the homogeneous linear system of equation with matrix
   $M$ as in \Ref{www}. Hence, $\tilde R=0$.
\end{proof}

\begin{rem}
Lemma \ref{inversegeneral} is a particular case of Theorem \ref{hierarchies13} for $m=1$.
\end{rem}

The compatibility conditions of equations \Ref{h2D} imply:

\begin{cor}\label{cor:hier}

If the Baker-Akhiezer function $\Psi$ is given by \Ref{pimuw}, then the
the corresponding operators $D_m$
satisfy the equations
\bean
\label{2DT}
\left[\p_{t_i}-D_i, \p_{t_j}-D_j\right]=0\,,
\eean
for all $i,j$.
\end{cor}

\begin{rem}
The collection of equations \Ref{2DT} is
the so-called Zakharov-Shabat presentation of a part of the  $2D$ {\it Toda hierarchy}.
We call the collection   of equations \Ref{2DT} the {\it positive part}
of the  $2D$ {\it Toda hierarchy}, see  Section \ref{sec 8.2}.

\end{rem}

\subsection{Rational RS hierarchy}

\subsubsection{}
\label{sebs}
Let $(u,\ga)\in\mc P_k$ be a point of the phase space. Let $L=L(u,\ga)$ be the Lax matrix.
We
define recurrently a set of rational functions
$\bar w_{1,m}(x),\bar w_{2,m}(x)$, \ldots, $\bar w_{m,m}(x)$ by the formula
\beq
\label{hk}
\bar w_{s,m}(x)=\sum_{i=1}^k\left(\frac{(L^{s-1}\ga)_i}{x-u_i}-\frac{(L^{s-1}\ga)_i}{x-u_i+m}-
\sum_{\ell=1}^{s-1}\bar w_{\ell,m}(x)\frac{(L^{s-1-\ell}\ga)_i}{x-u_i+m-\ell}\right),
\eeq
a set of  matrices $H_{1,m}, \dots, H_{m,m}$ by the formulas
\beq
\label{Hkmatrix}
(H_{s,m})_{ij}=\frac {\res_{x=u_i}\bar w_{s,m}(x)}{u_i-u_j+m-s}\,,
\eeq
\beq\label{Hmmatrix}
(H_{m,m})_{ij}=\widetilde H_{m,i} \delta_{ij}+(1-\delta_{ij})\,\,\frac {\res_{x=u_i} \bar w_{m,m}}{u_i-u_j}\,,
\eeq
where $\widetilde H_{m,i}$ is defined by the Laurent expansion of $\bar w_{m,m}(x)$ at $x=u_i$,
\beq\label{hmtilde}
\bar w_{m,m}(x)=\frac {\res_{x=u_i}\bar w_{m,m}}{x-u_i}+\widetilde H_{m,i} +\mc O(x-u_i),
\eeq
and the matrix $M_m$ by the formula
\beq\label{MRS}
M_m=\sum_{s=1}^m H_{s,m} L^{m-s}\,.
\eeq

\subsubsection{}
Let us return
to the situation of Section \ref{sec 5.5}. Let  the spectral data $(\mu, m, W)$ be given.
 Let $y(x,t)$ be the polynomial
defined by formula \Ref{pimuw1} and  $u(t)=(u_i(t))_{i=1}^{k}$ its roots. Let
$\ga(t)=(\ga_i(t))_{i=1}^{k}$ with $\ga_i(t) =\p_{t_1}u_i(t)$.
Having the pair $(u(t),\ga(t))$ we may define all the objects of Section \ref{sebs}, which will
depend on $t$.

\vsk.2>
Let $\bar t_m$ be the variables defined in \Ref{tt'}.

\begin{thm}\label{6.4} The pair  $(u(t),\ga(t))$
satisfies the equations of motion of
the hierarchy of the $k$ particle rational RS system. Namely, for all $m\geq 1$ we have
\beq\label{meq1}
\p_{\bar t_m} u_i=\res_{x=u_i} \bar w_{m,m}(x)\,,
\eeq
\beq\label{meq2}
\p_{\bar t_m} \ga_i =\sum_{j=1}^{k} \big((M_m)_{ij}L_{ji}-L_{ij}(M_{m})_{ji}\big)\,.
\eeq

\end{thm}
\begin{proof}

  The following lemma gives the Lax presentation of these flows in terms of the RS system.

\begin{lem} Let the linear equation
\beq\label{h2D'}
\left(\p_{\bar t_m}-T_x^m-\sum_{s=1}^m \,\bar w_{s,m}(x,\bar t)\,T_x^{m-s}\right)\psi(x, \bar t,\bar z)=0 \,
\eeq
with some (a'priory unknown) coefficients $\bar w_{s,m}(x,\bar t)$ has a solution of the form
\beq\label{psit'}
\psi(x,\bar t,\bar z)=\left(1+\sum_{i=1}^{k_n} \frac{C_i(\bar t,\bar z)}{x-u_i}
\right)\bar z^{\,x}e^{\sum_m \bar t\, \bar z^m},
\eeq
where $\bar z=z+1$ and $C$ is given by \Ref{kramerd} with
the matrix $L$ defined in \Ref{dL} with $\ga_i=\ga_i(\bar t_m), u_i=u_i(\bar t_m)$.
Then
 equations \Ref{meq1}, \Ref{meq2} hold.
\end{lem}
\begin{proof}
The vector $C$ with the coordinates $C_i$ given by \Ref{kramerd} solves  equation \Ref{jan29a}, i.e.
\beq\label{may301}(\bar z L^0-L)C=\ga,
\eeq
where $L^0=E$ is the identity matrix. This equation easyly implies that for any $s$ the equation
\beq\label{may302}
(\bar z^sL^0 - L^s)C=\sum_{\ell=0}^{s-1}\bar z L^{s-\ell-1}\ga
\eeq
holds.

The substitution of \Ref{psit'} into \Ref{h2D'} gives the equation
\bean
\label{may303}
&&
\sum_{i=1}^k \left(\frac{\bar z^m C_i}{x-u_i}+\frac{(\p_{\bar t_m} u_i) C_i}{(x-u_i)^2}
+\frac{\p_{\bar t_m}C_i}{x-u_i}\right)
\\
\noindent
&&
\phantom{aa}
=\
\sum_{i=1}^k \frac {\bar z^m C_i}{x-u_i+m}+
\sum_{s=1}^m \bar z^{m-s}\bar w_{s,m}\left(1+\sum_{i=1}^k \frac{C_i}{x-u_i+m-s}\right)\,.
\eean
Using \Ref{may302} and then equating the coefficients at $ \bar z^\ell$ for $ \ell=m-1, m-2,\ldots,0$
at both sides of the equation
we recurrently find that $\bar w_{s,m}(x)$ are given by  formulas \Ref{hk}. The remaining part of the
equations (of order $\mc O(\bar z^{-1}$)) are linear equations containing $C(\bar z)$. Equating the coefficients at $(x-u_i)^{-2}$ we get equation
\Ref{meq1}. Equating the coefficients at $(x-u_i)^{-1}$ we get that the vector $C$ satisfies the equation
\beq\label{RSMM}
\p_{\bar t_m}C=(M_m-L^m)C\,,
\eeq
where the matrix $M_m$ is defined in \Ref{MRS}. Comparing the leading coefficients at of the expansions in $\bar z^{-1}$ of the both sides of \Ref{RSMM} we get
\beq\label{Mga}
(M_m-L^m)\ga=0\,.
\eeq
From \Ref{may301} and \Ref{RSMM} it follows that
\beq\label{jun3}
[\p_{\bar t_m} - M_m,L]C=-(M_m -L^m)\ga=0\,.
\eeq
Since  equation \Ref{jun3} holds for $C=C(z)$ we have
\beq\label{junlax}
\p_{\bar t_m}L =[M_m,L]\,.
\eeq
The latter is the Lax presentation of equations \Ref{meq1} and \Ref{meq2}.

In the framework of the
dynamical $r$-matrix approach the matrices $M_m$ were obtained in \cite{suris}.
\end{proof}

Now Theorem \ref{6.4} follows from
 Theorem \ref{hierarchies13}.
\end{proof}

\section{Spectral transform for $N$-periodic Bethe ansatz equations}
\label{sec 7}

\subsection{Spectral data for solutions of the Bethe ansatz equations}
\label{sec 7.1}
We begin this section by identification of the spectral data corresponding to solutions of the $N$-periodic Bethe ansatz
equations. Recall that for a given sequence of generic polynomials
$(y_n(x))_{n\in\Z}$  whose roots satisfy the Bethe ansatz equations the solutions $(\psi_n(x,z))_{n\in\Z}$ of the generating linear problem constructed in Section \ref{S:gen}
are equal to
\beq\label{psipsi}
\psi_n(x,z)=z^n\psi(x,z\,|\,u^{(n)},\ga^{(n)})\,,
\eeq
where $(u^{(n)}_i)_{i=1}^{k_n}$ are roots of $y_n(x)$ and
 $(\ga^{(n)}_i)_{i=1}^{k_n}$ are defined in \Ref{dan}. Notice that
$\ga^{(n)}$ depends on the polynomials
 $y_n(x)$ and $y_{n+1}(x)$, only. By definition
  of generic polynomials, we have $(u^{(n)},\ga^{(n)})\in {\mathcal P}_{k_n}$.

\begin{lem}
\label{nozeros}
If $(y_n(x))_{n\in\Z}$ represents a solution of the $N$-periodic Bethe ansatz equations, then
the matrix $L(u^{(n)},\ga^{(n)})$ has only one eigenvalue $\mu=1$ (of multiplicity $k_n$).
\end{lem}

\begin{proof}

By Theorem \ref{thm:14} the function
$$\Psi(x,z\,|\,u^{(0)},\ga^{(0)}) =\det L(u^{(0)},\ga^{(0)})\, \psi_0(x,z)$$
satisfies equations \Ref{orthogonality} with $W_j^{(0)}:=W_j(u^{(0)},\ga^{(0)})$.
From equation \Ref{laxdd} it then follows that for functions
$\det L(u^{(0)},\ga^{(0)})\psi_n(x,z)$ for $n\geq 0$ equations
 \Ref{orthogonality} with $W_j^{(0)}:=W_j(u^{(0)},\ga^{(0)})$ hold,
 as well.  The $N$-periodicity of $(y_n)$ implies that $\psi_N=z^N\psi_0(x,z)$.
 Hence, $\Psi(x,z\,|\,u^{(0)},\ga^{(0)})$ satisfies equation
 \Ref{orthogonality} and the equations
\beq\label{orthogonality+}
\res_{z=\mu_j-1}\frac {g(z)z^N\Psi(x,z\,|\, u^{(0)},\ga^{(0)})}{(z-\mu_j+1)^{2k_j}}=0, \ \forall g\in W_j(u^{(0)},\ga^{(0)})
\eeq
Since $\Psi(x,z\,|\,u^{(0)},\ga^{(0)})$ defines $W_j^{(0)}$ uniquely,
 equations \Ref{orthogonality+} imply that $W_j^{(0)}$ is invariant
 under the action of the operator of multiplication by $z^N$. It follows from Lemmas  \ref{zerooder} and
\ref{finv}  that $\Psi(x,z\,|\,u^{(0)},\ga^{(0)})$ has zero
of order $m_j$ at $z=\mu_j-1$ for any $\mu_j\neq 1$, or
equivalently that the function $\psi(x,z\,|\,u^{(0)},\ga^{(0)},z)$ is holomorphic at $z=\mu_j-1$.

Now the reference to Lemma \ref{nonzero-stronger} finishes the proof.
\end{proof}

\begin{rem}

In Lemma \ref{4.5} we proved that the
poles of solutions $(\psi_n(x,z))_{n\in\Z}$ of the generating problem
corresponding to a sequence of polynomials $(y_n(x))_{n\in\Z}$ (possibly non-periodic)
are $n$-independent away from $z=0$. The lemma above gives a stronger statement: for periodic sequences of polynomials  the solutions $(\psi_n(x,z))_{n\in\Z}$  are holomorphic at $z\neq 0$.
\end{rem}

\subsection{The inverse spectral transform: construction}
\label{S:inverse}

By Theorem \ref{thm:14} and Lemma \ref{nozeros} the functions $\psi_n(x,z)$ constructed in Section \ref{S:gen} are uniquely defined by a sequence of points $W^{(n)}\in {\rm Gr}(k_n,2k_n)$ corresponding to the only
eigenvalue $\mu=1$
of the matrix $L(u^{(n)},\ga^{(n)})$. In this section we explicitly describe the data defining
such sequences and present in a closed form the construction of the solutions of
the $N$-periodic Bethe ansatz equations.

The parameters of the construction are  nonnegative integers $\nu$, $D$, and
an  $(N+\nu)\times (D+1)$-matrix
\bea
A=(a_{k,j})\,, \qquad k=1,\ldots, N+\nu, \quad j=0,\ldots, D\,.
\eea
We say that the matrix $A$ is {\it nondegenerate} if
for any $n=0,\dots,N$, the matrix  $A^{(n)}$ composed of the first $n+\nu$ rows of the matrix $A$ has rank $n+\nu$.

\vsk.2>
Two matrices $A, A'$ are called
{\it  equivalent} if $A=GA'$ where $G$ is an $(N+\nu)\times(N+\nu)$ nondegenerate matrix of the form
\beq\label{G}
G=\left(\begin{array}{cc} g & 0\\ *& g_1
\end{array}\right)
\eeq
where $g$ is a $\nu\times\nu$-matrix and $g_1$ is lower-triangular.

We call $A$ {\it reducible} if there is a nondegenerate $\nu\times \nu$-matrix $H$ such that
\beq\label{HA0}
HA^{(0)}=\left(\begin{array} {cc} E &0\\
0&*\\
\end{array}\right)
\eeq
where $E$ is the $\ell\times\ell$ unit matrix with $\ell\geq 1$.
We call $A$ {\it irreducible} otherwise.

\subsection{Function $\Om(x,t,z)$}\label{subsection:om}

Below we present some notations and properties of the function $\Om(x,t,z)$
defined in \Ref{Om0},
\bean
\label{Om}
\Om(x,t,z)\,=\, (1+z)^xe^{\sum_{j=1}^\infty t_jz^j}.
\eean

The function  $\Om(x,t,z)$ satisfies the equation
\bean
\label{Ome}
(z+1) \,\Om(x,t,z) = \Om(x+1,t,z) = \der_{t_1}\Omega(x,t,z)
\eean
and, more generally, the equations
\bean
\label{Omel}
\phantom{aaa}
z^\ell \,\Om(x,t,z)
&=&
 \sum_{m=0}^\ell (-1)^{\ell-m}\binom{\ell}{m} \Om(x+m,t,z) =
\Delta^{(\ell)}\Om(x,t,z)\,.
\\
\notag
\Om(x+\ell,t,z)& =& \der_{t_1}^\ell \Omega(x,t,z) = \der_{t_\ell} \Omega(x,t,z)\,, \ \ell\geq 1.
\eean
Introduce the polynomials $\chi_n(x,t),$  $n\in\Z_{\geq 0}$, by using the expansion
\bean
\label{chi}
\Om(x,t,z) \,=\, \sum_{n=0}^\infty \chi_n(x,t) z^n,
\eean
where $\chi_0(x,t)=1$,
\bean
\label{deg chi}
\phantom{aaa}
\chi_n(x,t)|_{t=0} = \binom{x}{n},\quad \chi_n(x,t)|_{x=0, \,t\rq{}=0} = t_1^n,
\quad
\deg_x \chi_n(x,t)=\deg_{t_1} \chi_n(x,t) = n\,.
\eean
For $n\geq 0$,
 we have
\bean
\label{chi rel}
\chi_{n}(x+1,t) - \chi_{n}(x,t) = \der_{t_1}\chi_n(x,t) = \chi_{n-1}(x,t),
\eean
where $\chi_{-1}(x,t) = 0$. More generally, we have
\bean
\label{de chi}
\Delta^{(\ell)}\chi_k(x,t) =  \der_{t_1}^\ell \chi_n(x,t)= \der_{t_\ell} \chi_n(x,t)=\chi_{k-\ell}(x,t)\,.
\eean
Let us write
\bea
e^{\sum_{j=1}^\infty t_jz^j} = \sum_{k=0}^\infty h_k(t) z^k\,,
\eea
where  $h_0(t)=1$.
Then
\bean
\label{chi h}
\chi_n(x,t) = \sum_{k=0}^n h_{n-k}(t) \binom{x}{k}\,.
\eean

\vsk.2>
Given the spectral data $A=(a_{k,j})$, define the polynomials $f_k(x,t)$ by the formula
\bean
\label{f_k}
f_k(x,t) = \sum_{j=0}^D a_{k,j}\,\chi_j(x,t), \qquad k=1,\dots, N+\nu\,.
\eean
For $k=1,\dots, N+\nu$, introduce the differential operators
\bean
\label{D_k}
\D_k =  \sum_{j=0}^D \frac{a_{k,j}}{j!}\,\frac{\der^j}{\der z^j}\ .
\eean
Then
\bean
\label{Om-f}
\Big[\D_k \Om(x,t,z)\Big]_{z=0} = f_k(x,t).
\eean
\begin{lem}
\label{lem Af}
If $A$ is nondegenerate, then for every $n=0,\dots,N$,

\begin{enumerate}

\item[(i)]
 the discrete Wronskian
$\Wh(f_1,\dots,f_{n+\nu})$ is nonzero;
\item[(ii)]
\bean
\label{Wr=}
\Wh(f_1,\dots,f_{n+\nu})=
\Wr_{t_1}(f_1,\dots,f_{n+\nu}),
\eean
where $\Wr_{t_1}(f_1,\dots,f_{n+\nu}) = {\det}_{i,j=1}^m\big(\der_{t_1}^{j-1}f_i\big)$ is the standard Wronskian  with respect to the variable $t_1$;
\item[(iii)]
\bean
\label{deg=}
\deg_x\Wh(f_1,\dots,f_{n+\nu})=
\deg_{t_1} \Wh(f_1,\dots,f_{n+\nu})\,.
\eean

\end{enumerate}
\qed
\end{lem}

\subsection{Baker-Akhiezer functions}
\label{sec BA}

For  every $n=0,\dots,N$, consider a polynomial of degree $n+\nu$ in $z$ of the form
\bean
\label{R_n}
R_n(x,t,z)\,=\,
z^{n+\nu}\left(1+\sum_{\ell=1}^{n+\nu}\xi_{\ell}^{(n)}(x,t)\,z^{-\ell}\right),
\eean
whose coefficients are some functions in $x$, $t$.

\begin{lem}\label{5.2}
If $A$ is nondegenerate, then for any $n=0,\dots,N$, there exists a unique
function $\psi_n(x,t,z)$ of the form
\bean
\label{si_n}
\psi_n(x,t,z)\,= \,\Om(x,t,z)\, R_n(x,t,z),
\eean
such that
\beq
\label{cusd}
\Big[\D_k \psi_n(x,t,z)\Big]_{z=0}=0, \qquad k=1,\ldots, n+\nu\,.
\eeq
\end{lem}
For fixed $n,x$ the function  $\psi_n(x,t,z)$ is a particular case of the {\it Baker-Akhiezer functions} introduced in \cite{K2} to construct rational solutions of the KP equation.
\begin{proof}

Using equation \Ref{Omel}, we rewrite equation \Ref{cusd} as
\bean
\label{5.11}
&&
\Big[\D_k \Big(\sum_{m=0}^{n+\nu} (-1)^{n+\nu-m}\binom{n+\nu}{m} \Om(x+m,t,z)
\\
\notag
&&
\phantom{aaaa}
+\sum_{\ell=1}^{n+\nu}\xi_{\ell}^{(n)}(x,t)
\sum_{m=0}^{n+\nu-\ell} (-1)^{n+\nu-\ell-m}\binom{n+\nu-\ell}{m} \Om(x+m,t,z)
 \Big) \Big]_{z=0}=0\,.
\eean
Using \Ref{chi} and \Ref{f_k} we rewrite \Ref{5.11} as
\bean
\label{f-syst}
&&
\sum_{m=0}^{n+\nu} (-1)^{n+\nu-m}\binom{n+\nu}{m} f_k(x+m,t)
\\
\notag
&&
\phantom{aaaa}
+\sum_{\ell=1}^{n+\nu}\xi_{\ell}^{(n)}(x,t)
\sum_{m=0}^{n+\nu-\ell} (-1)^{n+\nu-\ell-m}\binom{n+\nu-\ell}{m} f_k(x+m,t)
  =0\,.
\eean
The system of equations (\ref{f-syst}) is  the systems of $n+\nu$
inhomogeneous linear equations for the coefficients $\xi_{\ell}^{(n)}(x,t)$,
\beq
\label{sys}
\sum_{\ell=1}^{n+\nu} M_{k,\ell}^{(n)}(x,t)\, \xi_{\ell}^{(n)}(x,t)= F^{(n)}_k(x,t),
\eeq
where
\bean
\label{M}
M_{k,\ell}^{(n)}(x,t)
&=&
\sum_{m=0}^{n+\nu-\ell} (-1)^{n+\nu-\ell-m}\binom{n+\nu-\ell}{m} f_k(x+m,t) = \Delta^{n+\nu-\ell} f_k(x,t)
\\
\notag
F^{(n)}_k(x,t)
&=&
-\sum_{m=0}^{n+\nu} (-1)^{n+\nu-m}\binom{n+\nu}{m} f_k(x+m,t) = -\Delta^{n+\nu} f_k(x,t).
\eean
Using \Ref{chi rel} we may rewrite
\bean
\label{Mrel}
M_{k,\ell}^{(n)}(x,t)
&=&
\sum_{j=\ell-1}^D a_{k,j}\,\chi_{j-\ell+1}(x,t),
\\
\notag
F^{(n)}_k(x,t)
&=&
-\sum_{j=n+\nu}^D a_{k,j}\,\chi_{j-n-\nu}(x,t),
\eean
cf. formula for $f_k(x,t)$ in \Ref{f_k}.

Formula \Ref{M} implies that the determinant of the matrix $ M^{(n)}(x,t)$ equals
\beq
\label{formd}
y_{n}(x,t): =\det M^{(n)}(x,t)=\Wh(f_1,\dots, f_{n+\nu}),
\eeq
the discrete Wronskian of the polynomials $f_1(x,t),\dots, f_{n+\nu}(x,t)$
with respect to $x$.  By Lemma \ref{lem Af}
the determinant is a nonzero polynomial. Hence equations \Ref{cusd} determine uniquely a function $\psi_n(x,t,z)$.
The lemma is proved.

\vsk.2>
Below we  give a determinant formula for $\psi_n(x,t,z)$.
Define an
$(n+\nu+1)\times(n+\nu+1)$ matrix $\widehat M^{(n)}(x,t,z)$,
 whose rows and columns are
labeled by indices $1,\dots,n+\nu+1$ and entries are given by the formulas:
\bean
\label{Mhat}
\widehat M^{(n)}_{n+\nu+1,\ell}
&=&
\phantom{a}
z^{\ell-1}, \qquad \ell=1,\dots, n+\nu+1,
\\
\notag
\widehat M^{(n)}_{\ell, n+\nu+1}
&=&
-F^{(n)}_\ell,  \qquad  \ell=1\ldots, n+\nu,
\\
\notag
\widehat M^{(n)}_{k,\ell}\,
&=&
\phantom{a}
 M^{(n)}_{k,\ell}, \qquad k,\ell = 1,\dots, n+\nu.
\eean
Using the determinant expansion of $\widehat M^{(n)}(x,t,z)$ by  the last row  we obtain
\bean
\label{psiwdM}
\psi_n(x,t,z)=\Om(x,t,z)\, \frac {\det \widehat M^{(n)}(x,t,z)}{y_n(x,t)}.
\eean
Here is a useful formula for $\xi^{(n)}_1(x,t)$,
\bean
\label{xi1}
\xi^{(n)}_1(x,t) = -\frac{\Delta y_n(x,t)}{y_n(x,t)}  = -\frac{\der_{t_1} y_n(x,t)}{y_n(x,t)}\,.
\eean
\end{proof}

\begin{thm}
\label{thm BAk}

The Baker-Akhiezer functions $(\psi_m(x,t,z))_{m=0}^N$  satisfy
 equations (\ref{laxdd}) with indices $n=0, \dots,N-1$ in which the functions
$v_n(x,t)$ are given  in terms of $y_n(x,t)$ and $y_{n+1}(x,t)$
by  formula
\Ref{vpot}.

\end{thm}

\begin{proof}

Consider the function
\beq
\label{tpsi}
\tilde \psi_{n+1}(x,t,z)=\psi_{n}(x+1,t,z) - v_n(x,t) \psi_{n}(x,t,z)-\psi_{n+1}(x,t,z)\,.
\eeq
We need to show that
$\tilde \psi_{n+1}(x,t,z)$ is the zero function.

We have
\bea
&&
\tilde \psi_{n+1}(x,t,z) = \Om(x+1,t,z) R_{n}(x+1,t,z)
\\
&&
\phantom{aaa}
-v_n(x,t) \Om(x,t,z) R_{n}(x,t,z)  -\Om(x,t,z) R_{n+1}(x,t,z)
\\
&&
\phantom{aaaaaa}
= \Om(x,t,z)\Big((1+z)R_{n}(x+1,t,z) - v_n(x,t) R_{n}(x,t,z)-R_{n+1}(x,t,z) \Big).
\\
&&
\phantom{aaaaaaaaa}
= \Om(x,t,z) \,\tilde R_{n+1}(x,t,z),
\eea
where $\tilde R_{n+1}(x,t,z)$ is a polynomial in $z$ of degree at most $n+\nu$,
\beq
\label{Rn}
\tilde R_{n+1}(x,t,z) = \sum_{\ell=1}^{n+\nu+1} \tilde \xi^{(n+1)}_{\ell}(x,t)\,z^{\ell-1}\,.
\eeq

Each of the three functions $\psi_{n}(x+1,t,z)$, $v_n(x,t) \psi_{n}(x,t,z)$, $\psi_{n+1}(x,t,z)$ satisfy the equations
\Ref{cusd} for $k=1,\dots,n+\nu$. Hence the function
$\tilde \psi_{n+1}(x,t,z) $ satisfies equations
\Ref{cusd} for $k=1,\dots,n+\nu$.

\begin{lem}
\label{lem lae}
The function $\tilde \psi_{n+1}(x,t,z) $  satisfies equation \Ref{cusd} for $k=n+\nu+1$.
\end{lem}

\begin{proof}
Recall the function $\det M^{(n)}(x,t,z)\Om(x,t,z)$, see \Ref{Mhat}. Then
\bean
\label{D_khatM}
\D_{n+\nu+1}\Big[\big(\det M^{(n)}(x,t,z)\big)\,\Om(x,t,z)\Big]_{z=0} = \det M^{(n+1)}(x,t) = y_{n+1}(x,t),
\eean
see formulas \Ref{Omel} and \Ref{Om-f}.  Now we apply the operator $\D_{n+\nu+1}[\, ]_{z=0}$ to both sides of equation
\Ref{tpsi}. We have $\D_{n+\nu+1}\big[ \psi_{n+1}(x,t,z)\big]_{z=0}=0$ by definition of $\psi_{n+1}(x,t,z)$. Hence
\bea
&&
\D_{n+\nu+1}\Big[\tilde \psi_{n+1}(x,t,z)\Big]_{z=0}\!\! = \D_{n+\nu+1}\Big[ \psi_{n}(x+1,t,z)
 -v_n(x,t)
\psi_{n}(x,t,z)\Big]_{z=0}
\\
&&
\phantom{aaa}
= \ \D_{n+\nu+1}\Big[\Big(\frac{\det \widehat M^{(n)}(x+1,t,z)}{y_n(x+1,t)} \,-\,v_n(x,t)\,
\frac{\det \widehat M^{(n)}(x,t,z)}{y_n(x,t)}\Big)\Om(x,t,z)
\Big]_{z=0}
\\
&&
\phantom{aaa}
= \ \frac{y_{n+1}(x+1,t)}{y_n(x+1,t)} \
- \
\frac{y_n(x,t)y_{n+1}(x+1,t)}
{y_n(x+1,t)y_{n+1}(x,t)}\, \frac{y_{n+1}(x,t)}{y_n(x,t)}=0.
\eea
\end{proof}
Comparing the system of equations \Ref{sys} with $k=1,\dots,n+\nu$ for the coefficients $(\xi^{(n)}_\ell(x,t))_{\ell=1}^{n+\nu}$
with the system of equations \Ref{cusd}  with $k=1,\dots,n+\nu+1$
for the coefficients $(\tilde \xi^{(n+1)}_\ell(x,t))_{\ell=1}^{n+\nu+1}$ we conclude that the coefficients
$(\tilde \xi^{(n+1)}_\ell(x,t))_{\ell=1}^{n+\nu+1}$ satisfy the system of homogeneous equations
\beq
\label{sysn+1}
\sum_{\ell=1}^{n+\nu+1} M_{k,\ell}^{(n+1)}(x,t)\, \tilde\xi_{\ell}^{(n+1)}(x,t)= 0,
\eeq
with $k=1,\dots,n+\nu+1$.
According to our previous reasonings
the determinant
\\
$\det M^{(n+1)}(x,t)= y_{n+1}(x,t)$ of the matrix of this homogeneous system
is a nonzero polynomial. Hence  all the functions $\tilde\xi_{\ell}^{(n+1)}(x,t)$ are the zero functions,
the function $\tilde \psi_{n+1}(x,t,z) $ is the zero function, the functions $\psi_n(x,t,z)$ with $n=0,\dots, N-1$ satisfy
equations (\ref{laxdd}), and the theorem is proved.
\end{proof}

\subsection{Reconstruction of $A$}\label{Aconstruction}

Let $A$ and $\tilde A$ be two nondegenerate  $(N+\nu)\times (D+1)$-matrices,
\bea
A=(a_{k,j})\,,\quad \tilde A=(\tilde a_{k,j})\,, \qquad k=1,\ldots, N+\nu, \quad j=0,\ldots, D\,.
\eea
Let $(\psi_m(x,0,z))_{m=0}^N$ and
$(\tilde \psi_m(x,0,z))_{m=0}^N$ be the corresponding Baker-Akhieser functions
given by the above construction.

\begin{thm}
\label{thm rec2}
Assume that
\bean
\label{eqpsi}
\psi_n(x,0,z)=\tilde\psi_n(x,0,z), \qquad n=0,\dots,N.
\eean
Then $A=G\tilde A$ for a matrix $G$  as in \Ref{G}.
\end{thm}

\begin{proof}
For any $n$ the function $\psi_n(x,0,z)$ is given by the formulas \Ref{sys}, \Ref{M}.
Consider the linear difference operator of order $n+\nu$,
\bean
\label{diff}
 \Delta^{(n+\nu)} + \xi^{(1)}(x,0)\Delta^{(n+\nu-1)}
+  \xi^{(2)}(x,0)\Delta^{(n+\nu-2)}+ \dots +  \xi^{(n+\nu)}(x,0)\,.
\eean
By formulas  \Ref{sys}, \Ref{M} the  kernel of this difference
operator is generated by the polynomials $f_1(x,0), \dots,f_{n+\nu}(x,0)$ given by formula \Ref{f_k}.

If two matrices $A, \tilde A$ have the same $\psi_n(x,0,z)$ and $\tilde\psi_n(x,0,z)$, then
the $n+\nu$-dimensional space generated by the polynomials $f_1(x,0), \dots,f_{n+\nu}(x,0)$ coincides with the space
generated by  the polynomials
$\tilde f_1(x,0), \dots, \tilde f_{n+\nu}(x,0)$.
Hence $A=G\tilde A$ for suitable $G$.
\end{proof}

\subsection{Periodicity constraint}
\label{sec per}

Given  spectral data $A=(a_{k,j})$,
the construction of Section \ref{sec BA} gives $y_0(x,t),\ldots, y_N(x,t)$ and $\psi_0(x,t,z),
\ldots, \psi_N(x,t,z)$.  We say that these functions  {\it extend periodically} if
there exist sequences $(y_n(x,t))_{n\in\Z}$ and $(\psi_n(x,t,z))_{n\in\Z}$ such that
\bea
y_{n+N}(x,t) = y_n(x,t), \qquad
\psi_{n+N}(x,t,z) = z^N \psi_n(x,t,z),   \qquad n\in \Z,
\eea
and the sequence  $(\psi_n(x,t,z))_{n\in\Z}$ satisfies equations \Ref{laxdd} with $(v_n(x,t))_{n\in\Z}$
given by \Ref{vpot} in terms of  $(y_n(x,t))_{n\in\Z}$.

\vsk.2>

It is clear that the periodic extension is possible if and only if
\beq\label{psiperiod}
y_N(x,t) \,=\, y_0(x,t), \qquad
\psi_N(x,t,z)\,=\,z^N\psi_0(x,t,z)\,.
\eeq
Our goal is to identify matrices $A$ for which the periodicity equations \Ref{psiperiod} hold.

\subsection{Construction of matrices $A$}
\label{sec con}

 Given $\nu$, let $W$ be
an $(N+\nu)\times(N+\nu)$ matrix such that its upper-right $\nu\times\nu$ corner $U$ is {\it nilpotent},
\beq\label{S}
W=\left(\begin{array}{cc}
V & U\\
* & *\\
\end{array}\right)\qquad \text{and} \qquad   U^r=0  \quad \text{for\,some}\quad r <\nu\,.
\eeq

Using $W$ we construct an $(N+\nu)\times N(\nu+1)$-matrix $A=A(W)$ in three steps.
\vsk.2>

First using $V$ and $U$ we construct a $\nu\times N\nu$ matrix $Q$ as follows.
Let $V=(v_1,\dots,v_N)$ be column vectors of $V$ and $Q=(q_1,\dots,q_{N\nu})$ column vectors of $Q$.
Define $q_j=v_j$ for $j=1,\dots, N$. Define $q_j$ for $j>N$ recursively by the formula
\beq\label{q}
q_{N+j}\, =\, U q_j\,.
\eeq
Define an $(N+\nu)\times N(\nu+1)$-matrix $P$ by the formula
\beq
\label{B}
P=\left(\begin{array}{cc}
E&0\\
0& Q
\\
\end{array}\right),
\eeq
where  $E$ is the $N\times N$ unit matrix.
Define the matrix $A$ by the formula
\bean
\label{A=UY}
A\,=\, WP\,.
\eean
It is easy to see that the matrix $A$ has the form
\beq\label{Ablock}
A=\left(\begin{array}{cc}
Q & 0\\
*&*\\
\end{array}\right)
\eeq

\subsection{Properties of the construction}
\label{sec Prop}

\begin{lem}
\label{5.55} If a matrix $A= A(W)$ is given by the construction of Section \ref{sec con}, then
the functions  $y_0(x,t),\ldots, y_N(x,t)$ and $\psi_0(x,t,z),
\ldots, \psi_N(x,t,z)$ extend periodically.

\end{lem}

\begin{proof}
The functions $y_0(x,t)$, $\psi_0(x,t,z)$ are determined by the first $\nu$ rows of $A$.
That gives $\nu$
equations \Ref{cusd} for $\psi_0(x,t,z)$. The functions $y_N(x,t)$,  $\psi_N(x,t,z)$ are determined by the full matrix
$A$. That gives $N+\nu$ equations \Ref{cusd} for $\psi_N(x,t,z)$. The
periodicity  constraint \Ref{psiperiod} holds
if  the space of linear combinations of equations defining $\psi_N(x,t,z)$
contains $N$ equations $\p_z^{(j)}\psi_N(x,t,z)|_{z=0}=0$,
$j=0,\ldots, N-1$, and $\nu$ equations \Ref{cusd} defining $\psi_0(x,t,z)$ in which the operators
$\D_k =  \sum_j \frac{a_{k,j}}{j!}\,\frac{\der^j}{\der z^j}$ are replaces with
the operators
$\D_k =  \sum_j \frac{a_{k,j}}{j!}\,\frac{\der^{j+N}}{\der z^{j+N}}$.
The relations \Ref{B}, \Ref{A=UY},
\Ref{Ablock} mean exactly that.
The lemma is proved.
\end{proof}

Let $A=(a_{ij})$ be a nondegenerate  $(N+\nu)\times (D+1)$-matrix. Let $y_0(x,t),\ldots, y_N(x,t)$ and
$\psi_0(x,t,z), \ldots, \psi_N(x,t,z)$ be the associated functions.
Let $m$ be a positive integer. Define the $(N+\nu)\times (D+1+m)$-matrix
$\hat A = (\hat a_{ij})$ by the formula
\bea
\hat a_{ij}
&=&
 a_{ij}, \qquad j\leq D,
\\
\hat a_{ij}
&=&
 a_{ij}, \qquad j > D.
\eea
We call $\hat A$ the $m$-extension of $A$. Let $\hat y_0(x,t),\ldots, \hat y_N(x,t)$ and
$\hat\psi_0(x,t,z), \ldots, \hat\psi_N(x,t,z)$ be the functions associated with $\hat A$.
Clearly, we have
\bean
\label{ext}
\hat y_n(x,t) = y_n(x,t),
\qquad
\hat\psi_n(x,t,z) =\psi_n(x,t,z), \qquad n=0,\dots,N.
\eean

\vsk.2>
Let $A=(a_{ij})$ be a nondegenerate  $(N+\nu)\times (D+1)$-matrix with associated functions
 $y_0(x,t),\ldots, y_N(x,t)$ and  $\psi_0(x,t,z), \ldots, \psi_N(x,t,z)$ which extend periodically
Let $\hat A$ be the $N$-extension of the matrix of $A$. According to \Ref{ext} the matrix $\hat A$
has the same associated functions
 $y_0(x,t),\ldots, y_N(x,t)$ and  $\psi_0(x,t,z), \ldots, \psi_N(x,t,z)$ which extend periodically.

\begin{lem}
\label{lem ext}
Under these assumptions the matrix $\hat A$ is given by the construction of Section \ref{sec con},
namely, we have
$\hat A = \hat A(\hat W)$ for a suitable $\hat W$.

\end{lem}

\begin{proof}
We have $\psi_N(x,0,z)=z^N \psi_0(x,0,z)$ and the function $\psi_N(x,0,z)$ is defined by the
$(N+\nu)\times (D+1+N)$-matrix $\hat A$.  The same function
$\psi_N(x,0,z)$ is defined also by the
$(N+\nu)\times (D+1+N)$-matrix
\beq
\label{BP}
P=\left(\begin{array}{cc}
E&0\\
0& A^{(0)}
\\
\end{array}\right),
\eeq
where  $E$ is the $N\times N$ unit matrix and $A^{(0)}$ is the
$\nu\times (D+1)$-matrix formed by the first $\nu$ rows of the matrix $A$.
By Theorem \ref{thm rec2} this means that $\hat  A=\hat W P$ for a suitable
matrix $\hat W$.  It remains to show that the upper-right $\nu\time \nu$ corner of $\hat W$, denoted in \Ref{S} by $U$ is nilpotent. As it was already noted above from equations \Ref{B}, \Ref{A=UY} and \ref{Ablock} it follows that the columns  $q_j$ of $A{(0)}$ should satisfy equation \Ref{q}. Since $A^{(0)}$ is of rank $\nu$ and $q_j=0$ for $j>D$ we get that $U^r=0$ for some $r$. If that holds for some $r$ then $r<\nu$. From the latter it follows that the integer $D$ used in the construction in the $N$ periodic case is bounded by $D\leq N\nu$.
\end{proof}

\begin{thm}\label{5.5}
If an $N$-periodic sequence of polynomials $(y_1^0(x),\dots,y_N^0(x))$
represents a solution of the Bethe ansatz equations
\Ref{bae}, then there exists a matrix $A=A(W)$ given by the construction
of Section \ref{sec con}
such that the associated polynomials $y_0(x,t),\ldots, y_N(x,t)$
 have the property:
\bean
\label{on=n}
y_n(x,0)=y^0_n(x), \qquad n=1,\dots, N.
\eean
\end{thm}
\begin{proof}
By Lemma \ref{nozeros} the function $\psi_0(x,z)$ corresponding to a sequence polynomials
 $(y_n(x))_{n\in\Z}$ representing a periodic solutions of the Bethe ansatz equation has the form
\beq\label{psiofinal}
\psi_0(x,z)=\psi(x,z\,|\, u^{(0)}, \ga^{(0)}) = (z+1)^x\left(1+\sum_{i=1}^{\nu}\xi_i^{(0)}(x)z^{-i}\right)
\eeq
with $\xi_\nu^{(0)}\neq 0$. The integer $\nu\leq k_0$
is the order of the pole of $\psi_0$ at $z=0$. By Theorem
\ref{thm:14} the function $z^{m_0}\psi_0(x,z)$ satisfies
\Ref{orthogonality} for any $g\in W_0(u^{(0)},\ga^{(0)})$. The space $W_0(u^{(0)},\ga^{(0)})$ is
a $k_0$-dimensional subspace of polynomials
 of degree $2k_0-1$. The function
$z^{k_0}\psi_0(x,z)$ has zero of order $k_0-\nu$ at $z=0$.
Then, by Lemma \ref{zerooder} the polynomials
 $z^{2k_0-1-\ell}$ for $\ell=0, \ldots, k_0-\nu-1$ are
 in $W_0(u^{(0)},\ga^{(0)})$. Hence, the space  $W_0(u^{(0)},\ga^{(0)})$ contains a $\nu$-dimensional
 subspace $\widetilde W^{(0)}\subset W_0(u^{(0)},\ga^{(0)})$
of polynomials of degree $k_0+\nu-1$ such that the function
\beq\label{psiopolynomial}
z^\nu\psi(x,z\,|\,u^{(0)},\ga^{(0)})=(z+1)^x\left(z^\nu+\sum_{i=1}^{\nu}\xi_i^{(0)}(x)z^{\nu-i}\right)
\eeq
satisfies the equations
\beq\label{psi0++}
\res_{z=0} \frac {g(z)z^{k_0}\psi_0(x,z)}{z^{2k_0}}=\res_{z=0} \frac {g(z)(z^\nu\psi_0(x,z))}{z^{k_0+\nu}}=0\,, \  \ \forall g\in \widetilde W^{(0)}\,.
\eeq
Choose a basis $g_k(z), k=1,\ldots \nu,$ in the space $\widetilde W^{(0)}$. The coefficients $a_{k,j}$ of these polynomials
\beq\label{gbasis}
g_k(z)=\sum_{j=0}^{m_0+\nu-1} a_{k,j}z^{m_0+\nu-j-1}
\eeq
define a $\nu\times (k_0+\nu)$-matrix $\tilde A^{(0)}$, which for any $D\geq k_0+\nu-1$
can be trivially extended to a $\nu\times (D+1)$-matrix $A^{(0)}$ by setting $a_{k,j}=0\,, j\geq k_0+\nu$.
Then equations \Ref{psi0++} coincide with equations \Ref{cusd} defining the Baker-Akhiezer function $\psi_0(x,0,z\,|\,A^{(0)})$, where we have
included in the notation the dependence of the Baker-Akhiezer function on the defining matrix $A^{0}$, i.e.
\beq\label{psi+BApsi}
z^\nu \psi_0(x,z)=\psi_0(x,0,z\,|\,A^{(0)}).
\eeq
Applying recurrently equation \Ref{laxdd} we get that for $n\geq 0$ the solution of the linear generating problem has the form
\beq\label{psi-n}
z^\nu \psi_n(x,z)=(z+1)^x\left(z^{n+\nu}+\sum_{i=1}^{n+\nu}\xi_i^{(n)}(x)z^{n+\nu-i}\right)\,.
\eeq
Since $z^\nu\psi_n(x,t)=z^{n+\nu}\psi(x,z\,|\,u^{(n)},\ga^{(n)})$, we a'priory know
that the coefficients $\xi_i^{(n)}(x)$ are defined by a nondegenerate system
of equations of the form \Ref{laxdd} defined by an $(n+\nu)\times (D+1)$ matrix
$A^{(n)}$ for sufficiently large $D$. From \Ref{laxdd} it follows
that $z^\nu \psi_{n+1}(x,z)$ satisfies the system of $(n+\nu)$ linear
equations defining $z^\nu\psi_n(x,z)$. Hence, $A^{(n+1)}$ can be chosen
such that its first $(n+\nu)$ rows coincides with the matrix $A^{(n)}$. Then
we define $A$ in the construction of Section 5 to be equal to $A^{(N)}$. Theorem \ref{5.5}
is proved.
\end{proof}

\subsection{Remark on difference operators}

In formula \Ref{diff} we identified, roughly speaking, the Baker-Akhiezer function
$\psi_n(x,0,z)$ with the linear difference operator of order $n+\nu$, whose kernel is spanned by the polynomials
$f_1(x,0)$, \dots, $f_{n+\nu}(x,0)$.  From that point of view, the Baker-Akhiezer function $\psi_{n+1}(x,0)$ is identified with
the linear difference operator of order $n+\nu+1$, whose kernel is spanned by the polynomials
$f_1(x,0)$, \dots, $f_{n+\nu}(x,0)$ and one new polynomial $f_{n+\nu+1}(x,0)$.
The main formula  of this paper, that is, formula  \Ref{laxdd},
is the formula  expressing the second of these difference operators in terms of the first one.
The periodicity property of the functions $\psi_0(x,t,z)$,
\ldots, $\psi_N(x,t,z)$ can be reformulated as a special relation between the
kernels of the differential operators corresponding to $\psi_0(x,t,z)$ and $ \psi_N(x,t,z)$. That property is implicitly explained in Sections \ref{sec per} -- \ref{sec Prop}.  A version of this point of view is developed in Section \ref{sec TB}.

\subsection{Main theorem on commuting flows}

By Theorem \ref{5.5} any solution of the $N$-periodic Bethe ansatz equations is defined by some matrix $A$.
 Theorem \ref{thm BAk} implies that the space of solutions of the $N$-periodic Bethe ansatz
 equations is invariant with respect to times $t$
 under the deformations defined by the Baker-Akhiezer functions.
 For each $n$ the corresponding function $\Psi_n(x,t,z)$ is a particular case of the
 Baker-Akhiezer function corresponding to the rational
 $k_n$-particle rational RS system. Hence, by Theorem \ref{6.4} the dependence of roots of the corresponding
 polynomial $y_n(x,t)$ is described by equations of the rational RS system. Therefore we have the following theorem.

\begin{thm}\label{thm:imbed}

Let $(y_n(x))_{n\in\Z}$ be an $N$-periodic
sequence of polynomials of degrees $(k_n)$ representing a
solution of the $N$-periodic Bethe ansatz equations. The correspondence
\beq
\label{embedding}
(y_n)\longmapsto (u^{(n)},\ga^{(n)}),
\eeq
where $\ga=(\ga_1^{(n)}, \dots, \ga_{k_n}^{(n)})$
is given by  \Ref{dan}, is an embedding of the space of solutions
 of the Bethe ansatz equations into the product of $N$ phase spaces of
the $k_n$-particle rational RS systems, $n=1,\ldots, N$.  The image of this
map is invariant under the hierarchy
 the rational RS systems \Ref{meq1}, \Ref{meq2} acting diagonally on the product of the phase spaces.

Consider the extension of  the sequence $y=(y_n(x))_{n\in\Z}$ to the
family $y(t)=(y_n(x,t))_{n\in\Z}$, defined by Theorems \ref{thm BAk} and \ref{5.5}.
Then the
 correspondence in \Ref{embedding} sends the family $y(t)$
to a solution of the rational RS hierarchy.
\qed
\end{thm}

\section{Bethe ansatz equations and integrable hierarchies}
\label{S:flows}

The existence of the one parameter family $\Psi(z)$ of solutions of equations \Ref{laxdd}
having the form \Ref{bakdd} reveals the connection of the Bethe ansatz equations \Ref{baei} with
basic hierarchies of the soliton theory. We begin this section with a brief review of the hierarchy,
which is referred throughout the paper as the positive part of the $2D$-Toda hierarchy.

\subsection{Pseudo-difference operators}

We regard sequences $g=(g_n)_{n\in\Z}$ with $g_n\in\C$
as elements of the ring of functions of the discrete variable $n$. In particular we have addition
$f+g$  and multiplication $fg$ of sequences defined by the formulas $(f+g)_n=f_n+g_n$,
$(fg)_n=f_ng_n$. Let $T$ be the shift operator acting on sequences
$g=(g_n)_{n\in\Z}$ by  the formula $T: f\mapsto Tg$, where
$(Tg)_n= g_{n+1}$.

The space of pseudo-difference  operators is the space $\F$ of Laurent polynomials in $T^{-1}$,
 whose coefficients are functions of the variable $n\in \mathbb Z$, i.e.
\beq\label{Dpseudo}
F=\sum_{s=-M}^{\infty} f_s T^{-s}\,, \qquad f_s=(f_{n,s}),\qquad n\in\Z,
\eeq
for some integer $M$.
Recall that the coefficient $f_0$ in \Ref{Dpseudo} is called the {\it residue} of $F$,
\beq\label{resdef}
\res_T\, F\,:=\,f_0\,.
\eeq
The ring structure on $\F$ is defined by the ring structure on the space of coefficients and the composition rule
\beq\label{multiplic}
T(f T^m):= (Tf)T^{m+1},
\eeq
where $f$ is a sequence.

In what follows we will apply the pseudo-differential operators to sequences $\,\phi(z)=(\phi_{n}(z))_{n\in\Z}$, whose elements
are formal Laurent series in $z$ of the form
\beq\label{modul}
\phi_n(z)=z^n\left(\sum_{s=-K}^\infty\phi_{n,s} z^{-s}\right),
\eeq
where $K$ is some integer.

\subsection{Positive part of the $2D$ Toda hierarchy}
\label{sec 8.2}

The difference analog of the KP hierarchy is defined almost verbatim to the
definition in the continuous case, cf. \cite{SW}, \cite{dikii}. It leads us to the positive
 part of the $2D$ Toda hierarchy.

\vsk.2>
Consider
the affine space of monic pseudo-difference operators of
degree $1$, i.e., the space of pseudo-difference operators of the form
\beq\label{calL}
\LL=T+\sum_{s=0}^\infty w_{s}T^{-s}.
\eeq
The positive part of the $2D$ Toda hierarchy has time variables $t=(t_1,t_2,\dots)$.
The flow corresponding to the time variable $t_m$
 is defined by the equation
\beq\label{flows}
\p_m \LL=[\LL^m_+,\LL],\qquad \p_m:=\p_{t_m},
\eeq
where $\LL^i_+$ is the nonnegative part of the operator $\LL^m$, i.e. the difference operator such that
$\LL^m_-=\LL^m-\LL^m_+=\mc O(T^{-1})$.

The standard arguments show that \Ref{flows} is a well-defined system of equations on
the coefficients of the operator $\LL$. For that one needs to show
that the right-hand side of \Ref{flows} is a pseudo-difference operator of
degree at most zero. That follows from the equality $[\LL^m_+,\LL]=-[\LL_-^m, \LL]$ and the fact
that $\LL^m_-$ is a pseudo-difference operator of degree $\leq -1$.

The flows commute.
The proof of the commutativity of the flows, i.e. the proof that equations
 \Ref{flows} imply the equations
\beq\label{hierarchy}
[\p_m-(\LL^m_+),\,\p_\ell-(\LL^\ell_+)]=0,
\eeq
is standard and word by word follows its continuous variant, see \cite{dikii}.

\begin{rem} The hierarchy of commuting flows \Ref{flows} is a part of the
$2D$ Toda hierarchy. Recall that the full $2D$ Toda hierarchy is defined on the space of pairs of
pseudo-difference operators, one of which is as in \Ref{calL} and the other
is a pseudo-difference operator with respect to $T$,
\beq\label{calL-}
\LL^-\,=\,\sum_{s=-1}^\infty w_{s}^-\,T^{s}\,.
\eeq
The full set of time variables of the $2D$ Toda hierarchy  consists of
the variables $t=(t_1,t_2,\ldots )$  as above and the variables $t^-=(t_1^-,t_2^-,\ldots )$.
We do not give further details, see \cite{takasaki},
since the second part of the $2D$ Toda hierarchy is not relevant for our purposes.
\end{rem}

For any pseudo-difference operator $\LL$ of the form \Ref{calL}
there is a unique formal solution $\Psi^w(z)=(\Psi^w_n(z))_{n\in\Z}$ of the equation
\beq\label{may12}
\LL \Psi^w(z)=z \Psi^w(z)
\eeq
of the form
\beq\label{psiwave}
\Psi^w_n(z)=z^n\left(1+\sum_{s=1}^{\infty} \xi_{n,s} z^{-s}\right)\,,
\eeq
normalized by the condition
\beq\label{normpsi}
\Psi_0^w(z)=1 \quad\Leftrightarrow \quad\xi_{0,s}=0, \ \ s>0.
\eeq
The solution $\Psi^{w}(z)$ is called the {\it wave solution}.

\vsk.2>
Let the pseudo-difference operator $\LL$ depend on times,  $\LL=\LL(t)$.
One can check that this
pseudo-difference operator
  is a solution of the hierarchy equation \Ref{flows} if and only if the following equations hold:
\beq\label{Psiflow}
\p_m\Psi^w(t,z)=\LL^m_+(t) \Psi^w(t,z)+h_m(t,z)\Psi^w(t,z),
\eeq
where $h_m(t,z)$ is a scalar (not a sequence) Laurent series in $z$.
The comparison of the right and left-hand sides shows that $h_m(t,z)$ has the form
\beq\label{fm}
h_m(t,z)=z^m+\mc O(z^{-1}).
\eeq
From equation \Ref{hierarchy} it follows that
\beq\label{timenormalization}
\p_m h_\ell(t,z)\,=\,\p_\ell h_m(t,z)\,.
\eeq
Hence, there is a unique Laurent series $h(t,z)$ such that $\p_m h(t,z) =h_m(t,z)$
 and normalized by the condition $h(0,z)=1$. Then equation \Ref{fm} implies that
\beq \label{h}
h(t,z)=\sum_{m=1}^\infty t_mz^m+\mc O(z^{-1}).
\eeq
It is easy to see that the sequence $\Psi(t,z):=\Psi^w(t,z) e^{-h(t,z)}$ satisfies the equations
\beq\label{eigentime}
\LL(t)\Psi (t,z)=z\Psi(t,z), \qquad (\p_m-\LL^m_+)\Psi(t,z)=0\,.
\eeq
The elements $\Psi_n(t,z)$ of the sequence $\Psi(t,z)$ have the form
\beq\label{psi100}
\Psi_n(t,z)=z^n\left(1+\sum_{s=1}^\infty \chi_s(t)z^{-s}\right)e^{\sum_{m=1}^\infty t_m z^{m}}.
\eeq

\subsection{Discrete  $N$ mKdV hierarchy}

Consider sequences of functions $g=(g_n(x))_{n\in \mathbb Z}$. There are two shift operators acting
on them: $T$ and $T_x$. The action of $T$ is as above. The operator $T_x=e^{\p_x}$ is the shift in
the  $x$ variable,  $(T_xg)_n(x) =g_n(x+1)$.

\vsk.2>
Recall the generating equation \Ref{intBA},
 that can be written in the form
\beq\label{4.2may}
H \Psi\,=\,0\,,
\eeq
where
\beq\label{H}
H=T-T_x+v\,, \qquad v=(v_n(x))_{n\in\Z}\,,
\eeq
is a difference operator in $x$ and $n$.

\vsk.2>

The hierarchy, which we call the
{\it discrete $N$ mKdV hierarchy},
 is the compatibility condition of
 the positive part of the $2D$ Toda hierarchy, defined in \Ref{flows}, with  the generating equation \Ref{4.2may}.
 More precisely, the full set of equations of the
{ discrete  $N$ mKdV hierarchy} are equations \Ref{flows} and equations
\beq\label{may23a}
[\p_m - \LL^m_+, H]\,=\,D_mH,\qquad m\geq 1,
\eeq
where $D_m$ is some difference operator in $x$ and $n$ depending on $m$.

\begin{rem} The meaning of \Ref{may23a} is that the
operators $\p_m-\LL^m_+$ and $H$ commute on the space of
solutions of equation \Ref{4.2may}. In the theory of integrable systems this type of representation is  called
an $L,A,B$ triple, see \cite{dkn}.

By division with remainder it is easy to see
 that any difference operator ${\mathcal D}$  in $x$ and $n$ of degree $M$ has a unique presentation
\beq\label{DmodH}
\mathcal D=DH+D_1\,,
\eeq
where $D_1$ is a degree $M$ difference operator  in $n$ only,
i.e. $D_1$
is a polynomial of degree $M$ in $T$ with coefficients that are sequences of functions $(g_n(x))$. Equation \Ref{may23a} says that the corresponding operator $D_1$ equals
 zero. Therefore, for any given monic difference operator $B$ in $n$,
\beq\label{BH}
B=T^M+\sum_{s=1}^M b_s  T^{M-s},
\eeq
the equation
\beq\label{july3}
[\p_m-B, H]=DH
\eeq
with some $D$ is a system of $M+1$ equations on
 $M+1$ unknown coefficients $b_1,\ldots, b_M$ and $v$. The first $M$ of them
are difference equations. Unlike in
the differential case, where the corresponding equations allow us
to express the coefficients $b_1,\ldots, b_M$
 as the differential polynomials in $v$ and get a well-defined system
 of equations for the coefficients of $H$ only, in the difference case the reconstruction of $b_1,\ldots, b_M$
  in terms of $v$ requires some additional assumptions, see more on that below.
\end{rem}

Equations \Ref{flows} and \Ref{may23a} is a system of equations on the coefficients of the
pseudo-difference
operator $\LL$ and the sequence $v$. These equations can be written more explicitly following
the argument identical to the one in the proof of equation (5.66) in \cite{K4}. Namely, let
\beq
\label{resT}
F_m\,:=\,\res_T \,\LL^m, \qquad  F_m= (F_{m,n})_{n\in\Z}\,.
\eeq

\begin{lem}\label{vn}
The system consisting of equations  \Ref{flows} and  \Ref{may23a} is equivalent to the system consisting of equations
 \Ref{flows} and equations
\beq
\label{vflow}
\p_m (\ln\,v_n(x))=F_{m,n}(x)-F_{m,n}(x+1)\,, \qquad m\geq 1\,.
\eeq
\qed
\end{lem}

\subsection{Remark}
Notice
again that the system of equations \Ref{vflow} is not a closed system
with respect to $v(x)$ since
the right-hand sides are expressed in terms of the operator $\LL$.

A possible approach to eliminate $\LL$ from equations \Ref{vflow} is as follows.
Having an arbitrary $N$-periodic $v(x)$  determine a family of solutions $\psi(x,z)$
of equation the $H\psi=0$. Then
$\LL$ is uniquely determined
from the equation $\LL \psi(x,z) = z\psi(x,z)$. Put that $\LL$ into  \Ref{vflow} and
obtain
a system of equations on $v(x)$ only. Such an approach works well in similar situations but not in this one since
the desired family of solutions $\psi(x,z)$ to equation $H\psi=0$ is not unique.

Below we explain a construction of $\psi(x,z)$ from $v(x)$ and indicate why $\psi(x,z)$
is not unique.
That
fails this attempt to eliminate $\LL$ from equations \Ref{vflow}. The problem of elimination  of $\LL$
from \Ref{vflow}  deserves further analysis.

\begin{lem} Let $v=(v_n(x))_{n\in\Z}$ be any $N$-periodic sequence of functions, $v_n(x)=v_{n+N}(x)$.
Then there is a formal solution $\psi=(\psi_n(x,z))_{n\in\Z}$
of equation \Ref{intBA},
\beq\label{eqH}
H\psi=0
\eeq
with $\psi_n(x,z)$ of the form
\beq\label{psijune}
\psi_n(x,z)=z^n(z+1)^x\left(1+\sum_{s=1}^\infty \xi_{n,s}(x)z^{-s}\right)
\eeq
with periodic coefficients
\beq\label{xiperiodic}
\xi_{n,s}(x)=\xi_{n+N,s}(x)\,.
\eeq
\end{lem}

\begin{proof}
The substitution of \Ref{psijune} into \Ref{eqH} gives a system of equations for the unknown coefficients $\xi_{n,s}(x)$ in \Ref{psijune}
\beq\label{xieq0}
(T_x-T)\,\xi_{s+1}\,=\,-\,(v+T_x)\,\xi_s\,,
\eeq
i.e.,
\beq\label{xieq}
\xi_{n,s+1}(x+1)-\xi_{n+1,s+1}(x)=-v_{n}(x)\,\xi_{n,s}(x)-\xi_{n,s}(x+1)\,,\qquad s=1,2, \ldots \,.
\eeq

We prove the existence of $N$-periodic solutions of these equation by induction.
 The induction starts with $\xi_0 = (\xi_{n,0})_{n\in\Z}$ and $\xi_{n,0}=1$ for all $n$.
  Suppose that $\xi_s=(\xi_{n,s})$ is known and is $N$-periodic.
 Let us apply  the operator
 $\mathcal T_N:=\sum_{i=0}^{N-1} T_x^{N-i-1}T^{i}$ to  both sides of \Ref{xieq0}.
Using the periodicity of $\xi_s$ and $v$ we get the equation
\beq\label{xieq1}
(T_x^N-1)\xi_{s+1}={\mathcal T}_N(T_x\xi_s-v)\,.
\eeq
Invert the operator $T_x^N-1$,
\beq\label{Tinv}
(T_x^N-1)^{-1}:= \sum_{i=1}^\infty T_x^{-iN} .
\eeq
Then the $N$-periodic solutions of \Ref{xieq0} can be recurrently defined by the formula
\beq\label{xis+1}
\xi_{s+1}=(T_x^N-1)^{-1}{\mathcal T}_N(T_x\xi_s-v)\,.
\eeq
The lemma is proved.
\end{proof}

The choice of $(T_x-1)^{-1}$ is not unique. It can be replaced by
\beq\label{Tinv1}
(T_x^N-1)^{-1}:= -\sum_{i=0}^\infty T_x^{iN}\,.
\eeq

It is easy to see that for any formal solution $\psi(x,z)$ of \Ref{eqH} of the form \Ref{psijune} there is a unique pseudo-difference operator $\LL$ such that
\beq\label{eqjuly}
\LL\psi(x,z)\,=\,z\psi(x,z)\,.
\eeq
 Hence any choice of such a
 $\psi(x,z)$ makes the discrete $N$-periodic mKdV equations a well-defined system of equations for the functions $(v_n(x))_{n\in\Z}$ only.

\vsk.2>
Notice that if a sequence
$(v_n(x))$ is not an arbitrary $N$-periodic sequence of functions, but a sequence
defined by formula \Ref{vpot} with $(y_n(x))$ satisfying the Bethe ansatz equations,
then
 Theorem \ref{gener1} gives us another way to construct the family
 of solutions $\psi(x,z)$ to equation $H\psi(x,z)=0$.
In that case by constructing $\LL$ from \Ref{eqjuly} we may eliminate $\LL$ from \Ref{vflow}
and then solve the resulting equations on $(v_n(x))$ only.

\subsection{Deformations are solutions}

Let $y=(y_n(x))$ be an $N$-periodic sequence
of polynomials representing a solution of the Bethe ansatz equations.
By Theorems \ref{thm BAk} and \ref{5.5} we can extend $y$ to a family
$y(t)=(y_n(x,t))$.
Consider the corresponding solution of the generating problem $(\Psi_n(x,t,z))$,
\beq\label{psivse}
\Psi_n(x,t,z)=z^n\left(1+\sum_{s=1}^\infty \xi_s(x,t)z^{-s}\right)\Omega(x,t,z).
\eeq
This solution satisfies the hierarchy of linear equations \Ref{h2D}. Equations \Ref{d-congruence} identify the difference operators $D_m$ in \Ref{h2D} with the operators $\LL^m_+$. Hence we have the following theorem.

\begin{thm}
The  $N$-periodic sequence $(v_n(x,t))$ defined in terms of $y(t)$ by \Ref{vpot} is a solution the discrete  $N$ mKdV hierarchy.
\qed
\end{thm}

\subsection{Remark on discrete Miura opers}
\label{sec MO}
Denote by $L(z)$, $V(x)$ the $N\times N$-matrices
\bean
\label{def L}
L(z)
&=&
 E_{2,1} + E_{3,2} + \dots + E_{N,N-1} + z^{-N} E_{1,N}\,,
\\
\notag
V(x)
&=&
 v_1(x)E_{1,1} + \dots + v_N(x)E_{N,N}\,,
\eean
where $ v_1(x),\dots, v_N(x)$ are some given functions of $x$.
The first order linear difference operator
\bean
\label{Miura}
T - L -V
\eean
is called a {\it discrete Miura oper}, cf. \cite{MV3}.
Assume  that $(y_n(x))_{n\in\Z}$ is  an $N$-periodic sequence of polynomials
 representing a solution of the Bethe ansatz equations \Ref{bae},
 $y_{N+n}(x)=y_n(x)$.  Consider the corresponding
$N$-periodic sequence $(v_n(x))$ defined by formula \Ref{vpot} and the $N$-periodic sequence of
Baker-Akhieser functions $(\Psi_n(x,z))_{n\in\Z}$ given by Theorem \ref{gener1}, $\Psi_{N+n}(x,z)=z^N \Psi_n(x,z)$.
Consider the $N$ column vector  $\Psi(x,z)$ with coordinates
$\Psi_1(x,z),\dots,\Psi_N(x,z)$. Then
\bean
\label{BA Miura}
(T_x - L(z) -V(x) )\,\Psi(x,z)\,=\,0\,.
\eean

For example, if $N=3$, then
\bea
\left(
\begin{matrix}
\Psi_1
\\
\Psi_2
\\
\Psi_3
\end{matrix}
	\right)\!\!(x+1,z)
=
\left(\begin{matrix}
v_1(x) & 0 &z^{-3}
\\
1 & v_2(x) & 0
\\
0& 1& v_3(x)
\end{matrix}
	\right) \!\!
\left(\begin{matrix}
\Psi_1
\\
\Psi_2
\\
\Psi_3
\end{matrix}\right)\!\!(x,z)\,.
\eea
Our study in this paper of periodic sequences of $(y_n(x))_{n\in\Z}$, $(\Psi_n(x,z))_{n\in\Z}$
is the study of the difference equation
\Ref{BA Miura}.
\vsk.2>

The discrete Miura opers are discrete analogs of differential Miura opers, which are the first order differential operators
of the form
\bean
\label{dMO}
\frac{d}{dx} -\Lambda - V\,.
\eean
These differential operators
play an important role in the theory of the $N$ mKdV hierarchy, see for example \cite{DS, VWr}.

\section{Combinatorial data}
\label{sec COD}

In this section we follow Section 6 in \cite{VWr} and
review some combinatorial data, which will be used in Section \ref{sec TB}
to describe Baker-Akhieser functions of points of an infinite-dimensional Grassmannian.

\subsection{Subsets of virtual cardinal zero}
\label{Subsets of virtual cardinal zero}

By a {\it partition} we  mean an infinite sequence of nonnegative integers
$\la=(\la_0\geq\la_1\geq \dots)$ such that all except a finite number of the $\la_i$ are zero.
The number $|\la|=\sum_i\la_i$ will be called the {\it weight of} $\la$.

\vsk.2>

Following \cite{SW}, we say that a subset
 $S=\{s_0<s_1<s_2<\dots\} \subset \Z$ is of {\it virtual cardinal zero},
if $s_j=j$ for all sufficiently large $i$.
If $n$ is such that $s_j=j$ for all
 $j>n$, then we say that $S$ is of {\it depth $n$}.

\vsk.2>

  If $S$ is of depth $n$, then it is also of depth $n+1$.

\begin{lem} [\cite{SW}]
\label{lem SiW}
There is a one to one correspondence between elements of $\mc S$ and partitions, given by
$S \leftrightarrow\la$ where
\bea
\la_i\, =\, i-s_i\,.
\eea

\end{lem}

For a subset $S=\{s_0<s_1<s_2<\dots\} \subset \Z$  and an integer $k\in \Z$ we denote
by
$S+k$ the subset $\{s_0+k<s_1+k<s_2+k<\dots\} \subset \Z$.

\vsk.2>
Let $S$ be a subset of virtual cardinal zero. Let $A=\{a_1,\dots,a_k\} \subset \Z$ be a finite
subset of distinct integers.

\begin{lem}
[\cite{VWr}]
\label{lem shift}
If $\{a_1,\dots,a_k\} \cap (S+k) = \emptyset$. Then
$\{a_1,\dots,a_k\} \cup (S+k)$ is a subset of virtual cardinal zero.

\end{lem}

\subsection{KdV subsets}
\label{Sec KdV subsets}

Fix an integer $N>1$. We say that a subset $S$ of virtual cardinal zero is a
{\it KdV subset} if
$S+N\subset S$. For example, for any $N>1$,
\bea
S^\emptyset = \{0,1,2,\dots\}
\eea
 is a KdV subset.

\begin{lem} [\cite{VWr}]
\label{lem leading term}
Let $S$ be a  KdV subset. Then there exists a unique $N$-element subset
$A=\{a_1< \dots <a_N\}\subset \Z$ such that
$S = A\cup (S+N)$.

\end{lem}

The subset $A$ of the Lemma \ref{lem leading term} will be called the
{\it  leading term of}  $S$.

\vsk.2>
The leading term $A$ uniquely determines the KdV subset $S$, since
 $S$  is the union of $N$ non-intersecting arithmetic progressions
$\{a_i, a_i+N, a_i + 2N, \dots\}$, $i=1,\dots,N$.

\vsk.2>

Let $S$ be a KdV subset with leading term $A$. For any $a\in A$ the subset
\bean
\label{j mut}
S[a] = \{a+1-N\} \cup (S+1)
\eean
is a KdV subset with leading term $A[a] = (A+1) \cup \{a + 1-N\} -\{a+1\}$.
The subset $S[a]$ will be called the {\it mutation of the KdV subset} $S$ at $a\in A$.

\begin{lem}
[\cite{VWr}]
\label{lem two step}
${}$
\begin{enumerate}
\item[(i)]

Let $S_1$ be a  KdV subset with leading term $A$. Let $S_2$ be a KdV subset such that $S_1+1\subset S_2$.
Then $S_2$ is the mutation of $S_1$ at some element $a\in A$.

\item[(ii)]

Any KdV subset $S$ can be transformed to the KdV subset $S^\emptyset$ by a sequence of mutations.

\item[(iii)]

A subset $A=\{a_1<\dots<a_N\}$ is the leading term of a KdV subset if and only if
equation
\bean
\label{sum lead}
\sum_{i=1}^N a_i = \frac {N(N-1)}2
\eean
 holds true and $a_i-a_j$ is not divisible by $N$ for any $i\ne j$.

\end{enumerate}
\end{lem}

\subsection{mKdV tuples of subsets}
\label{sec mkdv tuples of subsets}

 We say that an $N$-tuple $ S = (S_1,\dots,S_N)$ of KdV subsets is an {\it mKdV tuple of subsets}
 if $S_i+1\subset S_{i+1}$ for all $i$, in particular, $S_N+1\subset S_1$.
\vsk.2>

For example, for any $N$, the $N$-tuple
\bea
 \mc S^\emptyset = (S^\emptyset,\dots,S^\emptyset)
\eea
is an mKdV tuple of subsets.

\vsk.2>
If  $ S = (S_1,\dots,S_N)$  is an
mKdV tuple, then  $(S_i,S_{i+1},\dots,S_N,$ $ S_1,S_2,\dots,S_{i-1})$
is an mKdV tuple of subsets for any $i$.

\vsk.2>

Let $S$ be a  KdV subset with leading term $A=\{a_1<\dots<a_N\}$. Let $\sigma$ be
an element of the permutation group $\Si_N$.
Define an $N$-tuple $ S_{S,\sigma} = (S_1,\dots,S_N)$,
where
\bean
\label{mkdv tuple}
\phantom{aaa}
S_i = \{a_{\sigma(1)}+i-N, a_{\sigma(2)}+i-N, \dots, a_{\sigma(i)}+i-N\} + (S+i),
\quad i=1\dots,N.
\eean
In particular, $S_N = A \cup (S+N) = S$.

\begin{lem}[\cite{VWr}]
\label{lem const}
${}$

\begin{enumerate}
\item[(i)]
The $N$-tuple $ S_{S,\sigma}$ is an mKdV tuple.

\item[(ii)]
Every  mKdV tuple is of the form $ S_{S,\sigma}$ for some KdV subset $S$ and some element
$\sigma\in\Si_N$.

\end{enumerate}
\end{lem}

\subsection{Mutations of  mKdV tuples }
\label{Mutations of  mKdV tuples }

\begin{lem}
[\cite{VWr}]
\label {lem constr}
Let $ S=(S_1,\dots,S_N)$ be an mKdV tuple.
Then for any $i=1,\dots,N$, there exists a unique mKdV tuple
\bean
\label{new mkdv tuple}
 S^{(i)}=(S_1,\dots,S_{i-1},\tilde S_i, S_{i+1},\dots,S_N)
\eean
which differs from $ S$ at the $i$-th position only.

\end{lem}

The mKdV tuple $ S^{(i)}$ will be called the {\it mutation} of the mKdV tuple $ S$ at the $i$-th position.
Denote by $w_i :  S \mapsto  S^{(i)}$ the mutation map.

\vsk.2>

Let   $\la^{i}, $ $\tilde \la^i$  be the partitions corresponding to the KdV subsets
$S_{i}, $ $\tilde S_i$, respectively.
The mutation $w_i :  S \mapsto  S^{(i)}$ will be called {\it degree decreasing} if
$|\tilde \la^i|<|\la^i|$.

\begin{thm}
[\cite{VWr}]
\label{thm mutation all}

Any mKdV tuple $ S$ can be transformed to the mKdV tuple $\mc S^\emptyset = (S^\emptyset,\dots,S^\emptyset)$
by a sequence of degree decreasing mutations.
\end{thm}

\section{Tau-functions and Baker-Akhieser functions}
\label{sec TB}

In this section we follow Section 7 in \cite{VWr} although we define the tau-functions
as discrete Wronskians while in \cite{VWr} the standard Wronskians are used.
The tau-functions in this paper are different from  the tau-functions in \cite{VWr}.

\subsection{Remarks on the construction of Section \ref{S:inverse}}
\label{sec Ri}

In Section \ref{sec TB} below we assign tau-functions and Baker-Akhieser
functions to vector subspaces of an infinite dimensional vector space. The assignment
is based on the construction of Section \ref{S:inverse}. We formulate two remarks on the construction.
\vsk.2>

In Section \ref{sec BA} starting from an $(\nu+N)\times (D+1)$-matrix
\bea
A=\{a_{k,j}\}\,, \qquad k=1,\ldots, N+\nu, \quad j=0,\ldots, D\,,
\eea
we constructed
the functions $y_n(x,t)$,   $\psi_n(x,t,z)$ for $n=0,\dots,N$.

\vsk.2>
Choose $n$, $0\leq n\leq N$. Consider the $(n+\nu)\times D$-matrix $A^{(n)}$ formed
by the first $n+\nu$ rows of the matrix $A$. Then the functions $y_n(x,t)$ and $\psi_n(x,t,z)$
are determined by formulas \Ref{formd} and \Ref{psiwdM} in terms of the matrix $A^{(n)}$ only.

Let $B$ be a nondegenerate $(n+\nu)\times (n+\nu)$-matrix.
Let $y_{n,B}(x,t)$ and $\psi_{n,B}(x,t,z)$ be the functions determined  by formulas \Ref{formd} and \Ref{psiwdM}, respectively,
 in which the entries of
the matrix $A^{(n)}$ are replaced with the corresponding
entires of the matrix $BA^{(n)}$. Then $y_{n,B}(x,t)=y_{n}(x,t)$ and $\psi_{n,B}(x,t,z)=\psi_{n}(x,t,z)$.
That is, the functions $y_{n}(x,t)$ and $\psi_{n}(x,t,z)$ are determined by the
$(n+\nu)$-dimensional vector space
spanned by the first $n+\nu$ rows of the matrix $A$ and do not depend on the choice of a basis in that space.

\vsk.2>

Consider the new $(\nu+N+1)\times (D+2)$-matrix
\bea
\tilde A=\{\tilde a_{k,j}\}\,, \qquad k=0,\ldots, N+\nu, \quad j=0,\ldots, D+1\,,
\eea
defined by the formulas
\bean
\label{ext A}
\tilde a_{0,j}
&=&
 \delta_{0,j}\,,\qquad j=0,\dots, D+1\,,
\\
\notag
\tilde a_{k,0}
&=&
0\,, \qquad
\phantom{aa} k=1,\dots, N+\nu\,,
\\
\notag
\tilde a_{k,j}
&=&
 a_{k,j-1}, \qquad j=1,\dots, D+1.
\eean
Apply  the construction of Section \ref{sec BA}
to the matrix $\tilde A$ and construct
the functions  $\tilde y_n(x,t)$ and $\tilde \psi_n(x,t,z)$ for $n=0,\dots,N$.

\begin{lem}
\label{lem stab}
We have
\bean
\label{ps-tps}
\tilde y_n(x,t)=y_n(x,t),\qquad
\tilde \psi_n(x,t,z) \,=\,z\,\psi_n(x,t,z)\,,\qquad n=0,\dots,N\,.
\eean
\qed
\end{lem}

Lemma \ref{lem stab} says  that the functions $y_{n}(x,t)$ and $\psi_{n}(x,t,z)$,
determined by the
$(n+\nu)$-dimensional vector space
spanned by the first $n+\nu$ rows of the matrix $A$, do not change up to  multiplication
of $\psi_n(x,t,z)$ by $z$, if  the
$(n+\nu)$-dimensional vector space  is extended to the
$(n+\nu+1)$-dimensional vector space  by formulas \Ref{ext A}.

\vsk.3>

\subsection{Grassmannian $\Gr$}
\label{SUBS}
For a Laurent polynomial  $v = \sum_i  v_{i} z^i$,
the number
 $\ord \,v \,=\, \text{min} \{ i : v_{i} \neq 0 \}$ will be called
the {\it order} of $v$.

\vsk.2>

Following \cite{SW},
let $H$ be the Hilbert  space $L^2(S^1)$ with
 orthonormal basis $\{z^j\}_{j\in \Z}$.  Let $H_+$ be the closure of
the span of $\{z^j\}_{j \geq 0}$ and $H_-$ the closure of the span of
$\{z^j\}_{j < 0}$. We have the orthogonal decomposition $H=H_+\oplus H_-$.

\vsk.2>

We consider the set of all  closed subspaces $W \subset H$ such that
\bean
\label{-k k}
z^q H_+ \subset W \subset z^{-q} H_+
\eean
for some $q > 0$. Such subspaces can be identified with subspaces
$W/z^qH_+$ of
$ z^{-q}H_+/z^qH_+$.
We say that  $W$ is of {\it virtual dimension zero} if
$\dim W/z^qH_+ =q$.  Denote by $\Gr$ the set of all subspaces of
virtual dimension zero.

\vsk.2>

Any $W\in\GR$ has a basis $\{v_j\}_{j\geq 0}$ consisting of Laurent polynomials.
We may assume that the numbers $s_j = \ord\,v_j$
form a strictly  increasing sequence
$S_W =\{s_0<s_1<s_2<\dots\}$.
The assignment $W\mapsto S_W$ is well-defined. The  subset $S_W$ will be called the
 {\it order subset} of  $W$.
The order subset $S_W$ is  of virtual cardinal zero.

\vsk.2>

For $W\in\Gr$,
a basis $\{v_j=\sum_{i\geq s_j}v_{j,i}z^i\}_{j\geq 0}$  of $W$ is
called {\it special of depth $n$},\  if it
 consists of Laurent polynomials such that
 $v_j=z^j$ for $j> n$ and $v_{j,i}=0$ if $i>n$ and $j\leq n$.

\vsk.2>
If $ \{v_j\}_{j\geq 0}$ is a basis of depth $n$, then it is also a basis of depth
$n+1$.

\subsection{Points in $\Gr$ and finite-dimensional spaces of polynomials in $x,t$}
\label{subs in all}

Let $S=\{s_0<s_1<\dots \}$  be a set of virtual cardinal zero of depth $n$.
For $W\in\Gr$ with  order subset $S$ let
 $\{v_j=\sum_{i\geq s_j}v_{j,i}z^i\}_{j\geq 0}$ be a special basis of depth $n$.

\vsk.2>

Introduce the $n+1$-dimensional complex vector space $V_{W,n}$
of polynomials  in $x,t$ as the space
spanned by the polynomials $f_{j,n}(x,t), j=0,\dots,n$, where
\bean
\label{fP}
 f_{j,n}(x,t)= \sum_{i=0}^{n-s_j} v_{j,n-i}\,\chi_i(x,t)\,,
 \qquad
 j=0,\dots,n\,.
 \eean
We have $\deg_x f_j(x,t) = n-s_j$.
\vsk.2>

It is clear that the space $V_{W,n}$ does not depend on the choice of a
special basis of $W$ with depth $n$.

\vsk.2>

The same basis of depth $n$ is also a basis of depth $n+1$.
Then the space $V_{W,n+1}$ is spanned by the polynomials
\bean
\label{fP+1}
 f_{j,n+1}(x,t)
&=&
 \sum_{i=0}^{n-s_j} v_{j,n-i}\,\chi_{i+1}(x,t)\,,
 \qquad
 j=0,\dots,n,
\\
\notag
 f_{n+1,n+1}(x,t) &=& \chi_0(x,t)\,.
 \eean
Therefore, the $n+2$-dimensional space  $V_{W,n+1}$ consists of all linear combinations
$g(x,t)$ of polynomials $\chi_i(x,t)$ such that $\Delta g(x,t) \in V_{W,n}$.
\vsk.2>

The space  $V_{W,n+2}$ is related to the space  $V_{W,n+1}$ in a similar way, and so on.
Thus, to a space $W\in\Gr$ we assigned a sequence of spaces $V_{W,n}$, $V_{W,n+1}$, \dots
related by formulas \Ref{fP} and \Ref{fP+1}.

\vsk.2>
The construction in the opposite direction goes as follows.
Let $S$ be a set of virtual cardinal zero. Let $n$ be such that
$s_j=j$ for all $j>n$.  Let $V$ be an
$n+1$-dimensional complex vector
space spanned by linear combinations of polynomials
$\chi_i(x,t)$, such that $V$ has a basis
 $(f_j(x,t))_{j=0}^n$ with $\deg_x f_j(x,t)=n-s_j$.
To this vector space $V$ with such a basis
\bean
\label{fPn}
 f_{j,n}(x,t)= \sum_{i=0}^{n-s_j} v_{j,n-i}\,\chi_i(x,t)\,,
 \qquad
 j=0,\dots,n,
 \eean
we assign $W_V\in \Gr$ with  special  basis $\{v_j\}_{j\geq 0}$ of depth $n$,
where
\bean
\label{bAs}
v_j=\sum_i v_{j,i}z^i\,, \qquad \text{for}\quad j=0,\dots,n\,,
\eean
and $v_j=z^j$ for all $j>n$.  The set $S$ is the order subset of $W_V$.
We also have $V_{W_V, n}=V$.

\vsk.2>
For $W\in\Gr$ with order subset
$S=\{s_0 <s_1<\dots\}$ of depth $n$,
we have $W=W_{V_{W,n}}$.

\subsection{Tau and Baker-Akhieser functions}
\label{sec taU}
Let  $W \in \Gr$ have order subset $S=\{s_0<s_1<\dots \}$ of depth $n$. Let
  $\{v_j=\sum_{i\geq s_j}v_{j,i}z^i\}_{j\geq 0}$
be a special basis of $W$ of depth $n$.
Consider the polynomials  $(f_j(x,t))_{j=0}^n$  defined in \Ref{fP}.
\vsk.2>

Define the {\it tau-function} of $W$ by the formula
\bean
\label{eqn tau-definition}
\phantom{aaa}
\tau_{W}(x,t) = \Wh (f_0(x,t),\dots,f_n(x,t)),
\eean
cf. \cite{SW}.
The tau-function is independent of the choice of $n$
 up to multiplication by a nonzero number, see Lemma \ref{lem stab}.

\vsk.2>

Let  the  order subset $S=\{s_0<s_1<\dots\}$ corresponds to
a partition $\la$.
Then
\bean
\label{deg tau a}
\tau_W(x,t)\, =\, {a}\, x^{|\la|}\  +\ (\text{low order terms in }\, x ),
\eean
where $a$ is a nonzero number independent of $x,t$.

\vsk.3>

Define the {\it Baker-Akhieser function} of $W$ by the formula
\bean
\label{def BAW}
\psi_W^{(n)}(x,t,z)=\Om(x,t,z)\, \frac {\det \widehat M_W^{(n)}(x,t,z)}{\tau_W(x,t)},
\eean
where the matrix $\widehat M_W^{(n)}(x,t,z)$ is defined as follows.

First we define an $(n+1)\times (n+1)$-matrix
$M_W^{(n)}(x,t)$
by the formula
\bean
\label{M_W}
M_{W, k,\ell}^{(n)}(x,t)
=
\Delta^{(\ell)}f_k(x,t), \qquad k,\ell = 0,\dots,n,
\eean
cf. \Ref{M}.
Define an
$(n+2)\times(n+2)$-matrix $\widehat M_W^{(n)}(x,t,z)$,
 whose rows and columns are
labeled by indices $0,\dots,n+1$ and entries are given by the formulas:
\bean
\label{MhW}
\widehat M^{(n)}_{W, k,\ell}\,
&=&
\phantom{aaaa}
 M^{(n)}_{W, k,\ell}, \qquad k,\ell = 0,\dots, n,
\\
\notag
\widehat M^{(n)}_{W, n+1,\ell}
&=&
\phantom{aaaaa}
z^{\ell}, \qquad
\phantom{aaa}
\ell=0,\dots, n+1,
\\
\notag
\widehat M^{(n)}_{W, \ell, n+1}
&=&  \Delta^{(n+1)}  f_k(x,t),  \quad  \ell=0\ldots, n,
\eean
cf. formula \Ref{Mhat}.

\begin{lem}
\label{lem BAt}
${}$

\begin{enumerate}
\item[(i)]
Let
 $\{v_j=\sum_{i\geq s_j}v_{j,i}z^i\}_{j\geq 0}$ be a special basis
 of $W$ of depth $n$.
Then the Baker-Akhieser
function
$\psi_W^{(n)}(x,t,z)$ does not depend on the choice of the special basis.

\item[(ii)]
If  another number $n\rq{}$ is chosen such that $s_j=j$ for all $j>n\rq{}$, then
\bean
\label{BA in}
\psi_W^{(n\rq{})}(x,t,z) = z^{n\rq{}-n}\,\psi_W^{(n)}(x,t,z)\,.
\eean

\end{enumerate}
\end{lem}

\begin{proof}
The lemma follows from  Lemma \ref{lem stab}.
\end{proof}

\subsection{ mKdV tuples of subspaces}
\label{PromKdV}

Fix an integer $N>1$. We say that a subspace $W\in\Gr$ is a {\it KdV subspace} if
$z^NW \subset W$.

\vsk.2>

For example, for any $N$ the subspace $H_+$
 is a KdV subspace.

\begin{lem}
[\cite{VWr}]
\label{lem order KdV}

Let $W$ be a  KdV subspace with order subset $S$. Then $S$ is a KdV subset.

\end{lem}

 We say that an $N$-tuple $ W = (W_1,\dots,W_N)$ of KdV subspaces is an {\it mKdV tuple of subspaces}
 if $zW_i\subset W_{i+1}$ for all $i$, in particular, $zW_N\subset W_1$.
 Denote by $\GM$ the set of all mKdV tuples of subspaces.

\vsk.2>
For example, for any $N$ the tuple $ W^\emptyset = (H_+,\dots,H_+)$
is an mKdV tuple.

\vsk.2>
If  $ W = (W_1,\dots,W_N)\in \GM$, then  $(W_i,W_{i+1},\dots,W_N,$ $ W_1,W_2,\dots,W_{i-1})\in \GM$
for any $i$.

\vsk.2>

Let  $ W = (W_1,\dots,W_N)\in\GM$. Let  $S_i$ be the order subset of $W_i$ and
  $ S = (S_1,\dots,S_N)$. Then $ S$ is an mKdV tuple of subsets.
\vsk.2>

Let $W$ be a KdV subspace with order subset $S$. Let $A=\{a_1<\dots<a_N\}$ be the leading term of $S$.
 Let $v=(v_1,\dots,v_N)$ be a tuple of elements of $ W$  such that
$\ord v_i=a_i$ for all $i$.
Let $\sigma \in \Si_N$. Define an $N$-tuple $ W_{W,v,\sigma} = (W_1,\dots,W_N)$
of subspaces by the formula
\bean
\label{mkdv tuple sp}
W_i = \langle z^{i-N}v_{\sigma(1)}, z^{i-N}v_{\sigma(2)}, \dots, z^{i-N}v_{\sigma(i)}\rangle + z^iW,
\eean
in particular, $W_N = z^NW \ + $\ \em span\em$\langle v_1,\dots,v_N\rangle = W$.

\begin{thm}
[\cite{VWr}]
\label{thm m tau}
The $N$-tuple $ W_{W,v,\sigma}$ is an mKdV tuple of subspaces.
Moreover, every  mKdV tuple of subspaces is of the form  $ W_{W,v,\sigma}$ for suitable $W,v,\sigma$.

\end{thm}

Here is another description of mKdV tuples of subspaces.

 \begin{thm}
 [\cite{VWr}]
\label{thm mkdv spaces}
Let $W$ be a KdV subspace.
 Let $z^NW=V_0\subset V_1\subset V_2\subset \dots\subset V_{N-1}\subset V_N=W$
 be a complete flag of vector subspaces such that $\dim V_i/V_{i-1}=1$ for all $i$.
Set
\bean
\label{MKDV tuples}
W_i = z^{i-N}V_i,\qquad i=1,\dots,N.
\eean
Then  $ W = (W_1,\dots,W_{N-1}, W_N=W)$ is an mKdV tuple of subspaces.
Moreover, every
 mKdV tuple of subspaces is of this form.

\end{thm}

Let $W$ be a KdV subspace.
It follows from Theorem \ref{thm mkdv spaces} that the set of mKdV tuples of subspaces with  prescribed last term $W_N=W$
is identified with the set of complete flags in the $N$-dimensional complex vector space  $W/z^NW$.

\subsection{Relations between Baker-Akhieser functions}

Let  $(W_1,\dots,W_N)\in \GM$. Let $(\tau_{W_1}(x,t),\dots,\tau_{W_N}(x,t))$
and $(\psi_{W_1}(x,t,z),\dots,\psi_{W_N}(x,t,z))$ be the corresponding tau and Baker-Akhieser functions.

\vsk.2>

Recall that each Baker-Akhieser function $\tau_{W_i}(x,t)$ is defined up to multiplication by a monomial
$z^m$, see Lemma \ref{lem BAt}.
 A Baker-Akhieser  function with a choice of this factor will be called a {\it graded Baker-Akhieser function} of $W_i$.

\begin{thm}
\label{thm tgrad}
 There exist graded Baker-Akhieser  functions $\psi_{W_1}(x,t,z)$,
\dots, $\psi_{W_N}(x,t,z)$
such that
\bean
\label{rel the}
\phantom{aaa}
\psi_{W_{i-1}}(x,t,z)
&=&
\psi_{W_i}(x+1,t,z)-  \frac{\tau_{W_{i}}(x,t)\,y_{W_{i-1}}(x+1,t)}
{y_{W_{i}}(x+1,t)\,y_{W_{i-1}}(x,t)} \,\psi_{W_i}(x,t,z),
\eean
for $i=2,\dots,N$, and
\bean
\label{rel theN}
z^{N}\psi_{W_{N}}(x,t,z)
&=&
\psi_{W_1}(x+1,t,z)-  \frac{\tau_{W_{1}}(x,t)\,y_{W_{N}}(x+1,t)}
{y_{W_{1}}(x+1,t)\,y_{W_{N}}(x,t)} \,\psi_{W_1}(x,t,z).
\eean

\end{thm}

\vsk.2>
Denote $y_n(x,t):= \tau_{W_{N-n+1}}(x,t)$, $n=1,\dots,N$, and extend this sequence
by periodicity, $y_{N+n}(x,t)=y_n(x,t)$ for all values of $n\in\Z$.
Denote $\psi_n(x,t,z):= \psi_{W_{N-n+1}}(x,t,z)$, $n=1,\dots,N$, and extend this sequence
by periodicity, $\psi_{N+n}(x,t,z)=z^N\psi_n(x,t,z)$ for all values of $n\in\Z$.
Introduce the sequence $(v_n(x,t))_{n\in\Z}$ by formula
\bean
\label{vpott}
v_n(x,t)=\frac{y_{n}(x,t)\,y_{n+1}(x+1,t)}{y_n(x+1,t) \,y_{n+1}(x,t)},
\eean
see \Ref{vpot}.

\begin{cor}
\label{cor W-peri} For any fixed $t$, the
 functions $(v_n(x,t))_{n\in\Z}$ and $(\psi_n(x,t,z))_{n\in\Z}$
satisfy relations \Ref{laxdd}.
\qed
\end{cor}

\vsk.3>
\noindent
{\it Proof of Theorem \ref{thm tgrad}.}\
Since  the tuple  $(W_2,W_{3},\dots,W_N,$ $ W_1)$ is also an mKdV tuple,
it is enough to prove \Ref{rel the} for $i=N$ only.

\vsk.2>
By Theorem \ref{thm m tau} the pair $W_{N-1}$, $W_N$ has the following
form. Let $S=\{s_0<s_1<\dots\}$ be the order subset of $W_N$.
Let $A=\{a_1<\dots <a_N\}$ be the leading term of $S$.
Choose one element  $a \in A$.

Let $S$ be of depth $n$.
Let  $\{v_j=\sum_{i\geq s_j}v_{j,i}z^i\}_{j\geq 0}$
be a special basis of $W$ of depth $n$. Let $w=\sum_i w_i z^i$ be the element of the basis
with $\on{ord} w = a$. Then $W_{N-1}$ is the space with basis
$\{z^{1-N}w\}\cup \{z v_j\}_{j\geq 0}$. This basis of $W_{N-1}$ is a basis of depth $n+1$.

\vsk.2>
The tau and Baker-Akhieser functions of $W_N$ are defined in terms of the basis
$\{v_j\}_{j\geq 0}$  of depth $n$ by polynomials $f_{j,n}(x,t)$, $j=0,\dots,n$,
see formula \Ref{fP}.

The tau and Baker-Akhieser functions of $W_{N-1}$ are defined in terms of its basis
$\{z^{1-N}w\}\cup \{z v_j\}_{j\geq 0}$  of depth $n+1$
by the same polynomials $f_{j,n}(x,t)$, $j=0,\dots,n$, and one additional polynomial $f_{n+1}(x,t)$ corresponding to the
basis element $z^{1-N}w$,
\bea
f_{n+1}(x,t) = \sum_i w_{N+n-i} \chi_i(x,t).
\eea
Now the functions $\tau_{W_{N-1}}, \tau_{W_{N-1}}$, $\psi_{W_{N-1}}, \psi_{W_{N-1}}$ satisfy \Ref{rel the} for $i=N$
by Theorem \ref{thm BAk}.
\qed

\subsection{Generation of new  mKdV tuples of subspaces}
\label{sec GmKdV}

Let  $ W = (W_1,\dots,W_N)\in \GM$.  By Theorem \ref{thm mkdv spaces}, the tuple $ W$
 is determined by a flag
\bea
z^NW_N=V_0\subset V_1\subset V_2\subset \dots\subset V_{N-1}\subset W_N.
\eea
The quotient $V_2/V_0$ is two-dimensional. Any line $\tilde V_1/V_0$ in $V_2/V_0$ determines a flag
$z^NW_N=V_0\subset \tilde V_1\subset V_2\subset \dots\subset V_{N-1}\subset W_N$,
 which in its turn
determines an mKdV tuple $ W^{(1)}=(\tilde W_1,W_2,\dots,W_N)$
with $\tilde W_1 = z^{1-N}\tilde V_1$.
Thus we get a family of mKdV tuples of subspaces parameterized by points
of  the projective line $P(V_2/V_0)$. The new tuples
are parametrized by points of the affine line $\A=P(V_2/V_0)-\{V_1/V_0\}$.
We get a map $X^{(1)}:\A\to \GM$ which sends   $a\in\A$ to the corresponding mKdV tuple
$ W^{(1)}(a)=(\tilde W_1(a),W_2,\dots,W_N)$.
This  map  will be called
the {\it generation of mKdV tuples from the tuple $ W$ in the first direction}.

\vsk.2>

Similarly, for any $i=2,\dots,N$, we construct a map $X^{(i)}:\A\to \GM$,
where $\A=P(V_{i+1}/V_{i-1}) - \{V_i/V_{i-1}\}$ which sends $a\in \A$
to the corresponding mKdV tuple
$ W^{(i)}(a)=(W_1,\dots,\tilde W_i(a),\dots,W_N)$.
This  map  will be called
the {\it generation of mKdV tuples of subspaces from the tuple $ W$ in the $i$-th direction.}

\vsk.2>
We say that the generation in the $i$-th direction is {\it degree increasing} if
for any $a\in \A$, we have $\deg_{x} \tau_{ W^{(i)}(a)}(x,t) > \deg_{x} \tau_{ W}(x,t)$.

\medskip
The tau-function $\tau_{\tilde W_i(a)}$ depends on $a$ linearly in the following sense.
Let $\{v_i\}_{i\geq 1}$ be a basis of $V_{i-1}$. Let $v_0\in V_i$ be such that $\{v_i\}_{i\geq 0}$
is a basis of $V_i$. Let $\tilde v_{0}\in V_{i+1}$ be such that $\{\tilde v_{0}, v_0,v_1,v_2,\dots\}$ is a basis
of $V_{i+1}$. Then the points of $ \A=P(V_{i+1}/V_{i-1}) - \{V_i/V_{i-1}\}$ are parametrized by complex numbers $c$.
A  number $c$ corresponds to the line generated by the subspace $\tilde V_i(c)$ with basis $\{\tilde v_{0} + cv_0,
v_1,v_2\dots\}$. This $c$ is an affine coordinate on $\A$.  Calculating the tau-function of the  subspace
$\tilde W_i(c)=z^{i-N}\tilde V_i(c)$ with respect to the basis $\{z^{i-N}(\tilde v_{0} + cv_0),
z^{i-N}v_1, z^{i-N}v_2\dots\}$ we get the formula
\bean
\label{tau linear}
\tau_{\tilde W_i(c)}  = \tau_{\tilde W_i(0)} + c \tau_{W_i}.
\eean

\begin{thm}
\label{thm taU Wr}
For the generation in the $i$-th direction,
 the tau-functions of the subspaces $\tilde W_i(c), W_i,W_{i-1},W_{i+1}$ satisfy the
 equation
\bean
\label{wR tAu}
\Wh(\tau_{W_i}(x,t), \tau_{\tilde W_i(c)}(x,t)) \,= \,\on{const}\, \tau_{W_{i-1}}(x,t)\,\tau_{W_{i+1}}(x+1,t)\,,
\eean
where $\on{const}$ is a number independent of $x,t$.
\end{thm}

\begin{proof}
The proof of this theorem is word by word the same as the proof of Theorems 6.10 and 7.10 in
\cite{VWr},
see also the proof of Theorem \ref{thm tgrad}.
\end{proof}

Define an infinite $N$-periodic sequence of polynomials $(y_n(x,t))_{n\in\Z}$ by the formula
\bean
\label{reverse}
y_n(x,t): = \tau_{W_{-n}}(x,t)\,.
\eean

\begin{cor}
\label{cor fertiL}
For any mKdV tuple $ W = (W_1,\dots,W_N)$ and any fixed $t$, the sequence
$(y_n(x,t))_{n\in\Z}$
of polynomials in $x$ is fertile.
\qed
\end{cor}

\begin{rem}
Theorem \ref{thm taU Wr} says that the
generation of mKdV tuples in the $i$-th direction  from the tuple $ W$
corresponds to the generation of tuples of polynomials in the $i$-th direction from the tuple
$(\tau_{W_{1}}(x,t), \dots, \tau_{W_{N}}(x,t))$, where the latter generation procedure  is described in
Section \ref{Elg}. In other words, the  two generation procedure and  the functor,
 which assigns to a point of  $\Gr$ its  tau-function, commute.

\end{rem}

\subsection{Transitivity of the generation procedure}
\label{sec transit}

\begin{thm}
[\cite{VWr}]
\label{thm exis}
Any mKdV tuple $ W\in \GM$ can be obtained
from the mKdV tuple
$ W^\emptyset=(H_+,\dots,H_+)$ by a sequence of degree increasing
generations.
\end{thm}

Combining this theorem and Theorem \ref{one gen} we obtain the following corollary.

\begin{cor}
\label{cor BA tau}
If a tuple $(y_1(x),\dots,y_N(x))$ represents a solution of the Bethe ansatz equations \Ref{bae},
then there exists an mKdV tuple of subspaces $(W_1,\dots, W_N)$ such that
\bean
\label{y=tau}
(y_1(x),\dots,y_N(x)) = (\tau_{W_{1}}(x,0),\dots,\tau_{W_{N}}(x,0))\,.
\eean
In particular, the tuple  $(y_1(x),\dots,y_N(x))$ can included into the
family $(\tau_{W_{1}}(x,t)$,
\dots, $\tau_{W_{N}}(x,t))$ of tuples depending on $t$, and then
extended to the sequences
of functions $(v_n(x,t))_{n\in\Z}$ and $(\psi_n(x,t,z))_{n\in\Z}$,
as explained in Corollary \ref{cor W-peri}, and those sequences
 $(v_n(x,t))_{n\in\Z}$ and $(\psi_n(x,t,z))_{n\in\Z}$
give a solution of the generating  linear problem equation \Ref{laxdd} depending on $t$
as stated in Corollary \ref{cor W-peri}.

\end{cor}

\subsection{ Commuting flows on $\Gr$}
\label{sec cfG}

For a subspace $W\in\Gr$, the subspace
\bean
\label{W(t)}
W(t):=e^{\sum_{i=1}^\infty t_iz^i}W
\eean
is a well-defined subspace in $\Gr$. Given $W$, the space $W(t)$ depends
only on finitely many of $t_1,t_2,\dots$. This construction gives us a family
of commuting flows on $\Gr$ with times $t_1, t_2,\dots$.  We will call them the
{\it discrete mKdV flows}.

\vsk.2>

The discrete mKdV flows on $\Gr$ induce a family of commuting flows on the space of $N$-tuples
$(\tau_{W_1}(x,0),\dots,y_{W_N}(x,0))$, representing solutions of the Bethe ansatz equations \Ref{bae}.
The construction goes as follows.

\vsk.2>
Let $(W_1,\dots,W_N)\in \Gr$. Let $(\tau_{W_1}(x,t),\dots,y_{W_N}(x,t))$ be the collection of tau-functions assigned to
$(W_1,\dots,W_N)$ in Section \ref{sec taU}.  The collection of polynomials
$(\tau_{W_1}(x,0)$,
\dots, $\tau_{W_N}(x,0))$ in $x$ will be called the tuple of {\it reduced tau-functions} of
$(W_1,\dots,W_N)$. When the tuple $(W_1,\dots,W_N)$ becomes dependent on $
t$ we obtain a family of tuples of reduced
tau-functions $(\tau_{W_1(t)}(x,0),\dots,\tau_{W_N(t)}(x,0))$.
Thus we obtain a family of commuting flows on the space of tuples of reduced
tau-functions, which will also be called the {\it discrete mKdV flows}.

\begin{lem}
\label{lem fltau}  For any $(W_1,\dots,W_N)\in\Gr$ we have
\bean
\label{flow tau}
(\tau_{W_1(t)}(x,0),\dots,\tau_{W_N(t)}(x,0))
= (\tau_{W_1}(x,t),\dots,\tau_{W_N}(x,t))\,.
\eean
\qed

\end{lem}

\end{document}